\documentclass[11pt,hidelinks,colorlinks=true]{article}

\pdfoutput=1
\usepackage{latexsym}
\usepackage{amsmath}
\usepackage{amssymb}
\usepackage{amsfonts,amsthm}
\usepackage{float}
\usepackage[top=1in, bottom=1in, left=1in, right=1in]{geometry}
\usepackage[lined, algoruled]{algorithm2e}

\usepackage{tikz}
\usetikzlibrary{arrows,decorations.pathmorphing,decorations.shapes,backgrounds,fit}
\pgfkeys{tikz/cc1/.style={dotted} }
\pgfkeys{tikz/cc4/.style={dashed} }
\pgfkeys{tikz/cc3/.style={decorate,decoration=bumps} }
\pgfkeys{tikz/cc2/.style={decorate,decoration=snake} }
\pgfkeys{tikz/solid/.style={line width=2pt} }
\pgfkeys{tikz/solid2/.style={line width=1pt} }
\pgfkeys{tikz/nodostile1/.style={draw,line width=1pt}}
\pgfkeys{tikz/nodostile2/.style={draw,line width=2pt}}
\pgfkeys{tikz/arcostile1/.style={very thick} }
\pgfkeys{tikz/arcostile2/.style={red,very thick} }
\pgfkeys{tikz/arcostile3/.style={thick} }
\pgfkeys{tikz/arcostile4/.style={line width=2pt} }

\newcommand{\ignore}[1]{}

\newcommand{\csbdom}{\stackrel{D}{\rightarrow}}
\newcommand{\csbrdom}{\stackrel{D^R}{\rightarrow}}

\newtheorem{theorem}{Theorem}[section]
\newtheorem{lemma}[theorem]{Lemma}

\newtheorem{corollary}[theorem]{Corollary}

\newtheorem{property}[theorem]{Property}

\usepackage{array}
\newcolumntype{C}[1]{>{\centering\let\newline\\\arraybackslash\hspace{0pt}}m{#1}}

\newcommand{\comment}[1]{}

\usepackage{hyperref,xcolor}
\definecolor{winered}{rgb}{0.5,0,0}

\hypersetup
{
    pdfborder={0 0 0},
colorlinks=true,
linkcolor={winered},
urlcolor={winered},
filecolor={winered},
citecolor={winered},
  linktoc=all,
}

\begin{document}
\title{\bf Strong Connectivity in Directed Graphs\\ under Failures, with Applications\thanks{An extended abstract of this work appeared in the {\em Proc.~28th ACM-SIAM Symposium on Discrete Algorithms,} pp.~1880--1899, 2017}}
\author{Loukas Georgiadis$^{1}$ \and Giuseppe F. Italiano$^{2}$ \and Nikos Parotsidis$^{3}$}

\date{\today}

\maketitle

\begin{abstract}
In this paper, we investigate some basic connectivity problems in directed graphs (digraphs). Let $G$ be a digraph with $m$ edges and $n$ vertices, and let $G\setminus e$ (resp., $G\setminus v$) be the digraph obtained after deleting edge $e$ (resp., vertex $v$) from $G$. As a first result, we show how to compute in $O(m+n)$ worst-case time:
\begin{itemize}
\vspace{-.15cm}
\item The total number of strongly connected components in $G\setminus e$ (resp., $G\setminus v$), \emph{for all edges $e$} (resp., \emph{for all vertices $v$}) in $G$.
\vspace{-.1cm}
\item The size of the largest and of the smallest strongly connected components in $G\setminus e$ (resp., $G\setminus v$), \emph{for all edges $e$} (resp., \emph{for all vertices $v$}) in $G$.
\end{itemize}
\vspace{-.1cm}
\noindent
Let $G$ be strongly connected. We say that edge $e$ (resp., vertex $v$) separates two vertices $x$ and $y$, if $x$ and $y$ are no longer strongly connected in $G\setminus e$ (resp., $G\setminus v$). As a second set of results,  we show how to build in $O(m+n)$ time $O(n)$-space data structures that can answer in optimal time the following basic connectivity queries on digraphs:
\vspace{-.15cm}
\begin{itemize}
\item Report in $O(n)$ worst-case time all the strongly connected components of $G\setminus e$ (resp., $G\setminus v$), for a query edge $e$ (resp., vertex $v$).
\vspace{-.1cm}
\item Test whether an edge or a vertex separates two query vertices in $O(1)$ worst-case time.
\vspace{-.1cm}
\item
Report all edges (resp., vertices) that separate two query vertices in optimal worst-case time, i.e., in time $O(k)$, where $k$ is the number of separating edges (resp., separating vertices). (For $k=0$, the time is $O(1)$).
\end{itemize}
\vspace{-.1cm}
All our bounds are tight and are obtained with a common algorithmic framework, based on a novel compact representation of the decompositions induced by the $1$-connectivity (i.e., $1$-edge and $1$-vertex) cuts in digraphs, which might be of independent interest. With the help of our data structures we can design  efficient algorithms for several other connectivity problems on digraphs and we can also obtain in linear time a strongly connected spanning subgraph of $G$ with $O(n)$ edges that maintains the $1$-connectivity  cuts of $G$ and the decompositions induced by those cuts.
\end{abstract}

\footnotetext[1]{Department of Computer Science \& Engineering, University of Ioannina, Greece. E-mail: \texttt{loukas@cs.uoi.gr}.}
\footnotetext[2]{LUISS University, Rome, Italy. E-mail: \texttt{gitaliano@luiss.it}.
}
\footnotetext[3]{University of Copenhagen, Denmark. E-mail: \texttt{nipa@di.ku.dk}.
The author is supported by Grant  Number 16582, Basic Algorithms Research Copenhagen (BARC), from the VILLUM Foundation.
}

\clearpage

\tableofcontents

\clearpage

\section{Introduction}
\label{sec:introduction}

In this paper, we investigate some basic connectivity problems in directed graphs (digraphs).
Before defining precisely the problems considered, we need few definitions.
Let $G=(V,E)$ be a directed graph (digraph), with $m$ edges and $n$ vertices. Digraph $G$ is \emph{strongly connected} if there is a directed path from each vertex to every other vertex.
The \emph{strongly connected components} of $G$ are its maximal strongly connected subgraphs. Two vertices $u,v \in V$  are \emph{strongly connected}  if they belong to the same strongly connected component of $G$. The \emph{size} of a strongly connected component is given by its number of vertices.
An edge (resp., a vertex) of $G$ is a \emph{strong bridge} (resp., a \emph{strong articulation point}) if its removal increases the number of strongly connected components. Note that strong bridges (resp., strong articulation points) are $1$-edge (resp., $1$-vertex) cuts for digraphs.
Let $G$ be strongly connected. $G$ is \emph{$2$-edge-connected} if it has no strong bridges, and it is
\emph{$2$-vertex-connected} if it has at least three vertices and no strong articulation points.
Let $C \subseteq V$.
The induced subgraph of $C$, denoted by $G[C]$, is the subgraph of $G$ with vertex set $C$ and edge set $E \cap (C \times C)$.
If $G[C]$ is $2$-edge-connected (resp., $2$-vertex-connected), and there is no set of vertices $C'$ with $C \subsetneq C' \subseteq V$ such that $G[C']$ is also
$2$-edge-connected (resp., $2$-vertex-connected), then $G[C]$ is a \emph{maximal $2$-edge-connected (resp., $2$-vertex-connected) subgraph of $G$}.
Hence, in the context of reliable communication, maximal
$2$-edge- and $2$-vertex-connected subgraphs correspond, respectively, to parts of a network that are resilient to single vertex and edge failures.
These concepts, however, do not capture the pairwise connectivity among the vertices.
%
Two vertices $u, v\in V$ are said to be \emph{$2$-edge-connected} (resp., \emph{$2$-vertex-connected}), and we denote this relation by  $u \leftrightarrow_{\mathrm{2e}} v$ (resp., $u \leftrightarrow_{\mathrm{2v}} v$), if there are two edge-disjoint  (resp., two internally vertex-disjoint) directed paths from $u$ to $v$  and two edge-disjoint  (resp., two internally vertex-disjoint) directed paths from $v$ to $u$ (note that a path from $u$ to $v$ and a path from $v$ to $u$ need not be edge- or vertex-disjoint). A \emph{$2$-edge-connected component} (resp., \emph{$2$-vertex-connected component}) of a digraph $G=(V,E)$ is defined as a maximal subset $B \subseteq V$ such that $u \leftrightarrow_{\mathrm{2e}} v$ (resp., $u \leftrightarrow_{\mathrm{2v}} v$) for all $u, v \in B$.
See Figure \ref{figure:2-connectivity-example}.
Let $G\setminus e$ (resp., $G\setminus v$) denote the digraph obtained after deleting edge $e$ (resp., vertex $v$ together with all its incident edges).
We say that edge $e$ (resp., vertex $v$) \emph{separates} vertices $x$ and $y$, if $x$ and $y$ are no longer strongly connected in $G\setminus e$ (resp., $G\setminus v$).

Connectivity-related problems for digraphs are notoriously harder than for undirected graphs, and indeed many notions for undirected connectivity do not translate to the
directed case.
As shown in~\cite{CHILP:2CS,biblocks:ES1980,2ECC:GILP:TALG,2VCB,2CC:HenzingerKL:ICALP15}, in digraphs $2$-vertex and $2$-edge connectivity have a much richer and more complicated structure than in undirected graphs.
For instance, in the case of undirected graphs the  $2$-edge- (resp., $2$-vertex-)
 connected components are identical to the maximal $2$-edge- (resp., \mbox{$2$-vertex-)} connected subgraphs. This is not the case for digraphs, however, where components can be different from components: namely, two vertices may be $2$-edge- (resp., $2$-vertex-) connected without being necessarily in the same maximal $2$-edge- (resp., $2$-vertex-) connected subgraph~\cite{2ECC:GILP:TALG,2VCB}.
Moreover, an undirected graph is naturally
decomposed by bridges (resp., articulation points) into a tree of $2$-edge- (resp., $2$-vertex-) connected components, known as the bridge-block (resp., block) tree (see, e.g.,~\cite{onlineBiconnected:WT92}). In
 digraphs, the decomposition induced by strong bridges (resp., strong articulation points) becomes much more complicated (see Figure \ref{figure:grid-example}):
 in general, it was shown by Bencz\'ur that in digraphs there can be no ``cut'' tree for various connectivity concepts~\cite{Benczur95}.
Hence, it is not surprising that $2$-connectivity problems on directed graphs
are much harder than on undirected graphs.
For undirected graphs,
it has been known for over 40 years how to compute the analogous notions (bridges, articulation points, $2$-edge- and $2$-vertex-connected components) in linear time,
by simply using depth first search~\cite{dfs:t}. In the case of digraphs, however, the same problems revealed to be much more challenging: although these problems have been investigated for quite a long time (see, e.g.,~\cite{biblocks:ES1980,makino,nagamochi}), obtaining fast algorithms for $2$-edge and $2$-vertex connectivity for digraphs has been an elusive goal for many years.
Indeed, it has been shown only recently that all strong bridges and strong articulation points of a digraph can be computed in linear time~\cite{Italiano2012}.
Additionally, it was shown very recently how to compute the $2$-edge- and $2$-vertex-connected components of digraphs in linear time \cite{2ECC:GILP:TALG,2VCB}, while
the best current bound for computing the $2$-edge- and the $2$-vertex-connected components in digraphs is not even linear, but it is
$O(\min\{m^{3/2}, n^2\})$~\cite{CHILP:2CS,2CC:HenzingerKL:ICALP15}.

\begin{figure}[t]
	\begin{center}
		\includegraphics[trim={1.5cm 5cm 1.5cm 1.5cm}, clip=true, width=1\textwidth]{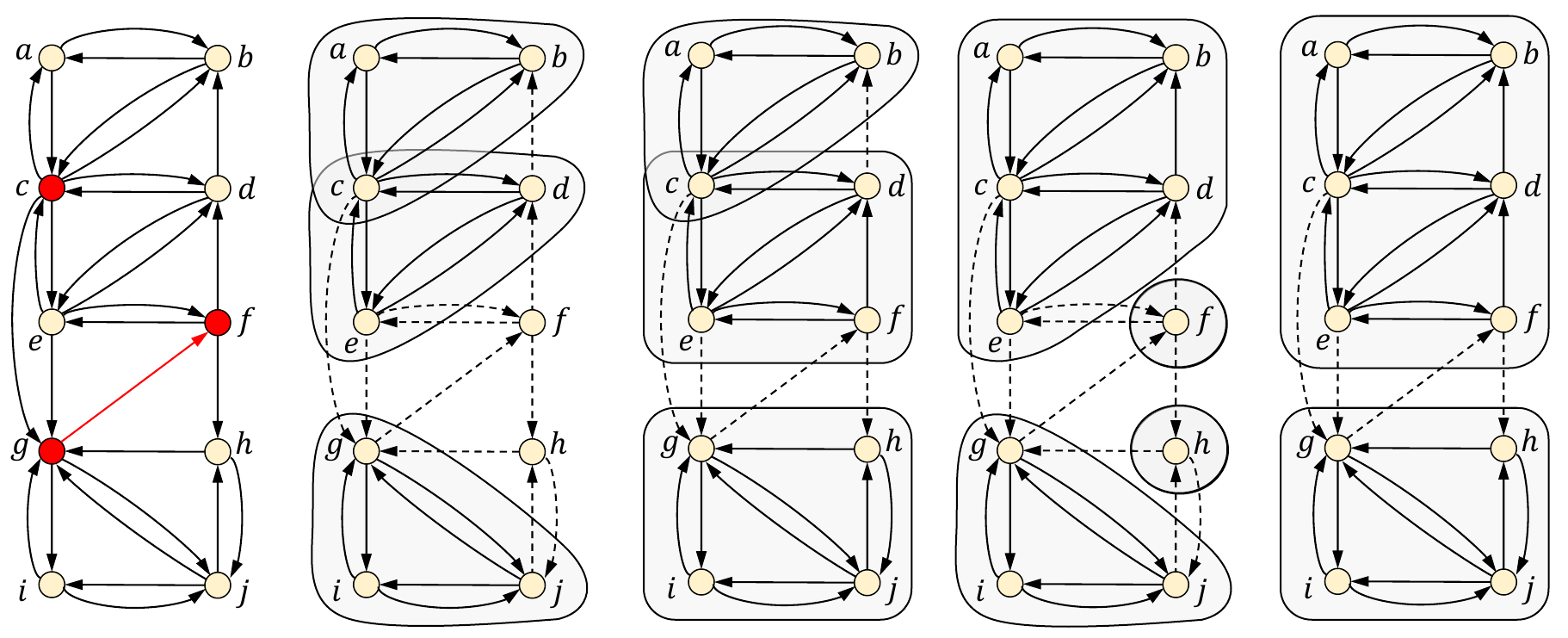}
		\begin{tabular}{C{0.0\textwidth}C{0.1\textwidth}C{0.23\textwidth}C{0.15\textwidth}C{0.20\textwidth}C{0.15\textwidth}C{0.20\textwidth}}
			\\[-.6cm]
			& (a) $G$ &  (b) $Max2\mathit{VCS}(G)$  & (c) $2\mathit{VCC}(G)$ &  (d) $Max2\mathit{ECS}(G)$ &  (e) $2\mathit{ECC}(G)$ &
		\end{tabular}
	\end{center}
	\caption{(a) A strongly connected digraph $G$, with strong articulation points and strong bridges shown in red (better viewed in color). (b) The maximal $2$-vertex-connected subgraphs of $G$. (c) The $2$-vertex-connected components of $G$. (d) The maximal $2$-edge-connected subgraphs of $G$. (e) The $2$-edge-connected components of $G$. Note that vertices $e$ and $f$ are in the same $2$-vertex- (resp., $2$-edge-) connected  component of $G$ since there are two internally vertex-disjoint (resp., edge-disjoint) paths from $e$ to $f$ and from $f$ to $e$. However, $e$ and $f$ are not in the same maximal $2$-vertex (resp., $2$-edge-) connected subgraph of $G$.
	}
	\label{figure:2-connectivity-example}
\end{figure}

\begin{figure}[t]
	\begin{center}
		\includegraphics[trim={1.5cm 10cm 1cm 1.5cm}, clip=true, width=1\textwidth]{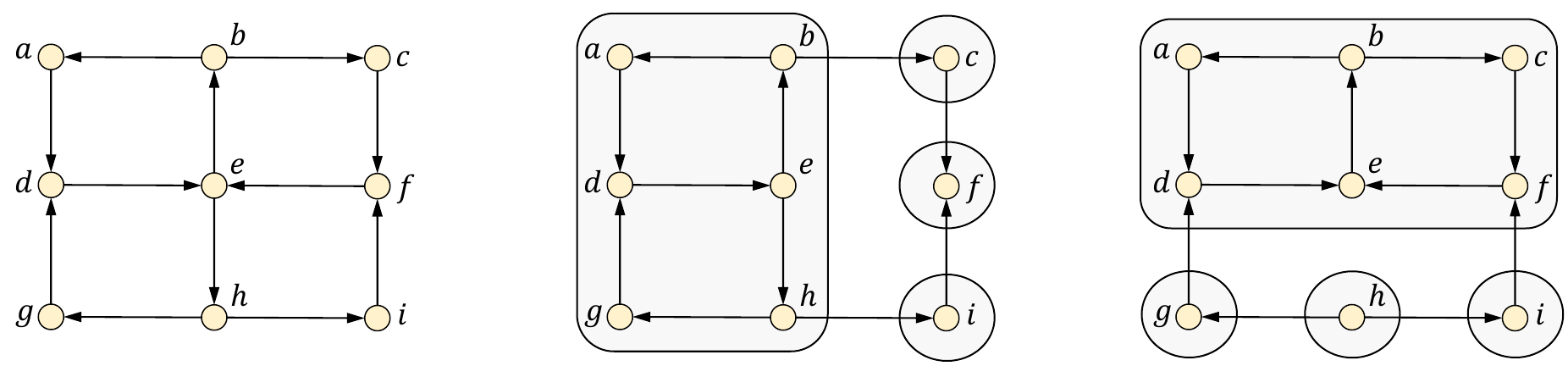}
		\begin{tabular}{C{0.25\textwidth}C{0.43\textwidth}C{0.23\textwidth}}
			(a) $G$ &  (b) $G \setminus (f,e)$  & (c) $G \setminus (e,h)$
		\end{tabular}
	\end{center}
	\caption{An example illustrating the complicated structure of $1$-edge cuts in digraphs.
		(a) A strongly connected digraph $G$.
		(b) The strongly connected components in $G\setminus (f,e)$.
		(c) The strongly connected components in $G\setminus (e,h)$.
		Note that a strongly connected component in $G\setminus (f,e)$ and a strongly connected component in $G\setminus (e,h)$ are neither disjoint nor nested.
		In fact, all edges are strong bridges, and the deletion of each edge creates many non-disjoint and non-nested sets in the resulting partitions.
	}
	\label{figure:grid-example}
\end{figure}

In this paper, we are interested in computing efficiently some properties of $G\setminus e$ and of $G\setminus v$,
for all possible edges $e$ and vertices $v$,
such as the number, or the largest or the smallest size of their strongly connected components,
or in
finding an edge $e$ or a vertex $v$ whose deletion minimizes/maximizes those properties.
Those problems are not only theoretically interesting, but they also arise in a variety of application areas, including
network analysis and
computational biology.
For instance, in several critical networked infrastructures, one is interested in  identifying
vertices and edges, whose removal results in a specific
degradation of the network global pairwise connectivity \cite{dinh}.
In social networks, finding vertices whose deletion optimizes some connectivity properties is related to identifying
key players whose removal fragments / disrupts the underlying network \cite{Borgatti,ventresca}. Applications in computational biology include the computation of steady states on digraphs governed by
Laplacian dynamics,
which include the symbolic derivation of kinetic equations and steady state expressions for biochemical
systems~\cite{Gunawardena12,MihalakUY15}. In particular, Mihal\'ak et al.~\cite{MihalakUY15} presented recently a recursive deletion-contraction
algorithm for such applications, whose efficient implementation  needs to find repeatedly the edge of a strongly connected digraph
whose deletion maximizes quantities such as the number of resulting strongly connected components or minimizes their maximum size.

Before stating our new bounds, we review some simple-minded solutions to the problems considered. 
A trivial approach to find
an edge whose deletion minimizes/maximizes the number, or the largest or the smallest size of the resulting strongly connected components,
would be to recompute the strongly connected components of $G\setminus e$, for each edge $e$ in $G$, which requires $O(m^2)$ time in the worst case. This trivial bound can be easily improved by observing that if an edge $e$ is not a strong bridge then, by definition, $G \setminus e$ remains strongly connected.
Hence, we can first compute all strong bridges of $G$ in $O(m+n)$ time~\cite{Italiano2012} and then
consider only the case where the edge to be deleted is a strong bridge,
by recomputing
the strongly connected components of $G\setminus e$ for each strong bridge $e$. Let $b$ the total number of strong bridges in $G$: this yields a total time of $O(mb)$, which is $O(mn)$ in the worst case, since there can be $O(n)$ strong bridges in a digraph $G$. Similar bounds apply in the case of vertex deletions: we can find
a vertex whose deletion minimizes/maximizes the number, or the largest or the smallest size of the resulting strongly connected components in $O(mp)$ time, where $p$ is the total number of strong articulation points in $G$. Since $p\leq n$, this can be $O(mn)$ in the worst case.

\paragraph{Our results.}
In this paper, we present new algorithms and data structures for computing in $O(m+n)$ worst-case time:
\begin{itemize}
\item The total number of strongly connected components in $G\setminus e$ (resp., $G\setminus v$),  \emph{for all edges $e$} (resp., \emph{for all vertices $v$}) in $G$,  thus improving the trivial bound of $O(mn)$. Our bound is asymptotically tight.
\item The size of the largest and of the smallest strongly connected components
in $G\setminus e$ (resp., $G\setminus v$), \emph{for all edges $e$} (resp., \emph{for all vertices $v$}) in $G$, thus improving the trivial bound of $O(mn)$. Our bound is again asymptotically tight.
\end{itemize}

\noindent
Note that this gives immediately an algorithm for finding in linear time an edge (resp., a vertex) whose deletion minimizes/maximizes the total number or the largest/smallest size of the resulting strongly connected components in the resulting digraph,
improving over the previous $O(mn)$ bounds.
We can also
build $O(n)$-space data structures that, after $O(m+n)$-time preprocessing, are able to answer in asymptotically optimal time the following basic $2$-edge and $2$-vertex connectivity queries on digraphs:
\begin{itemize}
\item Report in $O(n)$ worst-case time all the strongly connected components of $G\setminus e$ (resp., $G\setminus v$), for a query edge $e$ (resp., vertex $v$), improving the trivial bound of $O(m+n)$. Note that those bounds are asymptotically tight, as one needs $O(n)$ time to output the strongly connected components of a digraph.
\item Test whether an edge or a vertex separates two query vertices in $O(1)$ worst-case time, improving the trivial bound of $O(m)$.
\item Report all the edges (resp., vertices) that separate two query vertices in optimal worst-case time, i.e., in time $O(k)$, where $k$ is the number of separating edges (resp., separating vertices).
(For $k=0$, the time is $O(1)$). This
improves the trivial bound of $O(mn)$.
\end{itemize}

With our approach, we can design  efficient algorithms for several other
connectivity problems on digraphs.
After $O(m+n)$-time preprocessing, we can answer in $O(1)$ time for each edge $e$ (resp., vertex $v$) queries such as the number of strongly connected components in $G\setminus e$ (resp., $G\setminus v$), or the maximum / minimum size of a strongly connected component in $G\setminus e$ (resp., $G\setminus v$).
We can further output all the strongly connected components in $G\setminus e$, for all edges $e$ in $G$, in total $O(m+nb)$ worst-case time, and all the strongly connected components in $G\setminus v$, for all vertices $v$ in $G$, in total $O(m+np)$ worst-case time, improving over previous $O(mb)$ and $O(mp)$ bounds.
All those bounds are asymptotically tight.
Note that $O(m+n)$, $O(m+nb)$ and $O(m+np)$ are all $O(n^2)$ in the worst case, while $O(mb)$ and $O(mp)$ are $O(n^3)$. Thus, our data structures are able to improve one order of magnitude over previously known bounds.
Furthermore, our approach is able to provide alternative linear-time algorithms for computing
the $2$-edge-connected and $2$-vertex-connected components of a digraph, which are much simpler than the algorithms presented in~\cite{2ECC:GILP:TALG,2VCB}, and thus are likely to be more amenable to  practical implementations~\cite{2CC:ALENEX18}.
Finally, we show how to obtain in linear time, a
strongly connected spanning subgraph of $G$ with $O(n)$ edges that maintains: (i) the $1$-connectivity cuts
of $G$ (i.e., $1$-edge cuts given by strong bridges, and $1$-vertex cuts given by strong articulation points) and
the decompositions induced by those cuts, and (ii) the $2$-edge-connected and $2$-vertex-connected components of $G$.

Our framework provides efficient algorithms for several variations of the \emph{critical node detection}, defined as follows: We wish to find the vertex $x$ of $G$ that minimizes some connectivity function $f(G \setminus x)$ defined over the sizes of the strongly connected components of $G \setminus x$.
We refer to this vertex $x$ as the \emph{most critical node} of $G$. 
For instance, as spelled out in the survey on critical node detection problems in \cite{LALOU201892}, the function $f(G \setminus x)$ can be set to be the number of the strongly connected components of $G \setminus x$, or the size of the smallest or the largest strongly connected component, or the number of strongly connected pairs of vertices in $G \setminus x$.
Our framework provides linear time algorithms for finding the most critical node for a plethora of functions $f(G \setminus x)$, including all the aforementioned cases.
Based on our framework, \cite{Paudel:2018} obtains an alternative linear-time algorithm for the critical node detection problem, where $f(G \setminus x)$ is set to be the number of strongly connected pairs of vertices in $G \setminus x$, and shows that repeated applications of the most critical node detection algorithm gives an efficient heuristic for the more general problem where we wish to compute a set $S \subseteq V$ of at most $k$ vertices 
that minimizes $f(G \setminus S)$.
%
%
As noted in \cite{Arulselvan:2009}, the critical node detection problem has, in particular, several applications in the field of social network
analysis. Other applications of critical nodes include network immunization~\cite{Immune:PhysRevLett}, and the study of covert terrorist
networks~\cite{FM941}.

\paragraph{\bf Related work.}
The problem of preprocessing a digraph $G$ so that one can quickly answer queries under edge or vertex failures is not new. For instance, it has been investigated for reachability~\cite{KING2002150}
and shortest paths~\cite{Demetrescu:2008}. Specifically, King and Sagert~\cite{KING2002150} showed how to process a directed acyclic graph $G$
so that, for any pair of query vertices $x$ and $y$, and a query edge $e$, one can test in constant time if there is a path from $x$ to $y$ in $G \setminus e$.
Demetrescu et al.~\cite{Demetrescu:2008} considered the problem of preprocessing an edge-weighted digraph $G$ to answer queries that ask
for the shortest distance from any given vertex $x$ to any other vertex $y$ avoiding an arbitrary failed vertex or edge. They provide an oracle
that answers such queries in constant time using $O(n^2 \log{n})$ space.
Our framework allows us to answer in asymptotically optimal time and space various queries related to the strongly connected components of
a digraph $G$ under an arbitrary edge or vertex failure. We can not only compute the SCCs that remain in $G$ after the deletion of an edge
or a vertex, but we can also report various statistics such as the number of SCCs in \emph{constant time} per query (failed) edge or vertex.
We can also extend our framework to compute a class of functions over the number of SCCs again in constant time per query edge or vertex.
Our framework also supports constant-time path queries under failures, such as, are there paths from $x$ to $y$ and from $y$ to $x$ after the deletion of a given
failed edge or vertex?
Specifically, they presented an algorithm that reports the strongly connected components of $G \setminus F$, for any set $F$ of $k$ edges
or vertices, in $O(2^k n \log^2{n})$ time. Their algorithm uses a data structure for $G$ of size $O(2^k n^2)$, computed in
$O(2^k n^2 m)$ time during a preprocessing phase. Note that for $k=1$, our algorithm has faster preprocessing and query time, while requiring only $O(n)$ space.
Moreover, our  framework enables us to answer in asymptotically optimal time a collection of basic $2$-edge and $2$-vertex connectivity queries on digraphs.

\paragraph{\bf Key ideas.}
All our results are obtained with
a common algorithmic framework, based
on a novel combination of
dominance relations~\cite{Aho,Lorry} and loop nesting forests~\cite{st:t}.
We remark that the $1$-connectivity cuts of a digraph
can be found efficiently with the use of two dominator trees \cite{Italiano2012}.
However, these dominator trees alone do not reveal enough information in order
to determine the strongly connected components after the deletion of a single edge
or vertex in these cuts.
One of our key observations is that we can complement the dominance information with
loop nesting information, and obtain a new compact representation of the decompositions induced by
the $1$-connectivity cuts of digraphs, which consists simply of four trees: two dominator trees and two loop nesting trees.
Still, combining these four trees in order to extract the decompositions induced by
the $1$-connectivity cuts turns out to be a nontrivial task. One of the main
technical difficulties is to locate or count the vertices that are
common in different subtrees of these four trees. To overcome this obstacle,
we develop a novel technique, that takes advantage of relations between
dominator and loop nesting trees. Moreover, our techniques can be generalized
so that we can compute several functions defined on the decompositions induced by
the $1$-connectivity cuts.
We believe that the framework developed
in this paper may be of independent interest and may prove to be useful for other problems as well. In particular, after this work, we have been able to apply it to the incremental
maintenance of $2$-edge
connectivity properties on digraphs \cite{GIN16:ICALP,GIN18:LATIN}.
\nocite{cytron:91:toplas,domin:lt}

\paragraph{Organization of the paper.}
The remainder of the paper is organized as follows.
In Section \ref{sec:definitions} we introduce some preliminary definitions and terminology.
Section \ref{sec:scc} describes our new framework for
representing the structure of strong bridges and the decomposition they induce in a digraph. We use this framework to perform computations on the strongly connected components that could be obtained after an edge deletion, such as reporting all strongly connected components, counting the number of strongly connected components, and computing the size of the largest and of the smallest strongly connected component. In Section \ref{sec:2ECBs} we describe our new and simpler linear-time algorithm for computing the $2$-edge-connected components of a strongly connected digraph $G$, while in Section \ref{sec:other} we deal with pairwise $2$-edge connectivity queries. Section \ref{sec:vertices}, Section \ref{sec:2VCBs} and Section \ref{sec:vertices-other} extend to vertex connectivity respectively the results of Section \ref{sec:scc}, Section \ref{sec:2ECBs} and Section \ref{sec:other}, respectively.
In Section \ref{sec:sparse-certificate}, we describe how to construct in linear time, a sparse strongly connected spanning subgraph of
a strongly connected digraph $G$ that maintains the same decompositions induced by
the $1$-connectivity cuts of $G$, and the $2$-edge and $2$-vertex-connected components of $G$.
Finally, Section \ref{sec:conclusions} contains some concluding remarks and lists some open problems.

\section{Preliminaries}
\label{sec:definitions}

We assume that the reader is familiar with standard graph terminology.
All graphs in this paper are directed, i.e., an edge $e = (u,v)$ in digraph $G$ is directed from $u$, the \emph{tail} of $e$, to $v$, the \emph{head} of $e$. We also assume that at this point the reader is familiar with the definitions of $2$-edge and $2$-vertex connectivity on directed graphs already given in the introduction and contained in more detail in~\cite{2ECC:GILP:TALG,2VCB}.

Let $T$ be a rooted tree. Throughout the paper, we assume that the edges of $T$ are directed away from the root.
For each directed edge $(u,v)$ in $T$, we say that $u$ is a parent of $v$ (and we denote it by $t(v)$) and that $v$ is a child of $u$. Every vertex except the root has a unique parent.
If there is a (directed) path from vertex $v$ to vertex $w$ in $T$, we say that $v$ is an ancestor of $w$ and that $w$ is a descendant of $v$, and we denote this path by $T[v,w]$. If $v\neq w$, we say that $v$ is a proper ancestor of $w$ and that $w$ is a proper descendant of $v$, and denote by
$T(v, w]$ the path in $T$ to $w$ from the child of $v$ that is an ancestor of $w$. The ancestor/descendant relationship can be naturally extended to edges. Namely, let $w$ be a vertex and $(u,v)$ be an edge of $T$.
If $w$ is an ancestor (resp., proper ancestor) of $u$ (and thus also of $v$), we say that
$w$ is an ancestor (resp., proper ancestor) of edge $(u,v)$ and that $(u,v)$ is a descendant (resp., proper descendant) of $w$.
Similarly, if $w$ is a descendant (resp., proper descendant) of $v$ (and thus also of $u$), we say that
$w$ is a descendant (resp., proper descendant) of edge $(u,v)$ and that $(u,v)$ is an ancestor (resp., proper ancestor) of $w$. Let $(u,v)$ and $(w,z)$ be two edges in $T$.
If vertex $v$ is an ancestor of vertex $w$, we say that edge
$(u,v)$ is an ancestor of edge $(w,z)$ and that edge $(w,z)$ is a descendant of edge $(u,v)$.
Also, for a rooted tree $T$, we let $T(v)$ denote the subtree of $T$ rooted at $v$, and we also view $T(v)$ as the set of descendants of $v$.

Let $T$ be a depth first search (dfs)
tree of a digraph $G$, starting from a given vertex $s$.
Edge $(v, w)$ of the digraph $G$ is a \emph{tree edge} if $v = t(w)$, a \emph{forward edge} if $v$ is a proper ancestor of $t(w)$ in $T$, a \emph{back edge} if $v$ is a proper descendant of $w$ in $T$, and a \emph{cross edge} if $v$ and $w$ are unrelated in $T$.
A \emph{preorder} of $T$ is a total order of the vertices of $T$ such that, for every vertex $v$, the descendants of $v$ are ordered consecutively, with $v$ first. It can be obtained by a depth-first traversal of $T$, by ordering the vertices in the order they are first visited by the traversal.
The following lemma is an immediate consequence of depth-first search.

\begin{lemma} \emph{(Path Lemma \cite{dfs:t})}
\label{lemma:dfs}
Let $T$ be a dfs tree of a digraph $G$, and let $\mathit{pre}(v)$ denote the preorder number of vertex $v$ in $T$. If $v$ and $w$ are vertices such that $\mathit{pre}(v)<\mathit{pre}(w)$, then any path from  $v$ to $w$ must contain a common ancestor of $v$ and $w$ in $T$.
\end{lemma}

\subsection{Flow graphs, dominators, and bridges}
\label{sec:dominators}

A \emph{flow graph} is a directed graph with a distinguished \emph{start vertex} $s$ such that every vertex is reachable from $s$. Let $G=(V,E)$ be a strongly connected graph.
The \emph{reverse digraph} of $G$, denoted by $G^R=(V, E^R)$, is the digraph that results from $G$ by reversing the direction of all edges.
Throughout the paper we let $s$ be a fixed but arbitrary start vertex of $G$.
Since $G$ is strongly connected, all vertices are reachable from $s$ and reach $s$, so we can view both $G$ and $G^R$ as flow graphs with start vertex $s$. To avoid ambiguities, throughout the paper we will denote those flow graphs respectively by $G_s$ and $G_s^R$.

Let $G_s$ be a flow graph with start vertex $s$.
A vertex $u$ is a \emph{dominator} of a vertex $v$ ($u$ \emph{dominates} $v$) if every path from $s$ to $v$ in $G_s$ contains $u$; $u$ is a \emph{proper dominator} of $v$ if $u$ dominates $v$ and $u \not= v$.
Let \emph{dom$(v)$} be the set of dominators of $v$. Clearly, \emph{dom}$(s)$ $=\{s\}$ and for any $v\neq s$ we have that $\{s,v\}\subseteq$ \emph{dom$(v)$}: we say that $s$ and $v$ are the \emph{trivial dominators} of $v$ in the flow graph $G_s$.
The dominator relation is reflexive and transitive. Its transitive reduction is a rooted tree, the \emph{dominator tree} $D$: $u$ dominates $v$ if and only if $u$ is an ancestor of $v$ in $D$.
If $v \not= s$,  the parent of $v$ in $D$, denoted by $d(v)$, is the \emph{immediate dominator} of $v$: it is the unique proper dominator of $v$ that is dominated by all proper dominators of $v$. Similarly, we can define the dominator relation in the flow graph $G_s^R$, and let $D^R$ denote the dominator tree of $G_s^R$. We also denote the immediate dominator of $v$ in $G_s^R$ by $d^R(v)$.
Throughout the paper, we let $N$ (resp., $N^R$) denote the set of nontrivial dominators of  $G_s$ (resp., $G^R_s$).
Lengauer and Tarjan~\cite{domin:lt} presented an algorithm for computing dominators in  $O(m \alpha(m,n))$ time for a flow graph with $n$ vertices and $m$ edges, where $\alpha$ is a functional inverse of Ackermann's function~\cite{dsu:tarjan}.
Subsequently, several linear-time algorithms
were discovered~\cite{domin:ahlt,dominators:bgkrtw,dominators:Fraczak2013,dominators:poset}.

An edge $(u,v)$ is a \emph{bridge} of a flow graph $G_s$ if all paths from $s$ to $v$ include $(u,v)$.\footnote{Throughout the paper, to avoid confusion we use consistently the term  \emph{bridge} to refer to a bridge of a flow graph and the term \emph{strong bridge} to refer to a strong bridge in the original graph.}
The following properties were proved in ~\cite{Italiano2012}.
\begin{property}
\label{property:strong-bridge} \emph{(\cite{Italiano2012})}
Let $s$ be an arbitrary start vertex of $G$. An edge $e=(u,v)$ is strong bridge of $G$ if and only if it is a bridge of $G_s$ (so $u=d(v)$) or a bridge of $G_s^R$ (so $v=d^R(u)$) or both.
\end{property}

\begin{property}
\label{property:strong-articulation-point} \emph{(\cite{Italiano2012})}
Let $s$ be an arbitrary start vertex of $G$. A vertex $v \not= s$ is a strong articulation point of $G$ if and only if $v$ is a nontrivial dominator in $G_s$ or a nontrivial dominator in $G_s^R$ or both.
\end{property}

As a consequence of Property~\ref{property:strong-bridge},
all the strong bridges of the digraph $G$ can be obtained from the bridges of the flow graphs $G_s$ and $G_s^R$, and thus there can be at most $(2n-2)$ strong bridges in a digraph $G$. Figure~\ref{figure:example1ab} illustrates a strongly connected graph $G$, its reverse graph $G^R$ and the dominator trees $D$ and $D^R$.
Let $e=(u,v)$ be a strong bridge of $G$ such that $(v,u)$ is a bridge of $G_s^R$. Since there is no danger of ambiguity, with a little abuse of notation we will say that $e$ is a bridge of $G_s^R$ (although it is actually the reverse edge $(v,u)$ that is a bridge of $G_s^R$).
We refer to the edges that are bridges in both flow graphs $G_s$ and $G_s^R$ as \emph{common bridges}.
We will use the following lemmata from \cite{2ECC:GILP:TALG,2VCB} that hold for a flow graph $G_s$ of a strongly connected digraph $G$.

\begin{figure}[t!]
	\begin{center}
		\centerline{\includegraphics[trim={0 0 0 0cm}, clip=true, width=1\textwidth]{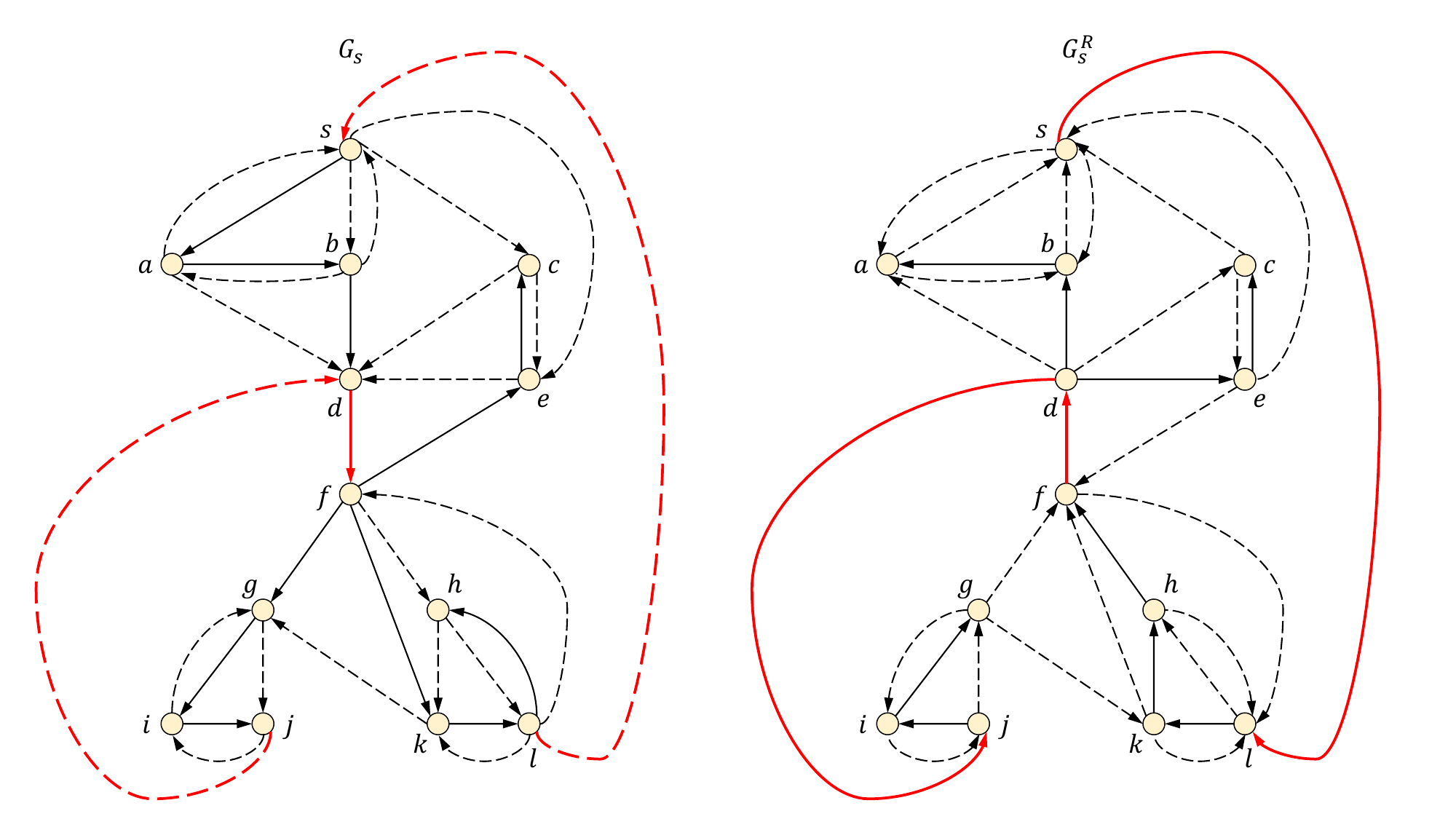}}
		\centerline{\includegraphics[trim={0 0 0 6cm}, clip=true, width=1.1\textwidth]{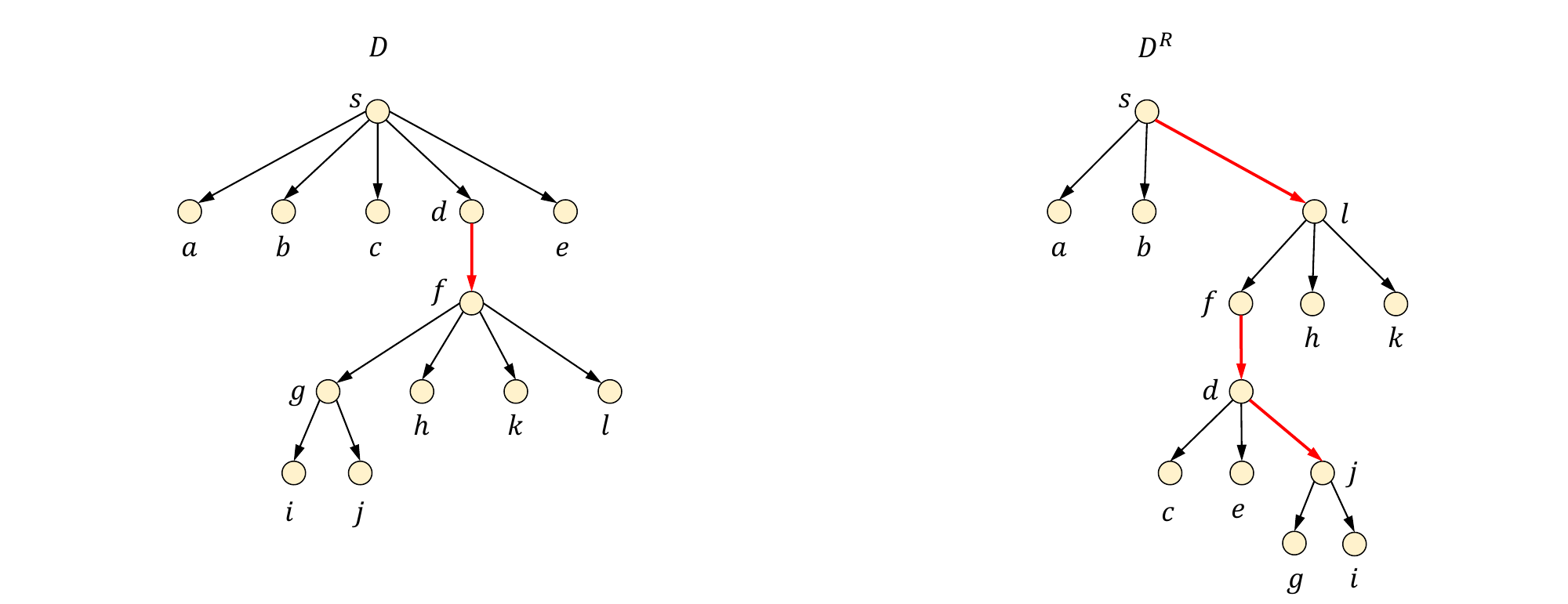}}
		\caption{A flow graph $G_s$ and its reverse $G_s^R$, and their dominator trees $D$ and $D^R$. The corresponding digraph $G$ is strongly connected. The solid edges in $G_s$ and  $G_s^R$ are the edges of depth first search trees with root $s$. Strong bridges of $G$ and $G^R$ and bridges of $G_s$ and $G_s^R$ in $D$ and $D^R$ are shown red. The bridge decomposition of $D$ and $D^R$ is obtained after deleting the red edges. (Better viewed in color.)}
		\label{figure:example1ab}
	\end{center}
\end{figure}

\begin{lemma}
\label{lemma:partition-paths} \emph{(\cite{2ECC:GILP:TALG})}
Let $G$ be a strongly connected digraph and let $(u,v)$ be a strong bridge of $G$. Also, let $D$ and $D^R$ be the dominator trees of the corresponding flow graphs $G_s$ and $G_s^R$, respectively, for an arbitrary start vertex $s$.
\begin{itemize}
\item[(a)] Suppose $u=d(v)$. Let $w$ be any vertex that is not a descendant of $v$ in $D$. Then there is a path from $w$ to $v$ in $G$ that does not contain any proper descendant of $v$ in $D$. Moreover, all simple paths in $G$ from $w$ to any descendant of $v$ in $D$ must contain the edge $(d(v),v)$.
\item[(b)] Suppose $v=d^R(u)$. Let $w$ be any vertex that is not a descendant of $u$ in $D^R$. Then there is a path from $u$ to $w$ in $G$ that does not contain any proper descendant of $u$ in $D^R$. Moreover, all simple paths in $G$ from any descendant of $u$ in $D^R$ to $w$ must contain the edge $(u,d^R(u))$.
\end{itemize}
\end{lemma}

\begin{lemma}\emph{(\cite{2VCB})}
\label{lemma:paths-through-SAP}
Let $G$ be a strongly connected digraph and let $v$ be a strong articulation point of $G$. Also, let $D$ and $D^R$ be the dominator trees of the corresponding flow graphs $G_s$ and $G_s^R$, respectively, for an arbitrary start vertex $s$.
\begin{itemize}
\item[(a)] Suppose $v$ is a nontrivial dominator of $G_s$. Let $w$ be any vertex that is not a descendant of $v$ in $D$.
Then there is a path from $w$ to $v$ in $G$ that does not contain any proper descendant of $v$ in $D$. Moreover, all simple paths in $G$ from $w$ to any descendant of $v$ in $D$ must contain $v$.
\item[(b)] Suppose $v$ is a nontrivial dominator of $G^R_s$. Let $w$ be any vertex that is not a descendant of $v$ in $D^R$.
Then there is a path from $v$ to $w$ in $G$ that does not contain any proper descendant of $v$ in $D^R$. Moreover, all simple paths in $G$ from any descendant of $v$ in $D^R$ to $w$ must contain $v$.
\end{itemize}
\end{lemma}

After deleting from the dominator trees $D$ and $D^R$ respectively the bridges of $G_s$ and $G_s^R$, we obtain the \emph{bridge decomposition}  of $D$ and $D^R$ into forests $\mathcal{D}$ and $\mathcal{D}^R$ (see Figure~\ref{figure:example1ab}).
Throughout the paper, we denote by $D_u$ (resp., $D_u^R$) the tree in $\mathcal{D}$ (resp., $\mathcal{D}^R$) containing vertex $u$, and by $r_u$ (resp., $r^R_u$) the root of $D_u$ (resp., $D_u^R$).

\subsection{Loop nesting forests}
\label{sec:loop-nesting}

Let $G=(V,E)$ be a directed graph.
A \emph{loop nesting forest} represents a hierarchy of strongly connected subgraphs of $G$~\cite{st:t}, and is defined with respect to a dfs tree $T$ of $G$ as follows.
For any vertex $u$, the \emph{loop} of $u$, denoted by $\mathit{loop}(u)$, is the set of all descendants $x$ of $u$ in $T$ such that there is a path from $x$ to $u$ in $G$ containing only descendants of $u$ in $T$. Vertex $u$ is the \emph{head} of $\mathit{loop}(u)$. Any two vertices in $\mathit{loop}(u)$ reach each other. Therefore, $\mathit{loop}(u)$ induces a strongly connected subgraph of $G$; it is the unique maximal set of descendants of $u$ in $T$ that does so.
The $\mathit{loop}(u)$ sets form a laminar family of subsets of $V$:
for any two vertices $u$ and $v$, $\mathit{loop}(u)$ and $\mathit{loop}(v)$ are either disjoint or nested (i.e., one contains the other).
The above property allows us to define the \emph{loop nesting forest} $H$ of $G$, with respect to $T$, as the forest in which the parent of any vertex $v$, denoted by $h(v)$, is the nearest proper ancestor $u$ of $v$ in $T$ such that $v \in \mathit{loop}(u)$ if there is such a vertex $u$, and null otherwise.
Then $\emph{loop}(u)$ is the set of all descendants of vertex $u$ in $H$, which we will also denote as $H(u)$ (the subtree of $H$ rooted at vertex $u$).
Since $T$ is a dfs tree, every cycle contains a back edge \cite{dfs:t}.  More generally, every cycle $C$ contains a vertex $u$ that is a common ancestor of all other vertices $v$ of $T$ in the cycle \cite{dfs:t}, which means that any $v \in C$ is in $\mathit{loop}(u)$. Hence, every cycle of $G$ is contained in a loop.
A loop nesting forest can be computed in linear time~\cite{dominators:bgkrtw,st:t}.
Since here we deal with strongly connected digraphs, each vertex is contained in a loop, so $H$ is a tree.
Therefore, we will refer to $H$ as the \emph{loop nesting tree} of $G$.
Figure \ref{figure:example1c} shows the loop nesting trees $H$ and $H^R$ of the flow graphs $G_s$ and $G_s^R$ given in  Figure \ref{figure:example1ab}.

\begin{figure}[ht!]
\begin{center}
\centerline{\includegraphics[trim={0 0 0 5.5cm}, clip=true, width=1.3\textwidth]{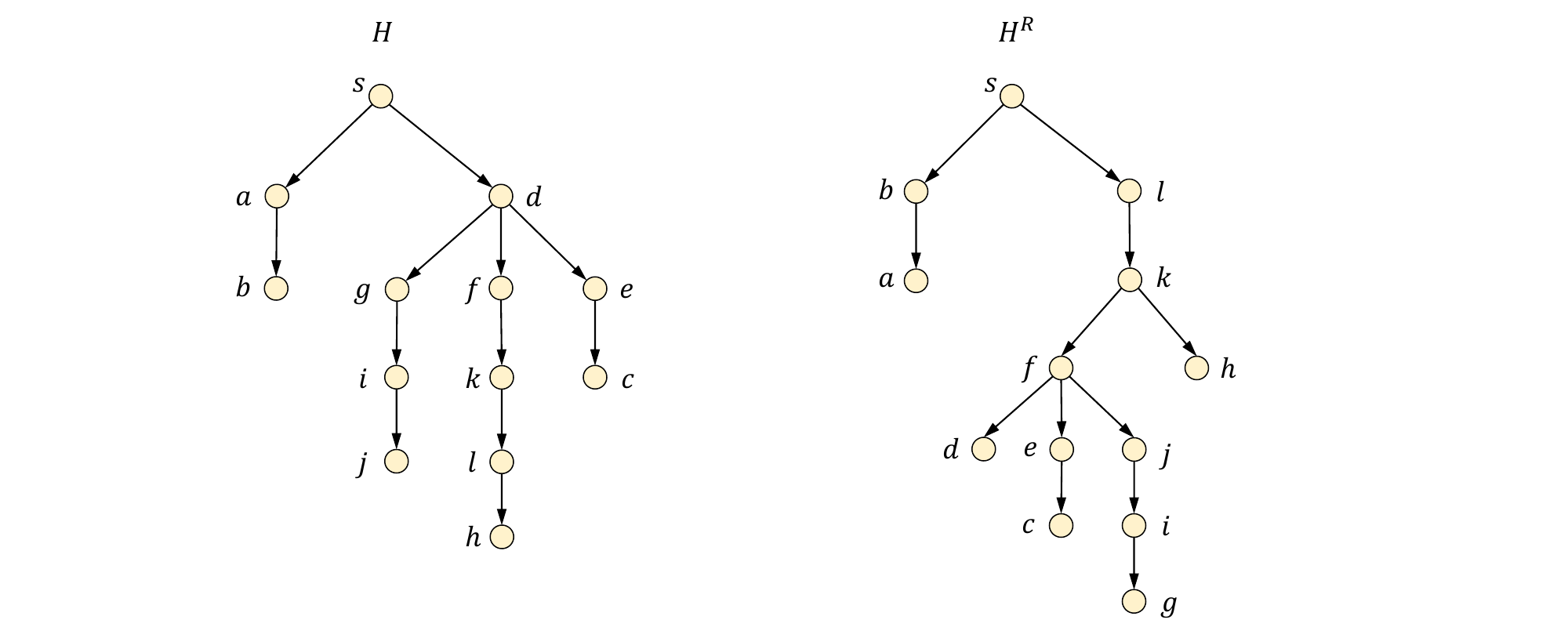}	}
\caption{The loop nesting trees $H$ and $H^R$ of the flow graphs $G_s$ and $G_s^R$ of Figure \ref{figure:example1ab} respectively, with respect to the dfs trees shown in Figure \ref{figure:example1ab}.}
\label{figure:example1c}
\end{center}
\end{figure}

The following lemma will be useful throughout the paper.
\begin{lemma}
Let $x\not=s$ be any vertex in $G$. 
Let $h(x)$ be the parent of $x$ in the loop nesting tree $H$ and ${r}_{h(x)}$ be the root of the tree $D_{h(x)}$ in the bridge decomposition $\mathcal{D}$.
Then
${r}_{h(x)}$ is an ancestor of $x$ in $D$.
\label{lemma:ancestor-of-x}
\end{lemma}
\begin{proof}
Assume by contradiction that ${r}_{h(x)}$ is not an ancestor of $x$ in $D$.
Let $T$ be the dfs tree based on which the loop nesting tree $H$ was built.
By the definition of dominators, all paths from $s$ to $h(x)$ contain ${r}_{h(x)}$, and therefore ${r}_{h(x)}$ is an ancestor of $h(x)$ in $T$.
By Lemma \ref{lemma:partition-paths}(a) all paths from $x$ to $h(x)$ in $G$ contain the strong bridge $(d({r}_{h(x)}),{r}_{h(x)})$.
But this is a contradiction to the fact that there is a path from $x$ to $h(x)$ containing only descendants of $h(x)$ in $T$, since ${r}_{h(x)}$ is an ancestor of $h(x)$ in $T$.
\end{proof}

\section{Strongly connected components in $G\setminus e$}
\label{sec:scc}

In this section, we describe our
compact representation of the structure of all the $1$-edge cuts (given by strong bridges) of a strongly connected digraph $G$. Let $G^R$ be the reverse digraph of $G$, and let $s$ be any vertex in $G$. Let $G_s$ be the flow graph with start vertex $s$ and let $G_s^R$ be the reverse flow graph with start vertex $s$.
Let $D$ and $D^R$ be the dominator trees of $G_s$ and $G_s^R$, and let $H$ and $H^R$ be the loop nesting trees of $G_s$ and $G_s^R$.
We show that the four trees $D$, $D^R$, $H$ and $H^R$  are sufficient to encode efficiently the decompositions that the strong bridges induce in $G$, i.e., all the strongly connected components of $G\setminus e$, \emph{for all strong bridges $e$} in $G$.
In particular, let $e$ be a strong bridge in $G$. We will show how the four trees $D$, $D^R$, $H$ and $H^R$ can be effectively exploited for solving efficiently the following problems:
\begin{itemize}
\item Compute all the strongly connected components of $G\setminus e$;
\item Count the number of strongly connected components of $G\setminus e$;
 \item{Find the size of the smallest or the largest strongly connected component of $G\setminus e$.}
\end{itemize}

Throughout this section, we assume without loss of generality that the input digraph $G$ is strongly connected (otherwise,
we apply our algorithms to the strongly connected components of $G$). If $G$ is strongly connected, then $m\geq n$, where $m$ and $n$ are respectively the number of edges and vertices in $G$, which will simplify some of the bounds. Since all our algorithms are based on dominator trees and loop nesting forests, we fix arbitrarily a start vertex $s$ in $G$.
We also restrict our attention to the strong bridges of $G$, since only the deletion of a strong bridge of $G$ can affect its strongly connected components.

Let $G=(V,E)$ be a strongly connected digraph, $s$ be an arbitrary start vertex in $G$,  and let $e=(u,v)$ be  a strong bridge of $G$. We will first prove some general properties of the strongly connected components of $G\setminus e$. We will then exploit those properties in Sections \ref{sec:all-scc}, \ref{sec:SCCs-num} and \ref{sec:minmax-scc}
to design efficient solutions for our problems.
Consider the dominator relations in the flow graphs $G_s$ and $G_s^R$.
By Property \ref{property:strong-bridge}, one of the following cases must hold:
\begin{itemize}
\item[(a)] $e$ is a bridge in $G_s$ but not in $G_s^R$, so $u=d(v)$ and $v \not=d^R(u)$.
\item[(b)] $e$ is a bridge in $G_s^R$ but not in $G_s$, so $u \not=d(v)$ and $v=d^R(u)$.
\item[(c)] $e$ is a common bridge, i.e., a bridge both in $G_s$ and in $G_s^R$, so $u=d(v)$ and $v=d^R(u)$.
\end{itemize}

We will show how to compute the strongly connected components of $G \setminus e$ in each of these cases. Consider $u=d(v)$, i.e., when either (a) or (c) holds.
Case (b) is symmetric to (a).
By Lemma \ref{lemma:partition-paths}, the deletion of the edge $e=(u,v)$ separates the descendants of $v$ in $D$, denoted by $D(v)$, from $V \setminus D(v)$. Therefore, we can compute separately the strongly connected components of the subgraphs of $G \setminus e$ induced by the vertices in $D(v)$ and by the vertices in $V \setminus D(v)$.
We begin with some lemmata that allow us to compute the strongly connected components in the subgraph of $G \setminus e$ that is induced by the vertices in $D(v)$. We recall here that, given the loop nesting tree $H$ of $G_s$ and a vertex $w$ in $H$, we denote by $h(w)$ the parent of $w$ in $H$ and by
$H(w)$ the set of descendants of  $w$ in $H$.

\begin{lemma}
\label{lemma:subtree1}
Let $e=(u,v)$ be a strong bridge of $G$ that is also a bridge in the flow graph $G_s$ (i.e., $u=d(v)$).
For any vertex $w \in D(v)$ the vertices in $H(w)$ are contained in a strongly connected component $C$ of $G \setminus e$ such that $C \subseteq D(v)$.
\end{lemma}
\begin{proof}
Let $x \in D(v)$ be a vertex such that $w=h(x) \in D(v)$. We claim that $w$ and $x$ are strongly connected in $G \setminus e$.
Let $T$ be the dfs tree that generated $H$.
Note that, since $w=h(x)$ and $x,w\in D(v)$,
the path $\pi_1$ from $w$ to $x$ in the dfs tree $T$ (in which $x$ is a descendant of $w$) avoids the edge $e=(u,v)$. To show that $w$ and $x$ are strongly connected in $G \setminus e$, we exhibit a path $\pi_2$ from $x$ to $w$ that avoids the edge $e$.
Indeed, by the definition of the loop nesting forest, there is a path $\pi_2$ from $x$ to $w$ that contains only descendants of $w$ in $T$.
Note that $\pi_2$ cannot contain the edge $e$ since $d(v)$ is a proper ancestor of $v$ in $T$ and all descendants of $w$ in $T$ are descendants of $v$ in $T$, since all paths from $s$ to $w$ contain $v$.
Assume by contradiction that either $\pi_1$ or $\pi_2$ contains a vertex $z \not\in D(v)$. But then, Lemma \ref{lemma:partition-paths} implies that the either the subpath of $\pi_1$ from $z$ to $x$
or the subpath of $\pi_2$ from $z$ to $w$ contains $e$, a contradiction.
This implies that every pair of vertices in $H(w)$ are strongly connected in $G \setminus e$. Let $C$ be the strongly connected component of $G \setminus e$ that contains $H(w)$. The same argument implies that all vertices in $C$ are descendants of $v$ in $D$.
\end{proof}

\begin{lemma}
\label{lemma:subtree2}
Let $e=(u,v)$ be a strong bridge of $G$ that is also a bridge in the flow graph $G_s$ (i.e., $u=d(v)$).
For every strongly connected component $C$ in $G\setminus e$ such that $C \subseteq D(v)$, there is a vertex $w \in C$ that is a common ancestor in $H$ of all vertices in $C$.
Moreover, $C=H(w)$.
\end{lemma}
\begin{proof}
Let $C$ be a strongly connected component of $G \setminus e$ that contains only descendants of $v$ in $D$.
Let $T$ be the dfs tree that generated $H$ and let $\mathit{pre}$ be the corresponding preorder numbering of the vertices.
Define $w$ to be the vertex in $C$ with minimum preorder number with respect to $T$.
Consider any vertex $z \in C \setminus w$.
Since $C$ is a strongly connected component, there is a path $\pi$ from $w$ to $z$ that contains only vertices in $C$.
By the choice of $w$, $\mathit{pre}(w)<\mathit{pre}(z)$, so Lemma \ref{lemma:dfs} implies that $\pi$ contains a common ancestor $q$ of $w$ and $z$ in $T$.
Also $q \in C$, since $\pi$ contains only vertices in $C$.
But then $q=w$, since otherwise $\mathit{pre}(q)<\mathit{pre}(w)$ which contradicts the choice of $w$.
Hence $w$ is also an ancestor of $z$ in $H$. 
Moreover, this implies $C \subseteq H(w)$.
Since $w \in D(v)$, we have $H(w) \subseteq C$ by Lemma \ref{lemma:subtree1}.
Hence, $C=H(w)$.
\end{proof}

\begin{lemma}
\label{lemma:subtree}
Let $e=(u,v)$ be a strong bridge of $G$ that is also a bridge in the flow graph $G_s$ (i.e., $u=d(v)$). Let $w$ be a vertex such that $w \in D(v)$ and $h(w) \not\in D(v)$. Then, the subgraph induced by $H(w)$ is a strongly connected component of $G \setminus e$.
\end{lemma}
\begin{proof}
From Lemma \ref{lemma:subtree1} we have that $H(w)$ is contained in a strongly connected component $C \subseteq D(v)$ of $G \setminus e$. Let $x \in C$ be the vertex that is a common ancestor in $H$ of all vertices in $C$, as stated by Lemma \ref{lemma:subtree2}. The fact that $h(w) \not\in D(v)$ implies $x=w$. 
So, by Lemma \ref{lemma:subtree2},
 $H(w)$ is a maximal subset of vertices that are strongly connected in $G \setminus e$. Thus $H(w)$ induces a strongly connected component of $G \setminus e$.
\end{proof}

Next we consider the strongly connected components of the subgraph of $G \setminus e$ induced by $V \setminus D(v)$.

\begin{lemma}
\label{lemma:above-tree-1}
Let $e=(u,v)$ be a strong bridge of $G$ that  is a bridge in $G_s$ but not in $G_s^R$ (i.e., such that $u=d(v)$ and $v \not= d^R(u)$).
Let $C = V \setminus D(v)$. Then, the subgraph induced by $C$ is a strongly connected component of $G \setminus e$.
\end{lemma}
\begin{proof}
By Lemma \ref{lemma:partition-paths}(a), a vertex in $C$ cannot be strongly connected in $G \setminus e$ to a vertex in $D(v)$. Thus, it remains to show that the vertices in $C$ are strongly connected in $G \setminus e$.
Note that by the definition of
$C$ we have that $s \in C$. Now it suffices to show that for any vertex $w \in C$, digraph $G$ has a path $\pi$ from $s$ to $w$ and a path $\pi'$ from $w$ to $s$ containing only vertices in $C$.
Assume by contradiction that all paths in $G$ from $s$ to $w$ contain a vertex in $D(v)$. Then, Lemma \ref{lemma:partition-paths} implies that all paths from $s$ to $w$ contain $e$, which contradicts the fact that $w \not \in D(v)$. We use a similar argument for the paths from $w$ to $s$. If all such paths contain a vertex in $D(v)$, then by Lemma \ref{lemma:partition-paths} we have that all paths from $w$ to $s$ contain $e$. Hence, also all paths from $u$ to $s$ must contain $e$, which contradicts the fact that  $v \not= d^R(u)$.
\end{proof}

Finally we deal with the more complicated case (c). We refer the reader to Figure \ref{figure:single bridge} for an illustration of the sets involved in the lemma.

\begin{figure}[ht!]
	\begin{center}
		\centerline{\includegraphics[trim={0 0 0 7cm}, clip=true, width=1.3\textwidth]{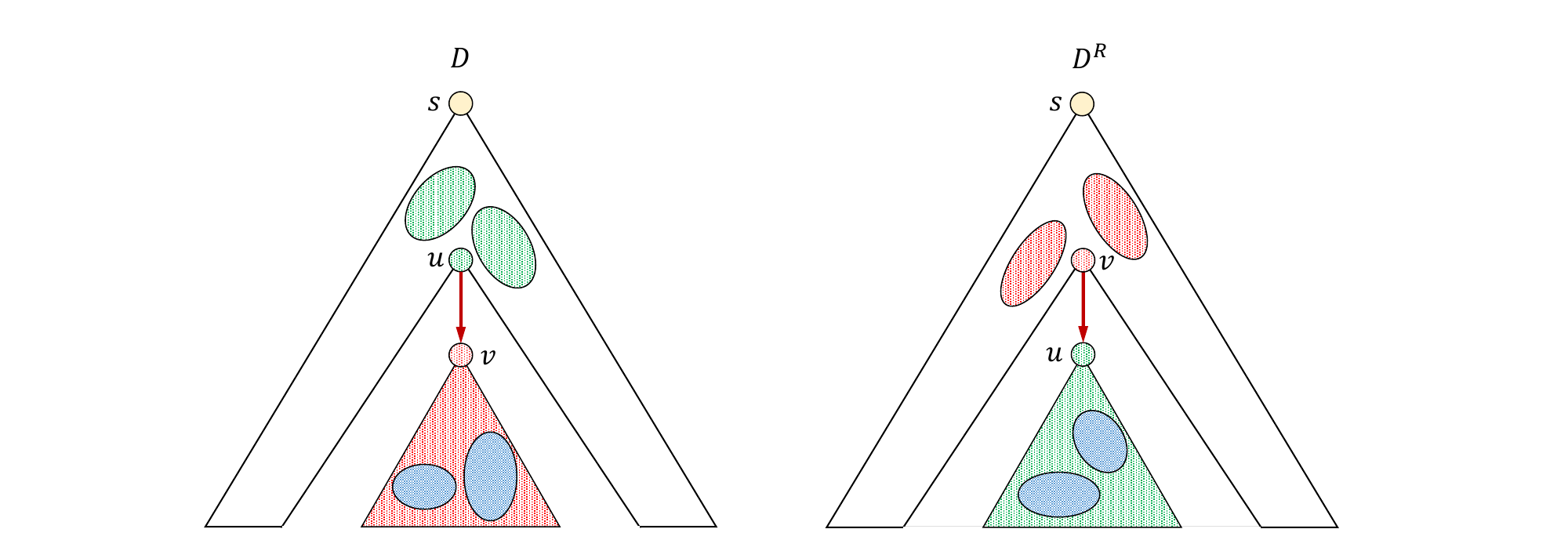}}
		\caption{An illustration of the sets $D(v)$, $D^R(u)$, $D(v) \cap D^R(u)$, and $C = V \setminus \big (D(v) \cup D^R(u) \big)$ that correspond to a common strong bridge $(u,v)$, as they appear with respect to the dominator trees $D$ and $D^R$. The vertices in $D(v) \setminus D^R(u)$ are shown in red, the vertices in $D^R(u) \setminus D(v)$ are shown in green, and the vertices in $D(v) \cap D^R(u)$ are shown in blue. (Better viewed in color.) The remaining vertices are in $C = V \setminus \big (D(v) \cup D^R(u) \big)$.}
		\label{figure:single bridge}
	\end{center}
\end{figure}

\begin{lemma}
\label{lemma:above-tree-2}
Let $e=(u,v)$ be a strong bridge of $G$  that is a common bridge of $G_s$ and $G_s^R$ (i.e., such that $u=d(v)$ and $v = d^R(u)$).
Let $C = V \setminus \big ( D(v) \cup D^R(u) \big ) $. Then, the subgraph induced by $C$ is a strongly connected component of $G \setminus e$.
\end{lemma}
\begin{proof}
The fact that $e$ is a strong bridge and Lemma \ref{lemma:partition-paths} imply
that
the following properties
hold
in $G$:
\begin{itemize}
\item[(1)]  There is a path from $s$ to any vertex in $V \setminus D(v)$ that does not contain $e$.
\item[(2)]   There is a path from any vertex in $V \setminus D^R(u)$ to $s$ that does not contain $e$.
\item[(3)]  There is no edge $(x,y)\neq e$ such that $x \notin D(v)$ and $y \in D(v)$. In particular, since $C \subseteq V \setminus D(v)$, there is no edge $(x,y) \neq e$ such that $x \in C$ and $y\in D(v)$.
\item[(4)]  Symmetrically, there is no edge $(x,y) \neq e$ such that $x \in D^R(u)$ and $y \not \in D^R(u)$. In particular, since $C \subseteq V \setminus D^R(u)$, there is no edge $(x,y) \neq e$ such that $x \in D^R(u)$ and $y\in C$.
\end{itemize}
Let $K$ be a strongly connected component of $G \setminus e$ such that $K \cap C \not = \emptyset$.
By properties (3) and (4) we have that $K$ contains no vertex in $V \setminus C = D(v) \cup D^R(u)$. Thus $K\subseteq C$.
Let $G_C$ be the subgraph of $G$ induced by the vertices in $C$.
We will show that for any vertex $x \in K$ digraph $G_C$ contains a path from $s$ to $x$ and a path from $x$ to $s$. This implies that
all vertices in $C$ are strongly connected in $G_C$, and hence also in $G \setminus e$. Moreover, since $K\subseteq C$ is a strongly connected component of $G\setminus e$, we have that $K=C$ and that it induces a strongly connected component in $G \setminus e$.

First we argue about the existence of a path from $s$ to $x \in C$ in $G_C$.
Let $\pi$ be a path in $G$ from $s$ to $x$ that does not contain $e$.
Property (1) guarantees that such a path exists.
Also, Lemma  \ref{lemma:partition-paths} implies that $\pi$ does not contain a vertex in $D(v)$, since otherwise $\pi$ would include $e$.
It remains to show that $\pi$ also avoids $D^R(u)$. Assume by contradiction that $\pi$ contains a vertex in $z \in D^R(u)$.
Choose $z$ to be the last such vertex in $\pi$. Since $x \not\in D^R(u)$ we have that $z \not= x$. Let $w$ be the successor of $z$ in $\pi$.
From the fact that $\pi$ does not contain vertices in $D(v)$ and by the choice of $z$ it follows that $w \in C$. But then, edge $(z,w) \not= e$ violates property (4), which is a contradiction.
We conclude that path $\pi$ also exists in $G_C$ as claimed.

The argument for the existence of a path from $x \in C$ to $s$ in $G_C$ is symmetric.
Let $\pi$ be a path in $G$ from $x$ to $s$ that does not contain $e$.
From property (2) and Lemma  \ref{lemma:partition-paths} we have that such a path $\pi$ exists and does not contain a vertex in $D^R(u)$.
Now we show that $\pi$ also avoids $D(v)$. Assume by contradiction that $\pi$ contains a vertex in $z \in D(v)$.
Choose $z$ to be the first such vertex in $\pi$.
Since $x \not\in D(v)$ we have that $z \not= x$.
Let $w$ be the predecessor of $z$ in $\pi$.
From the fact that $\pi$ does not contain vertices in $D^R(u)$ and by the choice of $z$ it follows that $w \in C$. So, edge $(w,z) \not= e$ violates property (3). Hence path $\pi$ also exists in $G_C$.
\end{proof}

\begin{lemma}
\label{lemma:common-bridge}
Let $e=(u,v)$ be a strong bridge of $G$  that is a common bridge of $G_s$ and $G_s^R$ (i.e., such that $u=d(v)$ and $v = d^R(u)$).
Let $C$ be a strongly connected component of $G \setminus e$ that contains a vertex in $D(v) \cap D^R(u)$. Then, $C \subseteq D(v) \cap D^R(u)$.
\end{lemma}
\begin{proof}
Consider any two vertices $x$ and $y$ such that $x \in D(v) \cap D^R(u)$ and
 $y \not\in D(v) \cap D^R(u)$.
We claim that $x$ and $y$ are not strongly connected in $G \setminus e$, which implies the lemma.
To prove the claim, note that by  Lemma \ref{lemma:above-tree-2}, $x$ is not strongly connected in $G \setminus e$ with any vertex in $V \setminus ( D(v) \cup D^R(u) )$.
Hence, we can assume that $y \in D(v) \setminus D^R(u)$ or $y \in D^R(u) \setminus D(v)$.
In either case, $x$ and $y$ are not strongly connected in $G \setminus e$ by Lemma \ref{lemma:partition-paths}.
\end{proof}

The following theorem summarizes the results of Lemmata \ref{lemma:subtree}, \ref{lemma:above-tree-1}, \ref{lemma:above-tree-2}, and \ref{lemma:common-bridge}.

\begin{theorem}
\label{cor:scc}
Let $G=(V,E)$ be a strongly connected digraph, $s$ be an arbitrary start vertex in $G$,  and let $e=(u,v)$ be a strong bridge of $G$. Let $C$ be a strongly connected component of $G \setminus e$. Then one of the following cases holds:
\begin{itemize}
\item[(a)] If $e$ is a bridge in $G_s$ but not in $G_s^R$ then either $C \subseteq D(v)$ or $C = V \setminus D(v)$.
\item[(b)] If $e$ is a bridge in $G_s^R$ but not in $G_s$ then either $C \subseteq D^R(u)$ or $C = V \setminus D^R(u)$.
\item[(c)] If $e$ is a common bridge of $G_s$ and $G_s^R$ then either $C \subseteq D(v) \setminus D^R(u)$, or $C \subseteq D^R(u) \setminus D(v)$, or $C \subseteq D(v) \cap D^R(u)$, or $C = V \setminus \big( D(v) \cup D^R(u) \big)$.
\end{itemize}
Moreover, if  $C \subseteq D(v)$ (resp., $C \subseteq D^R(u)$) then $C=H(w)$ (resp., $C=H^R(w)$) where $w$ is a vertex in $D(v)$ (resp., $D^R(u)$) such that $h(w) \not \in D(v)$ (resp.,  $h^R(w) \not \in D^R(u)$).
\end{theorem}

\subsection{Finding all strongly connected components of $G\setminus e$}
\label{sec:all-scc}

In this section we show how to
 exploit Theorem \ref{cor:scc}
in order to find efficiently the strongly connected components of $G\setminus e$, for each edge $e$ in $G$.
In particular, we present an $O(n)$-space data structure that, after $O(m)$-time preprocessing, given
a strong bridge $e$ of $G$ can report in $O(n)$ time all the strongly connected components of $G\setminus e$. This is a sharp improvement over the naive solution, which computes from scratch the strongly connected components of $G \setminus e$ in  $O(m)$ time. Furthermore, our bound is asymptotically tight, as one needs $O(n)$ time to output the strongly connected components of a digraph. With our data structure we can output in a total of $O(m+nb)$ worst-case time the strongly connected components of $G\setminus e$, for each edge $e$ in $G$, where $b$ is the total number of strong bridges in $G$.

We next describe our data structure. First, we process the dominator trees $D$ and $D^R$ in $O(n)$ time, so that we can test the ancestor/descendant relation in each tree in constant time~\cite{domin:tarjan}. To answer the query about a strong bridge $e=(u,v)$, we execute a preorder traversal of the loop nesting trees $H$ and $H^R$.
During those traversals, we will assign a label $\mathit{scc}(v)$ to each vertex $v$ that specifies the strongly connected component of $v$: that is, at the end of the preorder traversals of $H$ and $H^R$, each strongly connected component will consist of vertices with the same label.
More precisely,
if $e$ is a bridge of $G_s$ but not of $G^R_s$, then the preorder traversal of $H$ will identify the strongly connected components of all vertices.
Similarly, if $e$ is a bridge of $G^R_s$ but not of $G_s$, then we only execute a preorder traversal of $H^R$.
Finally, if $e$ is a common bridge, then the preorder traversal of $H$ will identify the strongly connected components of all vertices except for those in $D^R(u) \setminus D(v)$ (i.e., the green vertices in Figure \ref{figure:single bridge}). The strongly connected components of these vertices will be discovered during the subsequent preorder traversal of $H^R$.

We do this as follows. We initialize $\mathit{scc}(v)=v$ for all vertices $v$.
During our preorder traversals of $H$ and $H^R$ we will update $\mathit{scc}(v)$ for all vertices $v\neq s$. Throughout, we will always have $\mathit{scc}(s)=s$.
Suppose $e=(u,v)$ is a bridge only in $G_s$. We do a preorder traversal of $H$, and when we visit a vertex $w\neq s$ we test if the condition $(w \in D(v) \wedge h(w) \not\in D(v))$ holds. If it does, then the label of $w$ remains $\mathit{scc}(w)=w$, otherwise the label of $w$ is updated as $\mathit{scc}(w) = \mathit{scc}(h(w))$. We handle the case where $e$ is a bridge only in $G^R_s$ symmetrically.
Suppose now that $e$ is a common bridge.
During the preorder traversal of $H$, when we visit a vertex $w \not\in D^R(u) \setminus D(v)$, $w\neq s$, we test if the condition $(w \in D(v) \wedge h(w) \not\in D(v))$ holds. As before, if this condition holds then the label of $w$ remains $\mathit{scc}(w)=w$, otherwise we set $\mathit{scc}(w) = \mathit{scc}(h(w))$. Note that this process assigns $\mathit{scc}(x)=s$ to all vertices $x \in C = V \setminus \big ( D(v) \cup D^R(u) \big )$. Also, all vertices $x \in H(w)$, where $w \in D(v)$ and $h(w) \not\in D(v)$ are assigned $\mathit{scc}(x) = w$.
Finally, we need to assign appropriate labels to the vertices in $D^R(u) \setminus D(v)$. We do that by executing a similar procedure on $H^R$.  This time, when we visit a vertex $w \in D^R(u) \setminus D(v)$, $w\neq s$, we test if the condition $(w \in D^R(u) \wedge h^R(w) \not\in D^R(u))$ holds. If it does, then the label of $w$ remains $\mathit{scc}(w)=w$, otherwise we set $\mathit{scc}(w) = \mathit{scc}(h^R(w))$. At the end of this process we have that all vertices $x \in H^R(w)$, such that $w \in D^R(u)$ and $h^R(w) \not\in D^R(u)$, are assigned $\mathit{scc}(x) = w$.

In every case, Theorem~\ref{cor:scc} implies that the above procedure assigns correct labels to all vertices.
We remark that, during the execution of a query, each condition
can be tested in constant time, since it involves computing the parent of a vertex in a loop nesting tree or checking the ancestor/descendant relationship in a dominator tree. This yields the following theorem.

\begin{theorem}
\label{theorem:all-scc}
Let $G$ be a strongly connected digraph with $n$ vertices and $m$ edges. We can preprocess $G$ in $O(m)$ time and construct an $O(n)$-space data structure, so that given
an edge $e$ of $G$ we can report in $O(n)$ time all the strongly connected components of $G\setminus e$.
\end{theorem}

\begin{corollary}
	\label{corollary:all-scc}
Let $G$ be a strongly connected digraph with $n$ vertices, $m$ edges and $b$ strong bridges. We can output the strongly connected components of $G\setminus e$, for all strong bridges $e$ in $G$, in a total of $O(m+nb)$ worst-case time.
\end{corollary}

Figure \ref{figure:example1de} highlights the vertices of different sets (with respect to the strong bridge $(d,f)$)
in the dominator trees $D$ and $D^R$ and the loop nesting trees $H$ and $H^R$ of the flow graph $G_s$ and $G_s^R$ given in Figure \ref{figure:example1ab}.
Figure \ref{figure:example1f} shows the result of a reporting query for the digraph of Figure \ref{figure:example1ab} and for the edge $(d,f)$.

\begin{figure}[t!]
\begin{center}
\centerline{\includegraphics[trim={0 0 0 6cm}, clip=true, width=1.3\textwidth]{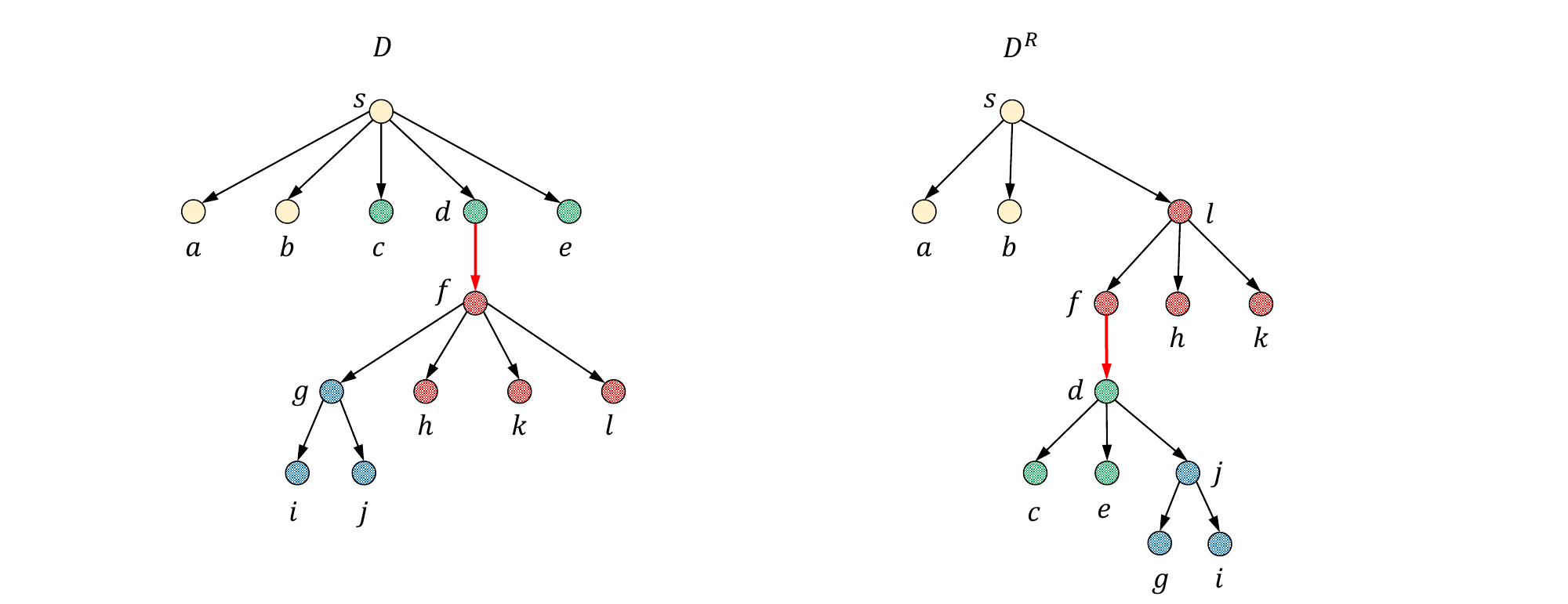}	}
\centerline{\includegraphics[trim={0 0 0 5.5cm}, clip=true, width=1.3\textwidth]{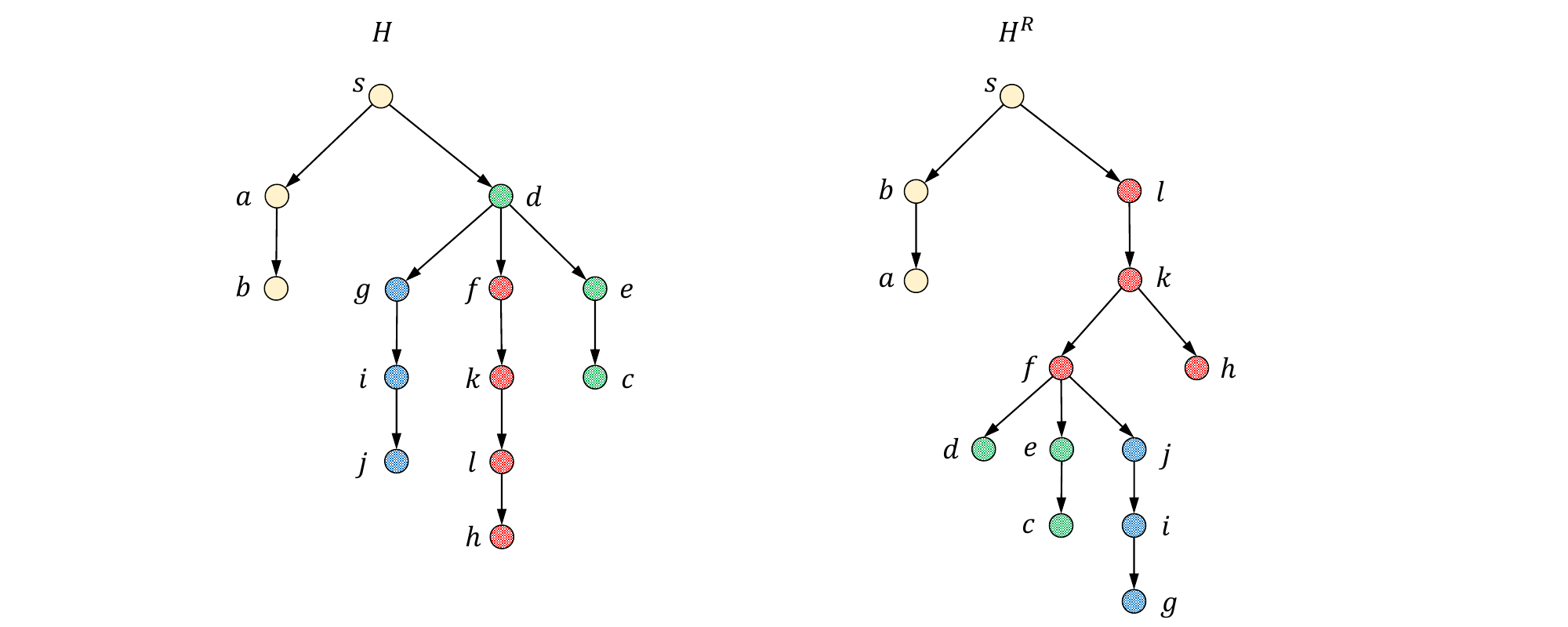}	}
\caption{The dominator trees $D$ and $D^R$ and the loop nesting trees $H$ and $H^R$ of the flow graphs $G_s$ and $G_s^R$ of Figure \ref{figure:example1ab} respectively.
 The edge $(d,f)$ in red is a common bridge. The vertices in $D(f) \setminus D^R(d)$ are shown in red. The vertices in $D^R(d) \setminus D(f)$ are shown in green, and the vertices in $D(f) \cap D^R(d)$  are shown in blue. (Better viewed in color.) The remaining vertices are in $C = V \setminus \big ( D(f) \cup D^R(d) \big )$.}
\label{figure:example1de}
\end{center}
\vspace{1cm}
\end{figure}

\begin{figure}[t!]
\begin{center}
\centerline{\includegraphics[trim={0 0 0 0cm}, clip=true, width=1.3\textwidth]{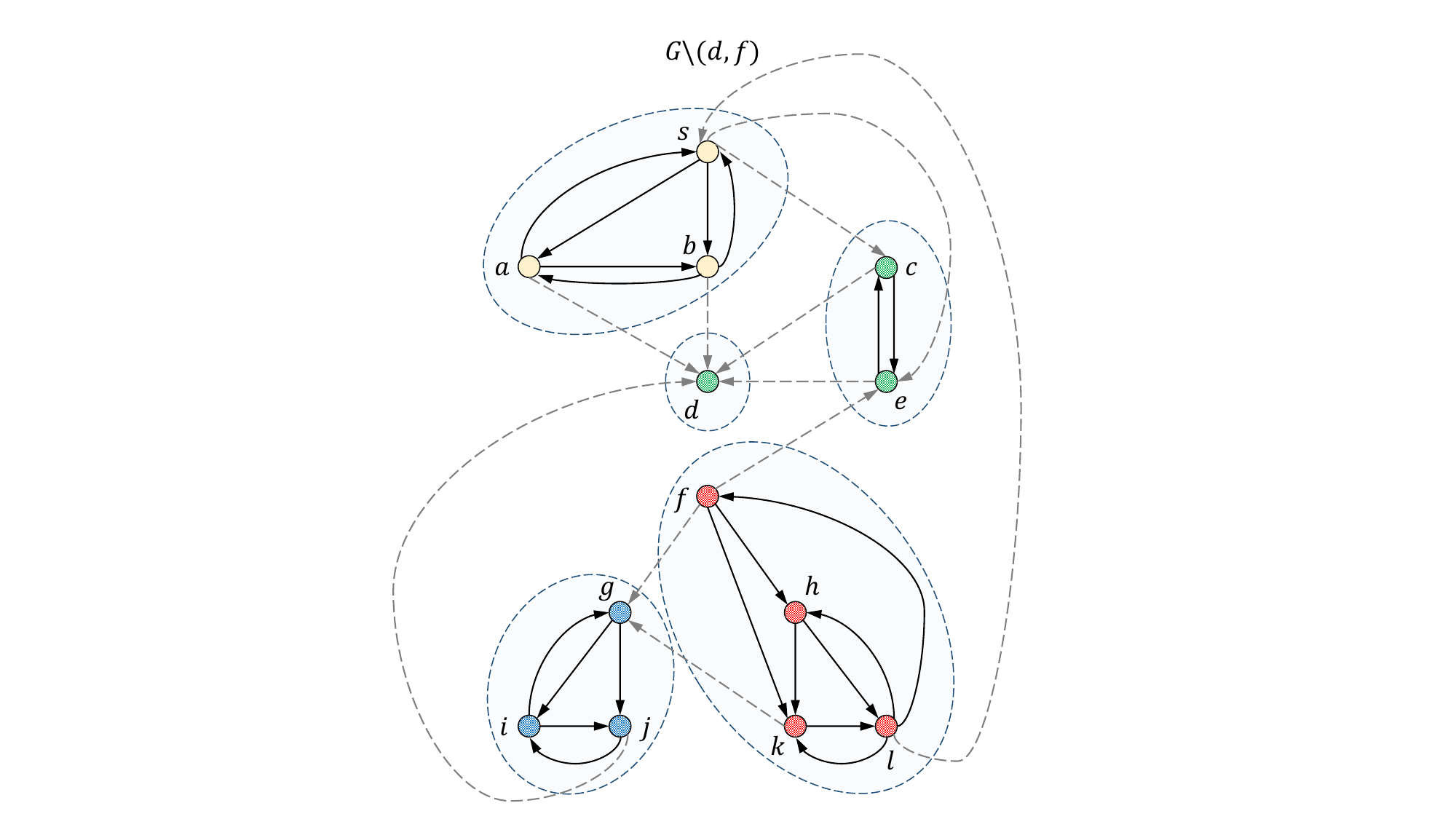}	}
\caption{The strongly connected components of $G \setminus (d,f)$ where $G$ is the graph of Figure \ref{figure:example1ab}.}
\label{figure:example1f}
\end{center}
\end{figure}

\subsection{Counting the  number of strongly connected components of $G\setminus e$}
\label{sec:SCCs-num}

In this section we consider the problem of computing the total number of strongly connected components obtained after the deletion of a single edge in a strongly connected digraph $G=(V,E)$.
In particular, we describe a data structure which, after $O(m)$-time preprocessing, is able to  answer the following aggregate query in worst-case time $O(n)$: ``Find the total number of  strongly connected components in $G \setminus e$, for all edges $e$.''
This provides a linear-time algorithm for computing the total number of strongly connected components obtained after the deletion of a single edge, for all edges in $G$.
Note that we need to answer this query only for edges that are strong bridges in $G$; indeed if $e$ is not a strong bridge, then $G\setminus e$ has exactly the same strongly connected components as $G$.
Our bound is tight, and it improves sharply over the naive $O(mn)$ solution, which computes from scratch the strongly connected components of $G \setminus e$ for each strong bridge $e$ of $G$.

Let $S\subseteq V$ be a subset of vertices of $G$. We denote by $\#SCC(S)$ the number of the strongly connected components in the subgraph of $G$ induced by the vertices in $S$.
Also, for an edge $e$ of $G$, we denote by $\#SCC_e(V)$ the number of the strongly connected components in $G \setminus e$. Our goal is to compute $\#SCC_e(V)$ for every strong bridge $e$ in $G$.
Theorem \ref{cor:scc} yields immediately the following corollary.

\begin{corollary}
\label{cor:nscc}
Let $e=(u,v)$ be a strong bridge of $G$.
Then one of the following cases holds:
\begin{itemize}
\item[(a)] If $e$ is a bridge in $G_s$ but not in $G_s^R$ then $\#SCC_e(V) = \#SCC(D(v))+1$.

\item[(b)] If $e$ is a bridge in $G_s^R$ but not in $G_s$ then $\#SCC_e(V) = \#SCC(D^R(u))+1$.

\item[(c)] If $e$ is a common bridge of $G_s$ and $G_s^R$ then
$\#SCC_e(V) = \#SCC(D(v)\cup D^R(u)) + 1 = \#SCC(D(v)) + \#SCC(D^R(u)) - \#SCC(D(v)\cap D^R(u))+1$.
\end{itemize}

\noindent
Moreover, let $C$ be a strongly connected component of $G \setminus e$. If  $C \subseteq D(v)$ (resp., $C \subseteq D^R(u)$) then $C=H(w)$ (resp., $C=H^R(w)$) where $w$ is a vertex in $D(v)$ (resp., $D^R(u)$) such that $h(w) \not \in D(v)$ (resp.,  $h^R(w) \not \in D^R(u)$).
\end{corollary}

Corollary \ref{cor:nscc}
shows that in order to compute the number of strongly connected components in $G\setminus e$ it is enough to compute $\#SCC(D(v))$,  $\#SCC(D^R(u))$ and $\#SCC(D(v)\cap D^R(u))$.
We will first show how to compute $\#SCC(D(v))$ and $\#SCC(D^R(u))$. Next, we will present an algorithm for computing $\#SCC(D(v)\cap D^R(u))$.

\paragraph{Computing $\#SCC(D(v))$ and $\#SCC(D^R(u))$.}
As suggested by Corollary \ref{cor:nscc},
 we can compute
$\#SCC(D(v))$ (resp., $\#SCC(D^R(u))$) by
counting the number of distinct vertices $w$ in $D(v)$ (resp., in $D^R(u)$) for which $h(w)\not \in D(v)$ (resp.,  $h^R(w) \not \in D^R(u)$).
We do this with the help of the bridge decomposition $\mathcal{D}$ (resp., $\mathcal{D}^R$) of $D$ (resp., $D^R$) defined in Section \ref{sec:dominators}. We recall that we denote by $D_x$ (resp., $D_x^R$) the tree in $\mathcal{D}$ (resp., $\mathcal{D}^R$) containing vertex $x$, and by $r_x$ (resp., $r^R_x$) the root of the tree $D_x$ (resp., $D_x^R$).

To compute $\#SCC(D(v))$,
we maintain a counter $SCC_e$ for each bridge $e=(u,v)$ in $G_s$. $SCC_e$ counts the number of  strongly connected components in $G \setminus e$ that are subsets of $D(v)$ encountered so far. Rather than processing the bridges of $G_s$ one at the time, we update simultaneously all counters $SCC_e$ for all bridges $e$ while visiting the bridge decomposition of the dominator tree $D$ in a bottom-up fashion.
To update the counters $SCC_e$, we exploit
the loop nesting tree $H$:
we increment the counter $SCC_e$ of a bridge $e=(u,v)$ in $G_s$ whenever we find a vertex $w$ in $D(v)$ such that $h(w) \not \in D(v)$, since $H(w)\subseteq D(v)$ is a strongly connected component of $G\setminus e$ by Corollary \ref{cor:nscc}.
After the bridge decomposition of $D$ has been processed, for each strong bridge $e=(u,v)$ in $G$
we have that $SCC_e=\#SCC(D(v))$, i.e.,
the counter $SCC_e$
stores exactly the number of strongly connected components of $G\setminus e$ containing only vertices in $D(v)$. In particular, by Corollary \ref{cor:nscc}(a) we can compute $\#SCC_e(V) = SCC_e + 1$ for each strong bridge $e=(u,v)$ which is a bridge in $G_s$ but not in $G_s^R$.

We can compute $\#SCC(D^R(u))$ in a similar fashion. We maintain a counter $SCC^R_e$ for each bridge $e=(u,v)$ in $G^R_s$.
$SCC^R_e$ counts the number of  strongly connected components in $G \setminus e$ that are subsets of $D^R(u)$ encountered so far.
We visit in a bottom-up fashion
the bridge decomposition of $D^R$ with the help of the loop nesting tree $H^R$:
we increment the counter $SCC_e^R$ of a bridge $e=(u,v)$ in $G_s^R$ whenever we find a vertex $w$ in $D^R(u)$ such that $h^R(w) \not \in D^R(u)$, since $H^R(w)\subseteq D^R(u)$ is a strongly connected component of $G\setminus e$ by Corollary \ref{cor:nscc}. At the end of this visit,
for each strong bridge $e=(u,v)$ in $G$
we have that $SCC_e^R=\#SCC(D^R(u))$.
By Corollary \ref{cor:nscc}(b) we can now compute $\#SCC_e(V) = SCC_e^R + 1$ for each strong bridge $e$ which is a bridge in $G_s^R$ but not in $G_s$.

Note that if $e=(u,v)$ is a common bridge of $G_s$ and $G_s^R$, then by Corollary \ref{cor:nscc}(c) we have that $\#SCC_e(V) = SCC_e + SCC_e^R - \#SCC(D(v) \cap D^R(u))+1$. Thus, to complete the description of our algorithm  we have still to show how to compute $\#SCC(D(v)\cap D^R(u))$,
i.e., the number of strongly connected components in the subgraph induced by the vertices in $D(v)\cap D^R(u)$.
We will deal with this issue later.

As previously mentioned, a crucial task for updating the counter $SCC_e$ (resp., $SCC_e^R$) for each strong bridge $e=(u,v)$ is to check for a
vertex $w$ in $D(v)$ (resp., in $D^R(u)$) such that $h(w)\not \in D(v)$ (resp.,  $h^R(w) \not \in D^R(u)$). An efficient method to perform this test hinges on the following lemma.

\begin{lemma}
For each vertex $w \not=s$ in $G$, the following holds: 
\vspace{-.25cm}
\begin{itemize}
\item[(a)]
$H(w)$ (resp., $H^R(w)$) induces a strongly connected component of $G\setminus e$, for any bridge $e$ of $G_s$ (resp., $G_s^R$) in the path $D[r_{h(w)},w]$ (resp., $D^R[r_{h^R(w)}^R,w]$)
from $r_{h(w)}$ to $w$ in $D$ (resp., from $r_{h^R(w)}^R$ to $w$ in $D^R$).
\vspace{-.25cm}
\item[(b)]
$H(w)$ (resp., $H^R(w)$) does not induce a strongly connected component of $G \setminus e$, for any bridge $e$ of $G_s$ (resp., $G_s^R$) in the path $D[s,r_{h(w)}]$ (resp., $D^R[s,r_{h^R(w)}^R]$) from $s$ to $r_{h(w)}$ in $D$ (resp., from $s$ to $r_{h^R(w)}^R$ in $D^R$).
\end{itemize}
\label{lemma:number-of-SCCs-ancestors}
\end{lemma}
\begin{proof}
We only prove the lemma for bridges of $G_s$, as the case for bridges of $G_s^R$ is completely analogous.
Let $e=(u,v)$ be any bridge of $G_s$ in the path $D[r_{h(w)},w]$: since
 $w \in D(v)$ and $h(w) \not \in D(v)$, (a) follows immediately from Lemma~\ref{lemma:subtree}.
Now we turn to (b). Let $e=(u,v)$ be any bridge of $G_s$ in the path $D[s,r_{h(w)}]$, and let $z$ be the nearest ancestor of $w$ in the loop nesting tree $H$ such that $z\in D(v)$ and $h(z) \not \in D(v)$.
Note that $z \not= w$ since $h(w)$ is a descendant of $v$ in $D$. Hence $z$ is a proper ancestor of $w$ in $H$.
By Theorem
\ref{cor:scc}, $H(z)$ induces a strongly connected component of $G\setminus e$.
Since $z$ is a proper ancestor of $w$ in $H$, then $H(w) \subset H(z)$.
As a consequence, $H(w)$ does not induce a maximal strongly connected subgraph in $G\setminus e$ and thus it cannot induce a strongly connected component of $G\setminus e$.
\end{proof}

We are now ready to describe our algorithm. We do not maintain the counters $SCC_e$ and $SCC_e^R$ for each strong bridge $e$ explicitly. Instead, for the sake of efficiency, we distribute this information along some suitably chosen vertices in the dominator trees $D$ and $D^R$.
We first compute
a compressed tree $\widehat{D}$ of the dominator tree $D$ as follows.
Let $x$ be any vertex of $G$, and let $D_x$ be the tree cointaining $x$ in the bridge decomposition $\mathcal{D}$: then $\widehat{D}$ is obtained by contracting all the vertices of $D_x$ into its tree root $r_x$.
Let $x$ be a vertex in the dominator tree $D$ such that
$h(x) \not \in D(r_x)$: then by Lemma \ref{lemma:ancestor-of-x} $r_{h(x)}$ is an ancestor of $x$ in $D$ and by Lemma~\ref{lemma:number-of-SCCs-ancestors}
we need to increment the counter of all bridges that lie in the path $D[r_{h(x)},x]$
from $r_{h(x)}$ to $x$ in $D$.
By construction, those bridges correspond exactly to all the edges in the path $\widehat{D}[r_{h(x)},r_x]$ from $r_{h(x)}$ to $r_x$ in the compressed tree $\widehat{D}$. We refer to such a  path
$\widehat{D}[r_{h(x)},r_x]$
as a \emph{bundle
starting from vertex $r_{h(x)}$ and ending at vertex $r_x$}.
For each bridge $e$ in the bundle  starting from $r_{h(x)}$ and ending at $r_x$, by Lemma~\ref{lemma:number-of-SCCs-ancestors}(a) there is a strongly connected component $H(x)$ in $G \setminus e$, where $x$ is in the tree $D_x$ rooted at $r_x$ of the bridge decomposition $\mathcal{D}$.

For each vertex $z \in \widehat{D}$ we maintain a value $bundle(z)$ which stores the number of bundles that end at $z$, minus the number of bundles that start at $z$. Note that for each vertex $z \in \widehat{D}$,
$\sum_{y\in \widehat{D}(z)} bundle(y)$ equals the total number of bundles that start from a proper ancestor of $z$ in $\widehat{D}$ and end at a descendant of $z$ in $\widehat{D}$.
\begin{lemma}
\label{lemma:bundle}
Let $(u,v)$ be a bridge of $G_s$.
Then $(u,v)$ corresponds to the edge $(r_u,v)$ in $\widehat{D}$ and $\#SCC(D(v))=\sum_{y\in \widehat{D}(v)} bundle(y)$.
\end{lemma}
\begin{proof}
Since $(u,v)$ is a bridge of $G_s$, $v$ is a root in the bridge decomposition of $D$.
So the fact that $(u,v)$ corresponds to the edge $(r_u,v)$ in $\widehat{D}$ follows from the construction of $\widehat{D}$.
Now we prove the second part of the lemma.
Each vertex $w$ of $G$ such that $h(w) \not \in D(v)$ contributes a $+1$ to $bundle(r_w)$ and a $-1$ to $bundle(r_{h(w)})$.
We consider the net contribution
of each vertex $w \in D(v)$ to $\#SCC(D(v))$, since the contribution of all other vertices is zero.
Recall that $r_{h(w)}$ is an ancestor of $w$ in $D$ by Lemma \ref{lemma:ancestor-of-x}.
So, for any $w \in D(v)$, we have that $v$ is in the bundle $D[r_{h(w)}, r_w]$ if and only if $h(w) \not\in D(v)$.
If $h(w) \not\in D(v)$, $v$ is a proper descendant of $r_{h(w)}$ and the net contribution of $w$ to $\#SCC(D(v))$ is $+1$.
Otherwise, if  $h(w) \in D(v)$, the net contribution of $w$ to $\#SCC(D(v))$ is zero.
Both cases are handled correctly, since by Corollary \ref{cor:nscc}(a), $\#SCC(D(v))$ is equal to the number of vertices
$w \in D(v)$ such that $h(w) \not\in D(v)$.
\end{proof}

We next show how to compute $bundle(z)$ for each vertex $z \in \widehat{D}$.
We process all vertices in $G$ in any order. Whenever we find a vertex $x$ in the digraph $G$ such that $h(x) \not \in D(r_{x})$,
we increment $bundle(r_x)$ and decrement $bundle(r_{h(x)})$. Indeed, by Lemma \ref{lemma:number-of-SCCs-ancestors}(a) there is a bundle
$\widehat{D}[r_{h(x)},r_x]$
starting
from $r_{h(x)}$ and ending at $r_x$:
for each bridge $e$ in
$\widehat{D}[r_{h(x)},r_x]$, $H(x)$
induces a strongly connected component  of $G\setminus e$.

Once $bundle(z)$ is computed for each vertex $z$ in $\widehat{D}$, we can compute for each bridge $(u,v)$ in $G_s$ the value $\#SCC(D(v))$.
Recall that each vertex $z\in\widehat{D}$ is the root of a tree in the bridge decomposition $\mathcal{D}$ and $(d(z),z)$ is a bridge in $G_s$.
From Lemma \ref{lemma:bundle}, we have that $\#SCC(D(v))=bundle(v) + \sum_{y}\#SCC(D(y))$, where the sum is taken for
all children $y$ of  $v$ in $\widehat{D}$.
So we can compute the $\#SCC(D(v))$ values by visiting the compressed tree $\widehat{D}$ in a bottom-up fashion as follows.
For each vertex $z$ visited in $\widehat{D}$, we set $\#SCC(D(z))=\#SCC(D(z))+bundle(z)$, and we increment  $\#SCC(D(r_{d(z)}))$ by the value $\#SCC(D(z))$.

We can compute $\#SCC(D^R(u))$ for each bridge $(u,v)$ in $G_s^R$ in a completely analogous fashion.
The pseudocode of the algorithm is provided below (see Algorithm \ref{alg:SCCsDescendants}).

\begin{algorithm}
	\LinesNumbered
	\DontPrintSemicolon
	\KwIn{Strongly connected digraph $G=(V,E)$}
	\KwOut{For each strong bridge $(u,v)$ the numbers $\#SCC(D(v))$ and $\#SCC(D^R(u))$}
	Compute the reverse digraph $G^R$.
Select an arbitrary start vertex $s \in V$.

	Compute the dominator trees $D$ and $D^R$ of the flow graphs $G_s$ and $G_s^R$, respectively.\;
	
	Compute the loop nesting trees $H$ and $H^R$ of the flow graphs $G_s$ and $G_s^R$, respectively.\;
	
		Compute $\mathcal{D}$ (the bridge decomposition of $D$) and $\mathcal{D}^R$ (the bridge decomposition of $D^R$).\;
		
		Compute the compressed trees $\widehat{D}$ and $\widehat{D}^R$ of the dominator trees $D$ and $D^R$, respectively.\;

	Initialize $bundle(x)=0$ for each $x\in \widehat{D}$ and $bundle^R(x)=0$ for each $x\in \widehat{D}^R(x)$.\;
	\ForEach{vertex $x \in V$}
	{	
		\If{$h(x)\not \in D(r_x)$}{
			Find the roots $r_x$ and $r_{h(x)}$ in the bridge decomposition $\mathcal{D}$\;
			Set $bundle(r_x) = bundle(r_x)+1$ and $bundle(r_{h(x)}) = bundle(r_{h(x)})-1$\;
		}
		\If{$h^R(x)\not \in D^R(r^R_x)$}{
			Find the root $r^R_x$ and $r^R_{h^R(x)}$ in the bridge decomposition $\mathcal{D}^R$\;
			Set $bundle^R(r^R_x) = bundle^R(r^R_x)+1$ and $bundle^R(r^R_{h^R(x)}) = bundle^R(r^R_{h^R(x)})-1$\;
		}
		
	}
	
	\ForEach{vertex $z \in \widehat{D}$, $z\neq s$, in a bottom-up fashion}
	{
        Set $\#SCC(D(z))=\#SCC(D(z)) + bundle(z)$\;
		Set $\#SCC(D(r_{d(z)})) = \#SCC(D(r_{d(z)})) + \#SCC(D(z))$\;
	}
	\ForEach{vertex $z \in \widehat{D}^R$, $z\neq s$, in a bottom-up fashion}
	{
		 Set $\#SCC(D^R(z))=\#SCC(D^R(z)) + bundle^R(z)$\;
		 Set $\#SCC(D^R(r^R_{d^R(z)})) = \#SCC(D^R(r^R_{d^R(z)})) + \#SCC(D^R(z))$\;
	}
	\caption{\textsf{SCCsDescendants}}
	\label{alg:SCCsDescendants}
\end{algorithm}

\begin{lemma}
	Given the dominator trees $D$ and $D^R$, the loop nesting trees $H$ and $H^R$, and the strong bridges of $G$, Algorithm \textsf{SCCsDescendants} runs in $O(n)$ time.
\end{lemma}
\begin{proof}
	Since each dominator tree has $n-1$ edges, we can locate the bridges of $G_s$ and $G_s^R$, and construct the bridge decomposition of $D$ and $D^R$ and the compressed trees $\widehat{D}$ and $\widehat{D}^R$ in $O(n)$ time.
	During the first \textbf{foreach} loop of Algorithm \textsf{SCCsDescendants} (Lines 7-16),  we visit all vertices of $G$ and compute bundles for vertices in $\widehat{D}$ and $\widehat{D}^R$, by performing  $O(1)$ computations for each vertex: indeed, we can test parent/descendant relationships in $D$ and $D^R$ in constant time \cite{domin:tarjan}.
	The other two \textbf{foreach} loops (Lines 17--20 and 21--24) visit all vertices in $\widehat{D}$ and $\widehat{D}^R$, performing again $O(1)$ time computations per vertex.
\end{proof}

\paragraph{Computing $\#SCC(D(v) \cap D^R(u))$.}
To complete the description of our algorithm,
we still need to describe how to compute the quantity
$\#SCC(D(v) \cap D^R(u))$ for each common bridge $e=(u,v)$ of $G_s$ and $G_s^R$. Before doing that, we need to introduce some new terminology.
Let $e=(u,v)$ be a common bridge.
We say that a vertex $w$ is a \emph{common descendant} of $(u,v)$ if $w \in D(v)$ and $w \in D^R(u)$; in this case we also say that $(u,v)$ is a \emph{common bridge ancestor} of $w$.
Previously, we have been
working with the bridge decomposition ${\mathcal{D}}$ and ${\mathcal{D}}^R$ of the dominator trees $D$ and $D^R$, as defined in Section \ref{sec:dominators}. Since we need to deal now with the common bridges of $G_s$ and $G_s^R$, we define a coarser partition of $D$ and $D^R$, as follows.
After deleting all the common bridges from the dominator trees $D$ and $D^R$, we obtain the \emph{common bridge decomposition} of $D$ and $D^R$ into forests $\breve{\mathcal{D}}$ and $\breve{\mathcal{D}}^R$ (see Figure \ref{figure:CommonBridgeDecomposition}).
We denote by $\breve{D}_u$ (resp., $\breve{D}^R_u$) the tree in $\breve{\mathcal{D}}$ (resp., $\breve{\mathcal{D}}^R$) that contains vertex $u$, and by $\breve{r}_u$ (resp., $\breve{r}^R_u$) the root of $\breve{D}_u$ (resp., $\breve{D}^R_u$). Lemma \ref{lemma:ancestor-of-u} extends Lemma \ref{lemma:ancestor-of-x}.

\begin{figure}[t!]
\begin{center}
\centerline{	\includegraphics[trim={0 0 0 4cm}, clip=true, width=1.3\textwidth]{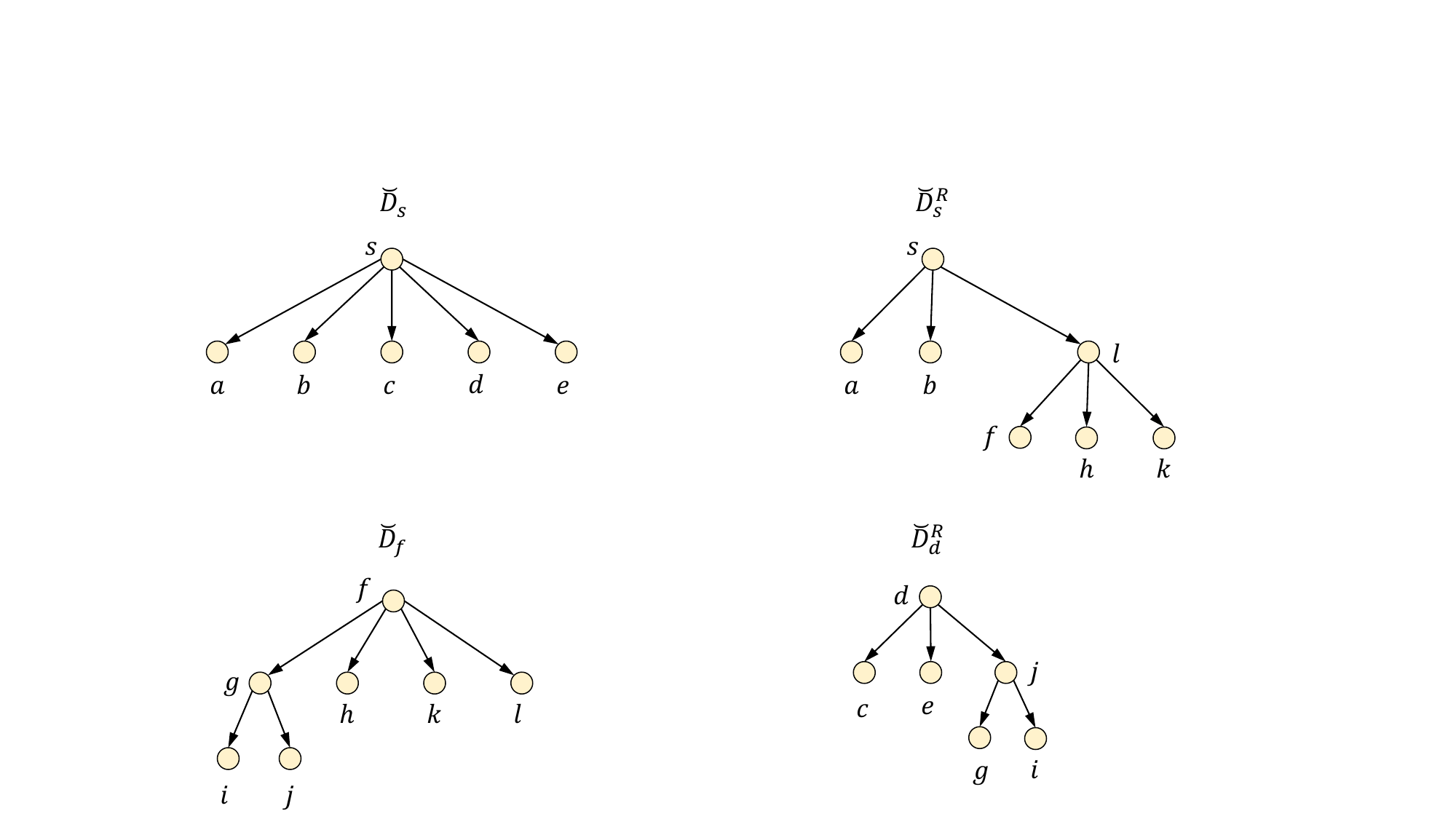}}	
	\caption{The common bridge decomposition of the dominator trees $D$ and $D^R$ of Figure \ref{figure:example1ab}. In this example the only common bridge is edge $(d,f)$.}
	\label{figure:CommonBridgeDecomposition}
\end{center}
\end{figure}

\begin{lemma}
Let $x\not= s$ be a vertex in $G$. 
Then $\breve{r}_{h(x)}$ is an ancestor of $x$ in $D$.
\label{lemma:ancestor-of-u}
\end{lemma}
\begin{proof}
Assume by contradiction that $\breve{r}_{h(x)}$ is not an ancestor of $x$ in $D$.
Let $T$ be the dfs tree based on which $H$ was built.
By the definition of dominators all paths from $s$ to $h(x)$ contain $\breve{r}_{h(x)}$, and therefore $\breve{r}_{h(x)}$ is an ancestor of $h(x)$ in $T$.
By Lemma \ref{lemma:partition-paths}(a) all paths from $x$ to $h(x)$ in $G$ contain the strong bridge $(d(\breve{r}_{h(x)}),\breve{r}_{h(x)})$. But this is a contradiction to the fact that all paths from $x$ to $h(x)$ contain only descendants of $h(x)$ in $T$, since $\breve{r}_{h(x)}$ is an ancestor of $h(x)$ in $T$.
\end{proof}

Let $e$ and $e'$ be two common bridges. If $e$ is an ancestor of $e'$ in $D$ (resp., $D^R$) then we use the notation $e \csbdom e'$ (resp., $e \csbrdom e'$) to denote this fact.

\begin{lemma}
	\label{lemma:csb-dom}
	Let $e$, $e'$, and $e''$ be distinct common bridges such that $e \csbdom e'' \csbdom e'$ and $e' \csbrdom e$. Then $e' \csbrdom e'' \csbrdom e$.
\end{lemma}
\begin{proof}
	Let $e = (u,v)$ and $e'=(w,z)$. The fact that $e \csbdom e'' \csbdom e'$ implies that all paths from $u$ to $z$ in $G$ contain $e$, $e''$ and $e'$ in that order. If $e''$ is not in the path from $z$ to $u$ in $D^R$ then there is path from $u$ to $z$ in $G$ that avoids $e''$, a contradiction.
\end{proof}

\begin{figure}[t!]
	\begin{center}
\centerline{		\includegraphics[trim={0 0 0 7cm}, clip=true, width=1.3\textwidth]{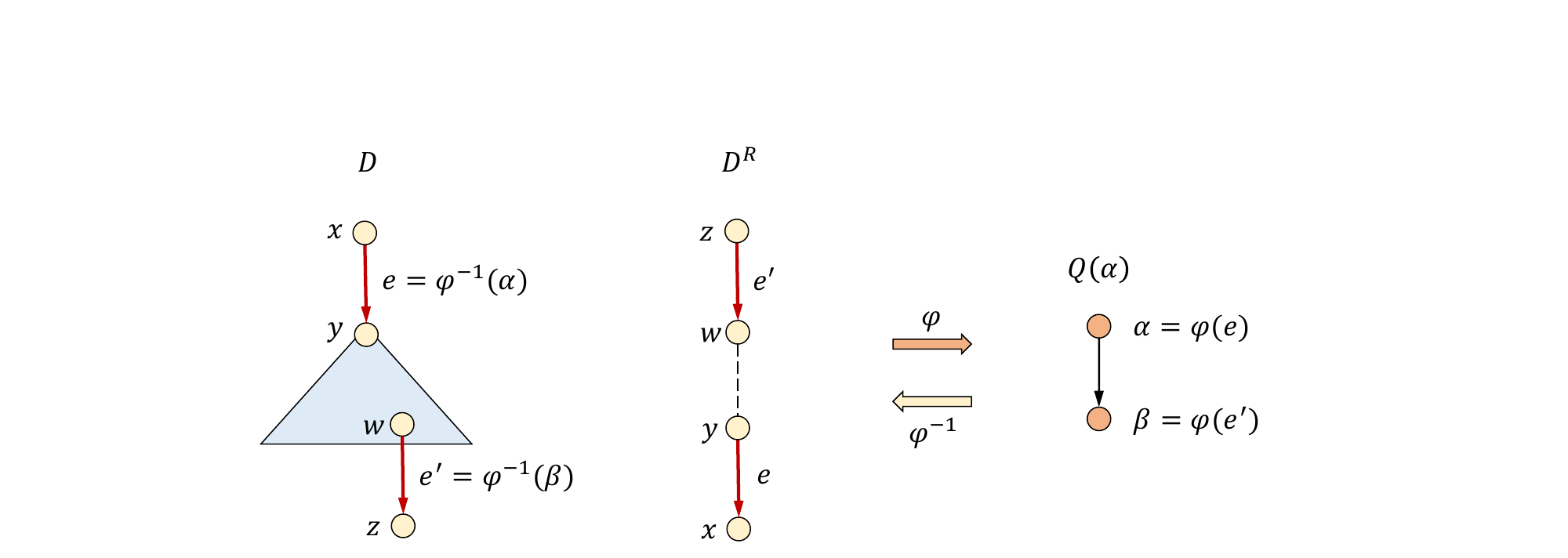}}
		\caption{An illustration of the definition of the common bridge forest  $\mathcal{Q}$.}
		\label{figure:commonDescendants}
	\end{center}
\end{figure}

We will use Lemma \ref{lemma:csb-dom}
to identify the common bridge ancestors of each vertex from a sequence of common strong bridges $e_1 \csbdom e_2 \csbdom \ldots \csbdom e_{\ell}$.
In order to have a compact representation of the relations in Lemma \ref{lemma:csb-dom} we define the \emph{common bridge forest} $\mathcal{Q}$ as follows (see Figure \ref{figure:commonDescendants}).
Forest $\mathcal{Q}$ contains a node $\varphi(e)$ for each common bridge $e$. We also define the reverse map $\varphi^{-1}$ from the nodes of $\mathcal{Q}$ to the common bridges of $G$, i.e., for any node  $\alpha$ of $\mathcal{Q}$, $\varphi^{-1}(\alpha)$ is the corresponding common bridge represented by $\alpha$.
Let $\alpha$ and $\beta$ be two distinct nodes of $\mathcal{Q}$.
Let $\varphi^{-1}(\alpha) = (x,y)$ and $\varphi^{-1}(\beta) = (w,z)$. Then $\mathcal{Q}$ contains the edge $(\alpha, \beta)$ if and only if $\breve{D}_y=\breve{D}_w$ and $\varphi^{-1}(\beta) \csbrdom \varphi^{-1}(\alpha)$. That is, there is an edge from node $\varphi(e)$ to $\varphi(e')$, where $e=(x,y)$ and $e'=(w,z)$ are common bridges, if $e$ is the common bridge that enters $\breve{D}_y=\breve{D}_w$ and $e'$ is an ancestor of $e$ in $D^R$.
To see that $\mathcal{Q}$ is indeed a forest, consider the tree $D'$ obtained from $D$ by contracting each subtree $\breve{D} \in \breve{\mathcal{D}}$ into its root.
Then, the edges of $D'$ are the common bridges. Now note that $\mathcal{Q}$ contains a node for each common bridge, and two nodes can be adjacent only if the corresponding common bridges are adjacent in $D'$.
We use the notation $Q(\alpha)$ to denote the tree in the common bridge forest $\mathcal{Q}$ that contains node $\alpha$.

\begin{lemma}
\label{lemma:common-descendants-path}
Let $\mathcal{Q}$ be the common bridge forest of $G$.
Suppose that vertex $x$ is a common descendant of a common bridge $e$ and let $e'$ be the nearest common bridge that is an ancestor of $x$ in $D$ such that $\varphi(e')\in Q(\varphi(e))$. Let $\pi= \left < \varphi(e) = \alpha_1, \alpha_2, \ldots, \alpha_{\ell} = \varphi(e') \right >$ be the path from $\varphi(e)$ to $\varphi(e')$ in $Q(\varphi(e))$.
Then $x$ is a common descendant of every bridge $\varphi^{-1}(\alpha_i)$,  $1\leq i \leq \ell$.
\end{lemma}
\begin{proof}
All common bridges that are represented by nodes in $\pi$ are ancestors of $x$ in $D$, since for each edge $(\alpha_i, \alpha_{i+1})$ in $\mathcal{Q}$ we have that $\varphi^{-1}(\alpha_i) \csbdom \varphi^{-1}(\alpha_{i+1})$.
Moreover, for each edge $(\alpha_i, \alpha_{i+1})$ we have $\varphi^{-1}(\alpha_{i+1}) \csbrdom \varphi^{-1}(\alpha_i)$: by the fact that $e=\varphi^{-1}(\alpha_1)$ is an ancestor of $x$ in $D^R$ (since we assumed that $e$ is a common bridge ancestor of $x$), it follows that all common bridges that are represented by nodes in $\pi$ are ancestors of $x$ in $D^R$, and thus are common bridge ancestors of $x$.
\end{proof}

To compute the number $\#SCC(D(v) \cap D^R(u))$ for each common bridge $e=(u,v)$ we follow an approach similar to Algorithm {\textsf{SCCsDescendants}}.
Namely, instead of computing the quantities $\#SCC(D(v)\cap D^R(u))$ one at the time
for each common bridge $e=(u,v)$, we find for each vertex $w$ its common bridge ancestors  $e=(u,v)$ such that $H(w)$ induces a strongly connected component in $G\setminus e$ and we increment the counter of $\#SCC(D(v) \cap D^R(u))$ for all those common bridges $e=(u,v)$.
The following lemma, combined with Lemma \ref{lemma:common-descendants-path}, 
allows us to accomplish this task efficiently.

\begin{lemma}
Let $e_1, e_2, \ldots, e_k$ be the common bridges in the path
$D[\breve{r}_{h(x)},x]$
from vertex $\breve{r}_{h(x)}$ to vertex $x$ in $D$, in order of appearance in that path. If $e_1$ is not a common bridge ancestor of $x$ then no $e_j$, $1 < j \le k$, is. Otherwise, let $j$ be the largest index $1 \leq j \leq k$ such that vertex $x$ is a common descendant of $e_j$. Then $e_{j} \csbrdom e_{j-1} \csbrdom \ldots \csbrdom e_1$ and $x$ is a common descendant of every common bridge $e_i$ for $1 \leq i \le j$.
\label{lemma:common-strong-bridge-relation}
\end{lemma}
\begin{proof}

It suffices to show that if a common strong bridge $e_{\ell}$, for any $1 < {\ell} \leq j$, is a common ancestor of $x$ then $e_{\ell} \csbrdom e_{{\ell}-1}$ and $e_{{\ell}-1}$ is a common ancestor of $x$.
Let $e_{\ell-1} = (u,u')$ and $e_{\ell} = (v,v')$.
Assume by contradiction that the above statement is not true,
 i.e., either $e_{\ell}$ is not an ancestor of $e_{\ell-1}$ in $D^R$ or $e_{{\ell}-1}$ is not a common ancestor of $x$.
Then, in either case there must be a path $\pi$ in $G$ from $x$ to $v$ avoiding $e_{{\ell}-1}$.
If $\pi$ contained a vertex $x' \not\in D(u')$, then the subpath of $\pi$ from $x'$ to $v$ would include $e_{{\ell}-1}$,
due to the fact that $u$ is an ancestor of $x$ in $D$ and $v \in D(u')$. Thus, all vertices in $\pi$ are descendants of $u'$ in $D$.
Let $T$ be the dfs tree that generated $H$, and let $\mathit{pre}$ be the preorder numbering in $T$.
We claim that there is a vertex $w$ in $\pi$ such that all vertices in $\pi$ are descendants of $w$ in $T$.
The claim implies that $x \in H(w)$, so $h(x)$ is a descendant of $u'$ in $D$.
But this contradicts the fact that $h(x)$ is not a descendant of $u'$. Hence, the lemma will follow.

To prove the claim, choose $w$ to be the vertex in $\pi$ such that $pre(w)$ is minimum.
Then Lemma \ref{lemma:dfs} implies that $w$ is an ancestor of $v$ in $T$.
Let $z$ be any vertex in $\pi$. We argue that $pre(w) \leq pre(z) < pre(w)+|T(w)|$, hence $z$ is a descendant of $w$ in $\pi$.
By the choice of $w$ we have  $pre(w) \le pre(z)$, so it remains to prove the second inequality.
Suppose $pre(z) \geq pre(w)+|T(w)|$. Since $x$ is descendant of $v$ in $D$ it is also a descendant of $v$ in $T$.
So $pre(x) < pre(z)$.
By Lemma \ref{lemma:dfs}, path $\pi$ contains a common ancestor $q$ of $x$ and $z$ in $T$.
Vertex $q$ is an ancestor of $w$ in $T$, since $x \in T(w)$ and $z \not \in T(w)$.
But then $pre(q)<pre(w)$, which contradicts the choice of $w$.

Therefore, any path $\pi$ from $x$ to $v$ must contain the common bridge $e_{\ell-1}$.
We conclude that $e_{\ell} \csbrdom e_{\ell-1}$.
\end{proof}

We are now ready to describe our algorithm for computing the number $\#SCC(D(v) \cap D^R(u))$ for each common bridge $e=(u,v)$.
That is, the number of strongly connected components in the subgraph induced by $D(v) \cap D^R(u)$, which, by Corollary \ref{cor:nscc}, is equal to
the number of strongly connected components in $G\setminus e$ that are composed only by common descendants of the common bridge $e$.
Throughout the algorithm, we maintain a list of pairs $L$ of the form $\langle x,y\rangle$, where $x$ and $y$ are vertices in $G$: the pair $\langle x,y\rangle$ will correspond to a path in the common bridge forest $\mathcal{Q}$, pinpointed by $x$ and $y$, on which we wish to identify the common bridge ancestors of vertex $y$.
The common bridge ancestors of any vertex $z$ in $G$ will be computed with the help of Lemma~\ref{lemma:common-strong-bridge-relation}.
Specifically, we use $\mathcal{Q}$ to locate the endpoints of each sequence of common bridges that satisfy Lemma~\ref{lemma:common-descendants-path}, and will propagate this information to the whole sequence in a second phase.
We note that, as a special case, such a sequence may consist of only one common bridge.
To accomplish this task, we will maintain for each node $\alpha$ of $\mathcal{Q}$
(corresponding to common bridge  $\varphi^{-1}(\alpha)$ of $G_s$ and $G_s^R$) two values, denoted respectively by $start(\alpha)$ and $end(\alpha)$, defined as follows.
The value $start(\alpha)$ (resp., $end(\alpha)$)
stores the number of strongly connected components that contain only common descendants of a bridge contained in a sequence (or sequences) of common bridges starting (resp., ending) at bridge $\varphi^{-1}(\alpha)$.

The pseudocode of the algorithm is given below (see Algorithm \ref{alg:SCCsCommonDescendants}).
We first consider all vertices of $G$
one at the time (Lines 9--15).
For each vertex $x$ such that $h(x)\not \in \breve{D}_{\breve{r}_{x}}$ we add the pair $\langle h(x),x\rangle$ to the list $L$ to indicate that for each common bridge $e$ that is both in the path $D[\breve{r}_{h(x)},x]$ from $\breve{r}_{h(x)}$ to $x$ in the dominator tree $D$ and it is a common ancestor of $x$,
the set of vertices $H(x)$ induces a strongly connected component in $G \setminus e$.
By Theorem \ref{cor:scc}(c) all vertices in the strongly connected component $H(x)$ in $G \setminus e$
are also common descendants of $e$, i.e., $H(x) \subseteq D(v)\cap D^R(u)$.
After traversing all the vertices, we process the pairs that were inserted in the list $L$.
Note that there can be at most $n-1$ pairs in $L$, since
at most one pair is inserted for each vertex $x \not= s$, and no pair is inserted for $x=s$.
Furthermore, by construction, all the pairs are of the form $\langle h(x),x\rangle$.
So let $\langle h(x),x\rangle$ be a pair extracted from $L$.
We first identify, the common bridge $(w,z)$ such that $\breve{D}_w=\breve{D}_{h(x)}$ and $w$ is an ancestor of $x$ in $D$.
If $e=(w,z)$ is a common ancestor of $x$, then by Lemma~\ref{lemma:common-descendants-path}, so are all common bridges that correspond to the nodes in the path $\pi$ from $\varphi(e)$ to $\varphi(e')$ in $Q(\varphi(e))$, where $e'=(y,d^R(y))$ is the last common bridge in the path from $w$ to $x$ in $D$ such that $\varphi(e') \in Q(\varphi(e))$.
That means that $H(x)$ is a strongly connected component in $G \setminus e''$ for each common bridge $e''$ such that $\varphi(e'') \in \pi$.
To account for this, we increment $end(\varphi(e'))$ and $start(\varphi(e))$ accordingly.
In the final loop (Lines 24--33), we visit all trees in the common bridge forest $\mathcal{Q}$.
We process the nodes of each tree $Q$ of the forest $\mathcal{Q}$ in a bottom-up fashion.
When we visit a node $\alpha \in Q$ we compute $\#SCC(D(v)\cap D^R(u))$, where $e = (u,v)= \varphi^{-1}(\alpha)$, as follows. We first initialize $\#SCC(D(v)\cap D^R(u))$ to $end(\alpha)$ (Line 27).
Next (Lines 29--30),
for each child $\beta$ of $\alpha$, where $e'= \varphi^{-1}(\beta) = (u',v')$, we increment
$\#SCC(D(v)\cap D^R(u))$ by the quantity
$ \#SCC(D(v')\cap D^R(u'))-start(\beta)$.

\begin{algorithm}
	\LinesNumbered
	\DontPrintSemicolon
	\KwIn{Strongly connected digraph $G=(V,E)$}
	\KwOut{For each strong bridge $e=(x,y)$ the number $\#SCC(D(y) \cap D^R(x))$}
	\textbf{Initialization:}\;
	Compute the reverse digraph $G^R$. Select an arbitrary start vertex $s \in V$.  \;
	Compute the dominator trees $D$ and $D^R$ of the flow graphs $G_s$ and $G_s^R$, respectively.\;
	Compute the loop nesting trees $H$ and $H^R$ of the flow graphs $G_s$ and $G_s^R$, respectively.\;
	Compute the common bridge decomposition forests $\breve{\mathcal{D}}$ and the common bridge forest $\mathcal{Q}$.\;
	
	Initialize an empty list of pairs $L$.\;
	 \lForAll{nodes $\alpha \in \mathcal{Q}$} {set $end(\alpha) = start(\alpha) = 0$}
	
	\textbf{Construct list of pairs:}\;
	\ForEach{tree $\breve{D}_r$ with root $r\neq s$}
	{
	 	\ForEach{$x\in \breve{D}_r$}
	 	{
	 		\If{$h(x) \not \in \breve{D}_r$}
	 		{
	 			$L = L \cup \langle h(x),x\rangle$
	 		}
		 }
	}
	\textbf{Process pairs:}\;
	\ForEach{pair $\langle x,y\rangle\in L$}
	{
		Let $e=(w,z)$ be the common bridge such that $w \in \breve{D}_x$ and $w$ is ancestor of $y$ in $D$\;
		\If{$w$ is an ancestor of $y$ in $D^R$}
		{
            Let $e'$ be nearest common bridge that is an ancestor of $y$ in $D$ with $\varphi(e') \in Q(\varphi(e))$\;
			Set $end(\varphi(e')) = end(\varphi(e'))+ 1$ and $start(\varphi(e)) =start(\varphi(e))+ 1$
		}
	}
	
	\ForEach{tree $Q$ in $\mathcal{Q}$}
	{
		\ForEach{node $\alpha \in Q$, in a bottom-up fashion}
		{
			Let $e = \varphi^{-1}(\alpha) = (u,v)$\;
            Set $\#SCC(D(v)\cap D^R(u))=end(\alpha)$\;
			\ForEach{edge $(\alpha,\beta)$ in $Q$}
			{
                Let $e' = \varphi^{-1}(\beta)=(u',v')$\;
                Set $\#SCC(D(v)\cap D^R(u)) = \#SCC(D(v)\cap D^R(u))+ \#SCC(D(v')\cap D^R(u'))-start(\beta)$\;
			}
		}
	}
	\caption{\textsf{SCCsCommonDescendants}}
	\label{alg:SCCsCommonDescendants}
\end{algorithm}

At the end of the algorithm, $\#SCC(D(v)\cap D^R(u))$ contains the number of the strongly connected components in $G \setminus e$ that include only common descendants  of the common bridge $e = (u,v)$, as shown by the following lemma.

\begin{lemma}
Algorithm \textsf{SCCsCommonDescendants} is correct.
\end{lemma}
\begin{proof}
For each vertex $x$ such that $h(x)\not \in \breve{D}_{x}$ the set $H(x)$ is a strongly connected component that contains only vertices that are common descendants of a strong bridge $e$ only if $e$ lies in the path $D[\breve{r}_{h(x)},x]$ and $e$ is a common bridge ancestor of $x$.
We prove that we correctly identify each such common bridge $e=(u,v)$ for each vertex $x$ and we account for $H(x)$ in $\#SCC(D(v) \cap D^R(u))$.
If $h(x) \not \in \breve{D}_{x}$ then by Lemma \ref{lemma:number-of-SCCs-ancestors}, $H(x)$ induces a strongly connected component in $G \setminus e'$ for all bridges $e'$ in $D[\breve{r}_{h(x)},x]$.
Therefore, to identify for which common bridges $e=(u,v)$ the set $H(x)$ induces a strongly connected component in $G \setminus e$ and $H(x) \subseteq D(v)\cap D^R(u)$, it is sufficient to check which of the common bridges in $D[\breve{r}_{h(x)},x]$ are common bridge ancestors of vertex $x$.
By Lemma \ref{lemma:common-strong-bridge-relation} if there is a bridge $e'=(w,z)$ in $D[\breve{r}_{h(x)},x]$ that is a common bridge ancestor of $x$ then all the bridges in the subpath $D[\breve{r}_{h(x)},z]$
are common bridge ancestors of $x$, and there is a path in $\mathcal{Q}$ from $\varphi(e'')$ to $\varphi(e')$ containing all common bridges in $D[\breve{r}_{h(x)},z]$, where $e''$ is the first bridge in $D[\breve{r}_{h(x)},x]$.
Let $e_1,e_2,\ldots,e_{k}$, be the common bridges in $D[\breve{r}_{h(x)},x]$, and let $e_{\ell}$ be the nearest common bridge that is an ancestor of $x$ in $D$ such that $\varphi(e_{\ell})\in Q(\varphi(e_1))$.
If $e_1$ is a common bridge ancestor of $x$, then by Lemmata \ref{lemma:common-descendants-path} and \ref{lemma:common-strong-bridge-relation} so are all the common bridges in the path from $\varphi(e_1)$ to $\varphi(e_{\ell})$.
If $e_1$ is not a common
bridge
ancestor of $x$, then by Lemma \ref{lemma:common-strong-bridge-relation} no other common bridge in $D[\breve{r}_{h(x)},x]$ is a common
bridge
ancestor of $x$.
Thus,
by simply testing whether $e_1$ is a common
bridge
ancestor of $x$
 we can determine whether all the common bridges $e_1,e_2,\ldots,e_{\ell}$,
are common
bridge
ancestors of $x$ and update their counter $\#SCC(D(v_i) \cap D^R(u_i))$, for $1\leq i\leq\ell$, where $e_i=(u_i,v_i)$.
We mark this relation by adding the pair $\langle h(x),x\rangle$ to the list $L$: the related update of the counters
$\#SCC(D(v_i) \cap D^R(u_i))$
 will be done during the second phase (Lines 17--33) when the pairs in the list $L$ will be processed.

To complete the proof, we need to show that the pairs in $L$ are handled correctly during the second phase of the algorithm.
Consider a pair $\langle x,y\rangle \in L$. Then, it must be $x=h(y)$, and by Lemma~\ref{lemma:ancestor-of-u} we have that $\breve{r}_{x}$ is an ancestor of $y$ in $D$.
Now we have to locate the common bridge ancestors of $y$ that are in the path $D[\breve{r}_x, y]$, and increase the count of the strongly connected components that contain only common descendants by one, since $H(y)$ induces a strongly connected component after deleting each one of them.
To find the common bridge ancestors of $y$, the algorithm locates the bridge $e=(w,z)$ that is an ancestor of $y$ in $D$ such that $\breve{D}_w=\breve{D}_x$.
If $e$ is also an ancestor of $y$ in $D^R$, it follows that $e$ is a common bridge ancestor of $y$.
Hence, Lemma~\ref{lemma:common-descendants-path} and  Lemma~\ref{lemma:common-strong-bridge-relation} imply that the common bridge ancestors we seek correspond to the nodes in the path $Q[\varphi(e), \varphi(e')]$, where $e'$ is the nearest common bridge that is an ancestor of $y$ in $D$ such that $\varphi(e')\in Q(\varphi(e))$.
This common bridge $e'$ is located, and then we need to increase by one the count of the strongly connected components that contain only common descendants of the common bridges in $Q[\varphi(e),\varphi(e')]$.
We do this by incrementing $end(\varphi^{-1}(e'))$ and $start(\varphi^{-1}(e))$ by one.
The actual number of strongly connected components that contain only common descendants of a common bridge $e=(u,v)$, i.e., $\#SCC(D(v) \cap D^R(u))$ is computed in the bottom-up traversal of $Q$. When we visit a node $\alpha \in \mathcal{Q}$, where $e=\varphi^{-1}(\alpha)=(u,v)$,
we first initialize its number $\#SCC(D(v) \cap D^R(u))$ to $end(\alpha)$. Next,  for each child $\varphi(e')$ of $\alpha$ in $\mathcal{Q}$,
we increment this number by $\#SCC(D(v') \cap D^R(u'))- start(\varphi(e'))$.
This way, we increase by one the count of strongly connected components that contain only common descendants of a common bridge if and only if it corresponds to a node in the path $Q[\varphi(e), \varphi(e')]$, as required.
\end{proof}

\begin{lemma}
	Given the dominator trees $D$ and $D^R$, the loop nesting trees $H$ and $H^R$, and the strong bridges of $G$, Algorithm \textsf{SCCsCommonDescendants} runs in $O(n)$ time.
\end{lemma}
\begin{proof}
Since each dominator tree has $n-1$ edges, we can locate the common bridges on $D$ and construct the common bridge decomposition $\breve{\mathcal{D}}$ in $O(n)$ time. The common bridge forest $\mathcal{Q}$ can also be constructed in $O(n)$ time, since it only requires a constant-time ancestor/descendant test for $D^R$~\cite{domin:tarjan} in order to identify the common bridges that are adjacent in $\mathcal{Q}$. If the dominator trees, the loop nesting trees, and the strong bridges are available, then the initialization phase of  Algorithm \textsf{SCCsCommonDescendants} (Lines 1--7) can be implemented in a total of $O(n)$ time.

Now we turn to the main loop of the algorithm (Lines 8--15) where we visit all vertices and insert the appropriate pairs into the list $L$.
It requires $O(n)$ time since it visits every vertex once and performs constant-time computations per vertex.
Here again, we use a constant-time ancestor/descendant test for both $D$ and $D^R$.
Next, we consider the last phase (Lines 16--33). As we already mentioned, $L$ contains at most $n-1$ pairs.
The \textbf{foreach} loop in Lines 17--23 requires $O(n)$ time, except for Lines 18 and 20, which we account later.
The last \textbf{foreach} loop (Lines 24--33) takes $O(n)$ time since it visits each node of $\mathcal{Q}$ only once and iterates over its children, performing only constant-time computations per child.

To complete the proof, we need to specify how to compute efficiently on Line 18 and on Line 20
the appropriate common bridges $e$ and $e'$, respectively.
We identify the common bridge $e=(w,d^R(w))$ that is an ancestor of $y$ in $D$ such that $\breve{D}_w=\breve{D}_x$.  (Recall that $h(y)=x$.) In order to locate efficiently the appropriate common bridge for every pair, we compute them together in a pre-processing step as follows. First, we perform a preorder traversal of $D$ and assign to each vertex $u$ a preorder number $pre(u)$. Then, we create two lists of triples, $A$ and $B$. List $A$ contains the triple $\langle \breve{r}_x, pre(y), 1\rangle$ for each pair $\langle x,y\rangle\in L$. List $B$ contains the triple $\langle\breve{r}_u, pre(u), 0\rangle$ for each common bridge $(u,v)$. We sort both lists in increasing order in $O(n)$ time by bucket sort and then merge them.
(Without loss of generality, we can assume that vertices of $G$ are integers from $1$ to $n$). Let $C$ be the resulting list.
We divide $C$ into sublists $C(r)$, where $r$ is the first vertex of each triple in $C(r)$. Consider a triple $\langle r, pre(y), 1\rangle$ that corresponds to the pair $\langle x,y\rangle\in L$, where $\breve{r}_x=r$.
The desired common bridge $(w,d^R(w))$ corresponds to the last triple of the form $\langle r,pre(w),0\rangle$ that precedes $\langle r, pre(y), 1\rangle$. 

Now, given a vertex $y$ and the common bridge $e=(w,z)$ computed above for a pair $\langle x,y\rangle\in L$, we wish to find the node $\varphi(e') \in \mathcal{Q}$ such that $e'$ is the nearest common bridge that is an ancestor of $y$ in $D$ and $\varphi(e') \in Q(\varphi(e))$. 
First, we assign to each tree in the common bridge forest $\mathcal{Q}$ a distinct integer id number in $[1,n]$. Also, for each node $\alpha \in \mathcal{Q}$, we store the id of the tree $Q(\alpha)$ that contains $\alpha$ in a label $\mathit{TreeID}(\alpha)$. Our task now is to find the node $\varphi(e') \in \mathcal{Q}$ such that $e'$ is the last edge in the path $\breve{D}[\breve{r}_w, \breve{r}_x]$ with
$\mathit{TreeID}(\varphi(e')) = \mathit{TreeID}(\varphi(e))$.
We compute these nodes $\varphi(e')$ of $\mathcal{Q}$ for all pairs $\langle x,y\rangle\in L$ simultaneously, by executing a depth-first search traversal of $\breve{D}$.
Before executing the dfs, we associate with each vertex $\breve{r}$ of $\breve{D}$ the list of pairs $\langle x,y\rangle\in L$ with $\breve{r}_y = r$.
We denote this list by $L_{\breve{r}}$.
During the dfs, we maintain in an array value $\mathit{Last}[j]$ the node $\alpha=\phi(e)$ such that $e$ is the last common bridge in the current dfs path with $\mathit{TreeID}(\alpha)=j$.
When we visit a vertex $\breve{r}$ of $\breve{D}$, we process each pair $\langle x,y\rangle \in L_{\breve{r}}$ and compute the desired common bridge $e'$ as follows.
Let $e$ be the common bridge computed in Line 18 of Algorithm \textsf{SCCsCommonDescendants}, as described above.
First, we set $j=\mathit{TreeID}(\varphi(e))$ and $\alpha = \mathit{Last}[j]$. Then, we have $e' = \phi^{-1}(\alpha)$.
Clearly, the running time of this process is $O(n)$ as claimed.
\ignore{
Again we create two lists of triples, $A'$ and $B'$. List $A'$ contains the triple $\langle\alpha, pre(y), 1\rangle$ for each pair $\langle x,y\rangle\in L$, where $\alpha$ is the root of $Q(\varphi((w,z)))$,
where $e=(w,z)$ is the common bridge computed above.
List $B'$ contains the triple $\langle\alpha, pre(u), 0\rangle$ for each common bridge $(u,v)$, where $\alpha$ is the root of $Q(\varphi((u,v)))$. We bucket sort both lists $A'$ and $B'$ in  increasing order and merge them in a list $C'$. As before, we divide $C'$ into sublists $C'(\alpha)$, where $\alpha$ is the first vertex of each triple in $C'(\alpha)$. Consider a triple $\langle\alpha, pre(y), 1\rangle$.
The desired nearest common bridge $(z,d^R(z))$ for $y$ corresponds to the last triple of the form $\langle\alpha ,pre(z), 0\rangle$ that precedes $\langle\alpha, pre(y), 1\rangle$.
}
\end{proof}

Then we get the following result.

\begin{theorem}
Given the dominator trees $D$ and $D^R$, the loop nesting trees $H$ and $H^R$, and the strong bridges of a strongly connected digraph $G$, we can compute the number of strongly connected components after the deletion of a strong bridge in $O(n)$ time for all strong bridges.
\end{theorem}

This yields immediately the following corollary.

\begin{corollary}
Given a directed graph $G=(V,E)$,  we can find in worst-case time $O(m+n)$ a strong bridge $e$  in $G$ that maximizes/minimizes the total number of strongly connected components of $G\setminus e$.
\end{corollary}

The algorithm given in this section can be extended
to a more general family of
functions
defined over the sizes of the strongly connected components of $G\setminus e$, for each edge $e$.
More precisely, let
$C_1,C_2,\ldots ,C_k$ be the
the strongly connected components in $G\setminus e$:
in this section we considered the specific function $f(|C_1|,|C_2|,...,|C_k|)=k$,  which computes the number of strongly connected components in $G\setminus e$ for each edge $e$. Let $e=(u,v)$ be an edge to be deleted. Our algorithm counts
independently the number of strongly connected components in $G \setminus e$ that contain vertices in $D(v)$, $D^R(u)$, $V\setminus{D(u) \cup D^R(u)}$, and $D(u) \cap D^R(v)$. Finally, the total number of strongly connected components in $G\setminus e$ is computed as $\#SCC_e(V) = \#SCC(D(v)) + \#SCC(D^R(u)) - \#SCC(D(v)\cap D^R(u))+1$. Note that we have to compute the quantity $\#SCC(D(v)\cap D^R(u))$ in order to ``cancel out'' the  strongly connected components induced by $D(v)\cap D^R(u)$, which are double counted in  $\#SCC(D(v)) + \#SCC(D^R(u))$.
This can be extended to the computation of more general functions, where this ``cancellation'' applies.

\begin{theorem}
\label{theorem:generic}
Let $\odot$ be an associative and commutative binary
operation such that its inverse operation $\odot^{-1}$ is defined and both $\odot$  and $\odot^{-1}$  are computable in constant time. Let $f(x)$ be a function defined on positive integers which can be computed in constant time. Given a strongly connected digraph $G$, we can compute in $O(m+n)$ time, for all edges $e$, the function $f(|C_1|)\odot f(|C_2|)\odot ...\odot f(|C_k|)$, where $C_1,C_2,...,C_k$ are the strongly connected components in $G \setminus e$.
\end{theorem}

The results of this section, i.e., counting the number of strongly connected components in $G\setminus e$, for all edges $e$, correspond to $\langle \odot, \odot^{-1} \rangle = \langle +, - \rangle$ and $f(x)=1$. By setting $\langle \odot, \odot^{-1} \rangle = \langle +, - \rangle$ and $f(x)=x(x-1)/2$, 
we can compute the number of strongly connected pairs in $G \setminus e$, for all edges $e$.
Hence, we obtain a linear-time algorithm to compute the \emph{most critical edge} of a directed graph with respect to strong connectivity, i.e., the edge $e$ of $G$ such that the number of strongly connected pairs of vertices in $G \setminus e$ is minimized. (See Section \ref{sec:vertices} for the vertex version of this problem, and \cite{Paudel:2018} for an alternative linear-time algorithm.) 
By setting $\langle \odot, \odot^{-1} \rangle = \langle *, / \rangle$ and $f(x)=x$, we can compute the product of the sizes of the strongly connected components in $G \setminus e$, for all edges $e$.

\subsection{Finding all the smallest and all the largest strongly connected components of $G\setminus e$}
\label{sec:minmax-scc}

Let $G$ be a strongly connected digraph, with $m$ edges and $n$ vertices. Since $G$ is strongly connected, $m\geq n$. In this section we consider the problem of answering the following aggregate query: ``Find the size of the largest/smallest strongly connected component of $G \setminus e$, for all edges $e$.'' We recall here that the size of a  strongly connected component is given by its number of vertices, so the largest component (resp., smallest component) is the one with the maximum (resp., minimum) number of vertices.
In the following, we restrict ourselves to the computation of the largest strongly connected components of $G\setminus e$, since the smallest strongly connected components can be  computed in a completely analogous fashion.
Note again that the naive solution is to compute the strongly connected components of $G \setminus e$ for all strong bridges $e$ of $G$, which takes $O(mn)$ time.
We will be able to provide a linear-time algorithm for this problem, which also gives an asymptotically optimal $O(m)$-time algorithm for the motivating biological application discussed in the introduction, i.e., finding the strong bridge $e$ that minimizes the size of the largest strongly connected component in $G\setminus e$. Once we find such a strong bridge $e$,
we can report the actual strongly connected components of $G\setminus e$ in $O(n)$ additional time by using the algorithm of Section \ref{sec:all-scc}.

Let $S\subseteq V$ be a subset of vertices of $G$. We denote by $\mathit{LSCC}(S)$ the size of the largest strongly connected component in the subgraph of $G$ induced by the vertices in $S$.
Also, for an edge $e$ of $G$, we denote by $\mathit{LSCC}_e(V)$ the size of the largest strongly connected component in $G \setminus e$. Our goal is to compute $\mathit{LSCC}_e(V)$ for every strong bridge $e$ in $G$.
Then, Theorem \ref{cor:scc} immediately implies the following:
\begin{corollary}
\label{cor:lscc}
Let $e=(u,v)$ be a strong bridge of $G$ and let $s$ be an arbitrary vertex in $G$.
The cardinality of the largest strongly connected component of $G \setminus e$ is equal to
\begin{itemize}
\item[(a)] $\max \{ \mathit{LSCC}(D(v)), |V \setminus D(v)| \}$ when $e$ is a bridge in $G_s$ but not in $G_s^R$.
\item[(b)] $\max \{ \mathit{LSCC}(D^R(u)), |V \setminus D^R(u)| \}$ when $e$ is a bridge in $G_s^R$ but not in $G_s$.
\item[(c)] $\max \{ \mathit{LSCC}(D(v)), \mathit{LSCC}(D^R(u)), |V \setminus ( D(v) \cup D^R(u) )| \}$ when $e$ is a common bridge of $G_s$ and $G_s^R$.
\end{itemize}
Moreover, $\mathit{LSCC}(D(v)) = \max_{w} \{ |H(w)| \ : \ w \in D(v) \mbox{ and } h(w) \not \in D(v) \}$ and
$\mathit{LSCC}(D^R(u)) = \max_{w} \{ |H^R(w)| \ : \ w \in D^R(u) \mbox{ and } h^R(w) \not \in D^R(u) \}$.
\end{corollary}

Now we develop an algorithm that applies Corollary \ref{cor:lscc}. Our algorithm, detailed below (see Algorithm \ref{algorithm:LSCC}),
uses the dominator and the loop nesting trees of $G_s$ and its reverse $G_s^R$, with respect to an arbitrary start vertex $s$, and computes for each strong bridge $e$ of $G$ the size of the largest strongly connected component of $G \setminus e$, denoted by $\mathit{LSCC}_e(V)$.

\begin{algorithm}
\LinesNumbered
\DontPrintSemicolon
 \KwIn{Strongly connected digraph $G=(V,E)$}
 \KwOut{Size of the largest strongly connected component of $G\setminus e$ for each strong bridge $e$}

 \textbf{Initialization:}\;
 	Compute the reverse digraph $G^R$. Select an arbitrary start vertex $s \in V$.  \;
	Compute the dominator trees $D$ and $D^R$ of the flow graphs $G_s$ and $G_s^R$, respectively.\;
	Compute the loop nesting trees $H$ and $H^R$ of the flow graphs $G_s$ and $G_s^R$, respectively.\;
	Compute the set of bridges $\mathit{Br}$ and $\mathit{Br}^R$ of the flow graphs $G_s$ and $G_s^R$, respectively.\;

 \textbf{Process bridges:}\;
 \ForEach{$e=(u,v)$ in $\mathit{Br}$ in a bottom-up order of $D$}
 {
	\eIf{$e \not \in \mathit{Br}^R$}
    {
	  Compute $\mathit{LSCC}(D(v))$\;
      Set $\mathit{LSCC}_e(V) = \max \{ \mathit{LSCC}(D(v)), |V|-|D(v)| \}$
      }
	{					
	  Compute $\mathit{LSCC}(D(v))$, $\mathit{LSCC}(D^R(u))$, and $|D(v) \cup D^R(u)|$\;
	 Set $\mathit{LSCC}_e(V) =\max\{ \mathit{LSCC}(D(v)), \mathit{LSCC}(D^R(u)), |V|-|D(v)\cup D^R(u)| \}$
   }
 }

\ForEach{$e=(u,v)$ in $\mathit{Br}^R$ in a bottom-up order of $D^R$}
{
\If{$e \not \in Br$}
{
Compute $\mathit{LSCC}(D^R(u))$\;
Set $\mathit{LSCC}_e(V) =\max\{\mathit{LSCC}(D^R(u)), |V|-|D^R(u)|\}$
 }
 }
 \caption{\textsf{LSCC}}
 \label{algorithm:LSCC}
\end{algorithm}

In order to get an efficient implementation of our algorithm, we need to specify how to compute efficiently the following quantities:

\begin{itemize}
\item[(a)] $\mathit{LSCC}(D(v))$  for every vertex $v$ such that the edge $(d(v),v)$ is a bridge in flow graph $G_s$.
\item[(b)] $\mathit{LSCC}(D^R(u))$ for every vertex $u$ such that the edge $(u,d^R(u))$ is a bridge in flow graph $G_s^R$.
\item[(c)] $|D(v)\cup D^R(u)|$ for every strong bridge $(u,v)$ of $G$ that is a common bridge of flow graphs $G_s$ and $G_s^R$.
\end{itemize}

We deal with computations of type (a) first. The computations of type (b) are analogous. We precompute for all vertices $v$ the number of their descendants in the loop nesting tree $H$, and we initialize $\mathit{LSCC}(D(v)) = 0$. Then we process $D$ in a bottom-up order.
For each vertex $v$, we store the nearest ancestor $w$ of $v$ in $D$ such that $(d(w),w)$ is a bridge in $G_s$.
Recall that we use the notation $r_v$ to refer to
the vertex $w$ with this property, and
we have $r_v = s$ if no such $w$ exists for $v$.
For every vertex $v$ in a bottom-up order of $D$  we update the current value of $\mathit{LSCC}(D(r_v))$ by setting
$$
\mathit{LSCC}(D(r_v))=\max\{\mathit{LSCC}(D(r_v)), |H(v)| \}.
$$
If $(d(v),v)$ is also a bridge in $G_s$, we update the current value of $\mathit{LSCC}(D(r_{d(v)}))$ by setting
$$
\mathit{LSCC}(D(r_{d(v)}))=\max\{\mathit{LSCC}(D(r_{d(v)})),  \mathit{LSCC}(D(v))\}.
$$
These computations take $O(n)$ time in total.

Finally, we need to specify how to compute the values of type (c), that is, we need to compute the cardinality of the union $D(v)\cup D^R(u)$ for each strong bridge $(u,v)$ that is a common bridge in both flow graphs $G_s$ and $G_s^R$.
Since $|D(v)\cup D^R(u)| = |D(v)| + |D^R(u)| - |D(v)\cap D^R(u)|$, it suffices to compute the cardinality of the intersections $D(v)\cap D^R(u)$ for all common bridges $(u,v)$.
We describe next how to compute these values in $O(n)$ time.
This gives an implementation of Algorithm \textsf{LSCC} which runs in $O(m+n)$ time in the worst case, or in $O(n)$ time in the worst case once the dominator trees, loop nesting trees and strong bridges are available.

\subsubsection*{Computing common descendants of strong bridges}
\label{sec:commonDescendants}

We now consider the problem of computing
the number of common descendants, i.e.,
the cardinality of the intersections $D(v)\cap D^R(u)$, for all common bridges $e=(u,v)$.
Let $x$ be a common descendant of a strong bridge $e=(u,v)$ (i.e., $x\in D(v)$ and $x\in D^R(u)$), and let $C$ be the strongly connected component containing $x$ in $G\setminus e$.
By Theorem \ref{cor:scc}(c), $C \subseteq D(v) \cap D^R(u)$, i.e., all vertices in the same strongly connected component $C$ of $G \setminus e$ as $x$ must also be common descendants of $e$.
Therefore, the number of common descendants of a strong bridge $e$ can be computed as the sum of the sizes of the strongly connected components containing a vertex that is a common descendant of $e$. This observation will allow us to solve efficiently our problem with a simple variation of Algorithm \textsf{SCCsCommonDescendants} from Section \ref{sec:SCCs-num}.

Recall that Algorithm \textsf{SCCsCommonDescendants} maintains for each strong bridge $e=(u,v)$ a counter $\#SCC(D(v) \cap D^R(u))$ for the number  of strongly connected components in $G\setminus e$ that contain only vertices in $D(v) \cap D^R(u)$.
The high-level idea behind the new algorithm is
to maintain for each strong bridge $e=(u,v)$ a counter for \emph{the sum of the sizes} of the strongly connected components in $G\setminus e$ that contain only vertices in $D(v) \cap D^R(u)$. We can maintain those counters  by proceeding as in Algorithm \textsf{SCCsCommonDescendants}. Exactly as before, we identify for each vertex $x$ all the common bridges $e$ after whose deletion the set $H(x)$ induces a strongly connected component containing only vertices that are common descendants of $e$.
However, this time we need to compute for each common bridge $e=(u,v)$ the total number of common descendants (i.e., vertices in $D(v) \cap D^R(u)$)
rather than the number of strongly connected components in $D(v) \cap D^R(u)$:
when we process $H(x)$ we increase the counter of the common bridge $e$ by $|H(x)|$ (rather than by 1 as before).
This is correct since by Theorem \ref{cor:scc} all vertices in $H(x)$ are common descendants of $e$.

We next describe the two small modifications needed in the pseudocode of Algorithm \textsf{SCCsCommonDescendants} to deal with our problem.
Recall that Algorithm \textsf{SCCsCommonDescendants}
added a pair
$\langle{}h(x),x\rangle{}$ to the list $L$ (Lines 9--15) to indicate that
the set $H(x)$ induces a strongly connected component in $G\setminus e$,
for each common bridge $e$ that lies in the path $D[\breve{r}_{h(x)},x]$ from $\breve{r}_{h(x)}$ to $x$ in the dominator tree $D$ and that is also a common ancestors of $x$.
This time, we will insert a triple
of the form $\langle{}h(x),x,|H(x)|\rangle{}$, to denote that the set $H(x)$ contributes $|H(x)|$ common descendants to all common bridges in $D[\breve{r}_{h(x)},x]$ that are common ancestors of $x$.
Algorithm \textsf{SCCsCommonDescendants} used the information stored in the pairs
$\langle{}h(x),x\rangle{}$ while processing the list $L$ in Lines 17--23: for each pair $\langle{}h(x),x\rangle{}$, the algorithm incremented by 1 the values $start(\varphi(e))$ and $end(\varphi(e'))$, where $e$ and $e'$ were respectively the first and last bridges in the sequence of common strong bridges in $D[\breve{r}_{h(x)},x]$ that are common ancestors of $x$.
The new algorithm will process a triple in similar fashion:
when a triple $\langle{}h(x),x,|H(x)|\rangle{}$
is extracted from the list $L$, the algorithm will increment by $|H(x)|$ the corresponding values $start(\varphi(e))$ and $end(\varphi(e'))$.
Clearly these small modifications do not affect the $O(n)$ time complexity of Algorithm \textsf{SCCsCommonDescendants}:

\begin{lemma}
Given the dominator trees $D$ and $D^R$, the loop nesting trees $H$ and $H^R$, and the common bridges of $G$, we can compute the number of the common descendants of all common strong bridges in $O(n)$ time.
\end{lemma}

Hence, we have the following results.

\begin{theorem}
Given the dominator trees $D$ and $D^R$, the loop nesting trees $H$ and $H^R$, and the strong bridges of $G$, we can compute the size of the largest or the smallest strongly connected component after the deletion of a strong bridge in $O(n)$ time for all strong bridges.
\end{theorem}

\begin{corollary}
Given a strongly connected directed graph $G$,  we can find in worst-case time $O(m+n)$ a strong bridge $e$  that minimizes/maximizes the size of the largest/smallest strongly connected component in $G\setminus e$.
\end{corollary}

\section{Computing $2$-edge-connected components}
\label{sec:2ECBs}

In this section we describe a new and simpler linear-time algorithm for computing the $2$-edge-connected components of a strongly connected digraph $G$. See Figure \ref{figure:2ECB-example}.
As the algorithms in the previous sections, we will use again the dominator trees, $D$ and $D^R$, and the loop nesting trees, $H$ and $H^R$, of flow graphs $G_s$ and $G_s^R$, with respect to an arbitrary start vertex $s$.

\begin{figure}[t!]
\begin{center}
\centerline{\includegraphics[trim={0 0 0 0cm}, clip=true, width=1.4\textwidth]{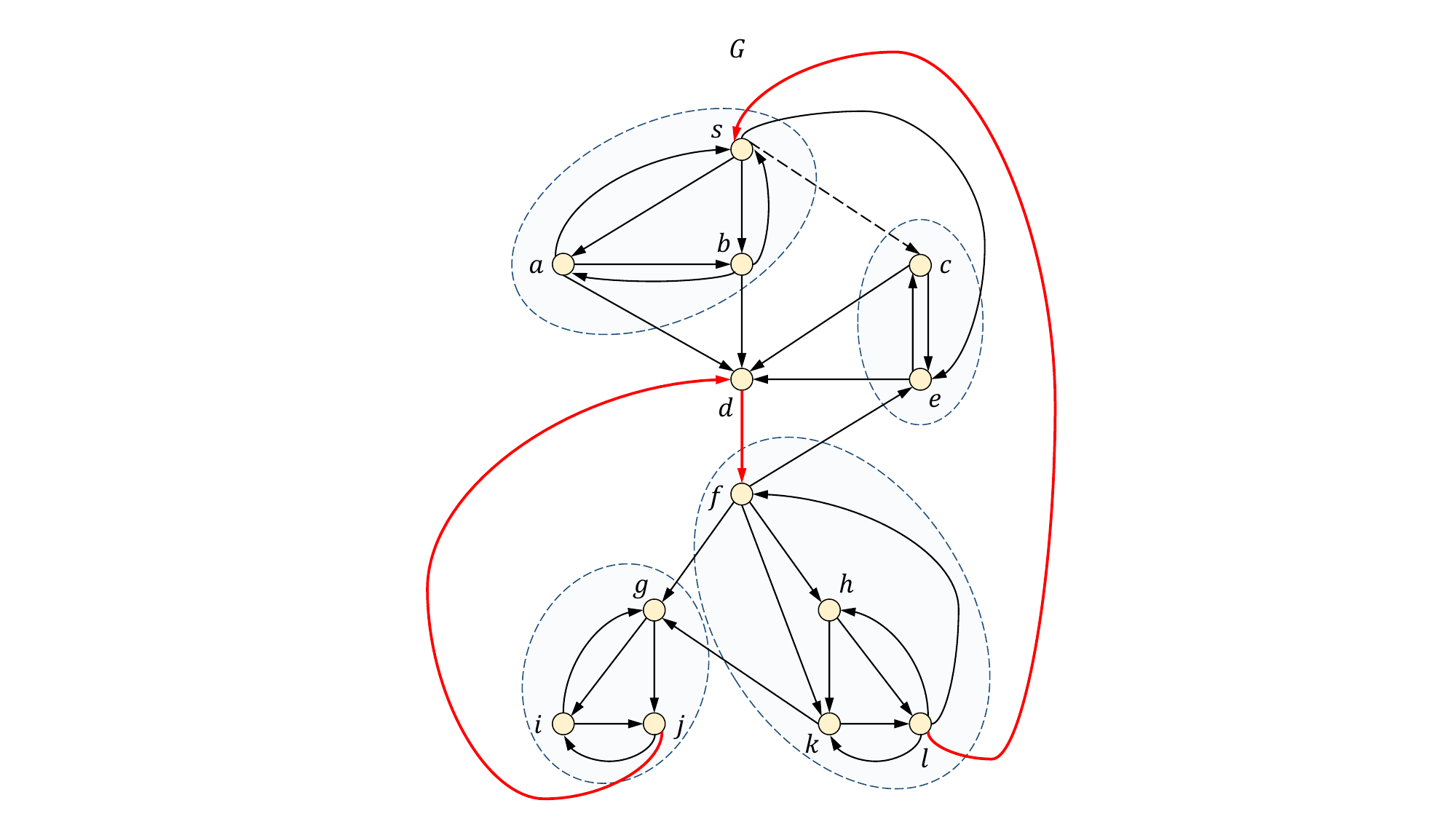}}	
\caption{The $2$-edge-connected components of the digraph of Figure \ref{figure:example1ab}. Strong bridges of $G$ are shown in red. (Better viewed in color.)}
\label{figure:2ECB-example}
\end{center}
\end{figure}

Recall the bridge decomposition
of $D$ and $D^R$ into forests $\mathcal{D}$ and $\mathcal{D}^R$, that we obtain by deleting from $D$ and $D^R$ the bridges of $G_s$ and $G_s^R$ respectively. We denote by $D_u$ (resp., $D^R_u$) the tree in $\mathcal{D}$ (resp., $\mathcal{D}^R$) containing vertex $u$, and by $r_u$ (resp., $r^R_u$) the root of $D_u$ (resp., $D^R_u$). This bridge decomposition gives an initial partition of the vertices into coarser components, as the following lemma suggests:
\begin{lemma} \emph{(\cite{2ECC:GILP:TALG})}
\label{lemma:bridge-decomposition}
Any two vertices $u$ and $v$ are $2$-edge-connected in $G$ only if $r_u = r_v$ and $r^R_u = r^R_v$.
\end{lemma}

We briefly review the algorithm in \cite{2ECC:GILP:TALG} in order to highlight the differences in our approach.\footnote{We note that \cite{2ECC:GILP:TALG} refers to the $2$-edge-connected components of a digraph as \emph{$2$-edge-connected blocks}.}
In \cite{2ECC:GILP:TALG}, the bridge decomposition of $D$ is used in order to partition the vertices into coarse components that are in the same tree of $\mathcal{D}$. Then, for each tree $D_r \in \mathcal{D}$, it constructs an auxiliary graph $G_r$ that maintains the same $2$-edge-connectivity relation as $G$ for the vertices in $D_r$. This process is repeated in each $G_r^R$, producing a second level of auxiliary graphs. Finally, the $2$-edge-connected components of $G$ are formed by the strongly connected components of the second-level auxiliary graphs, after removing a specific strong bridge in each of them.
Similar to the algorithm in \cite{2ECC:GILP:TALG} we use the dominator trees to get an approximate partition of the vertices into coarser components. However, differently from \cite{2ECC:GILP:TALG} we consider the bridge decomposition of both $D$ and $D^R$ at the same time, as suggested by Lemma \ref{lemma:bridge-decomposition}, and
exploit the information about the hierarchy of strongly connected subgraphs provided by the two loop nesting trees. This will simplify substantially the resulting algorithm.

We define a relation among the vertices with respect to the bridge decomposition $\mathcal{D}$ of the dominator tree $D$ and the loop nesting tree $H$ of $G$ as follows.
We say that a vertex $x$ is \emph{boundary} in $H$ if $h(x) \not\in D_x$, i.e., when $x$ and its parent in $H$ lie in different trees of $\mathcal{D}$. As a special case, we also let $s$ be a boundary vertex of $H$.
The \emph{nearest boundary vertex} of $x$ in $H$, denoted by $h_x$, is the nearest ancestor of $x$ in $H$ that is a boundary vertex in $H$. Hence, if $r_x = s$ then $h_x = s$. Otherwise, $h_x$ is the unique ancestor of $x$ in $H$ such that $h_x \in D_x$ and $h(h_x) \not \in D_x$.
We define the \emph{nearest boundary vertex} of $x$ in $H^R$ similarly.
A vertex $x$ is \emph{boundary} in $H^R$ if $h^R(x) \not\in D^R_x$, i.e., when $x$ and its parent in $H^R$ lie in different trees of $\mathcal{D}^R$. Again, we let $s$ be a boundary vertex of $H^R$.
Then, the \emph{nearest boundary vertex} of $x$ in $H^R$, denoted by $h^R_x$, is the nearest ancestor of $x$ in $H^R$ that is a boundary vertex in $H^R$.

As we show below, we can compute $h_x$ and $h^R_x$, for all vertices $x$, in $O(n)$ time. First we describe how to use these concepts in order to compute the $2$-edge-connected components of $G$.  Our new algorithm, dubbed \textsf{HD2ECC}
since it uses the loop nesting trees $H$ and $H^R$, and the dominator trees $D$ and $D^R$, works as follows. We first assign a label to each vertex $x$ that specifies the location of $x$ in the dominator and loop nesting trees. Specifically we set $label(x) = \langle r_x, h_x, r^R_x, h^R_x\rangle$. These labels have the property that two vertices $x$ and $y$ are $2$-edge-connected if and only if they have exactly the same label.
Thus, we can form the $2$-edge-connected components by bucket sort in $O(n)$ time. That is, we form a list of tuples $\langle label(x),x\rangle$ and sort them lexicographically by label. The pseudocode of the algorithm is given below (see Algorithm \ref{algorithm:HD2ECC})

\begin{algorithm}
\LinesNumbered
\DontPrintSemicolon
 \KwIn{Strongly connected digraph $G=(V,E)$}
 \KwOut{The $2$-edge-connected components of $G$}
 \textbf{Initialization:}\;
 	Compute the reverse digraph $G^R$.
Select an arbitrary start vertex $s \in V$.

	Compute the dominator trees $D$ and $D^R$ of the flow graphs $G_s$ and $G_s^R$, respectively.\;
	
	Compute the loop nesting trees $H$ and $H^R$ of the flow graphs $G_s$ and $G_s^R$, respectively.\;
	
		Compute $\mathcal{D}$ (the bridge decomposition of $D$) and $\mathcal{D}^R$ (the bridge decomposition of $D^R$).\;

 \ForEach{$x \in V$}
 {Find the roots $r_x$ and $r^R_x$ in the bridge decomposition\;
  Find the nearest boundary vertices $h_x$ and $h^R_x$\;
  Set $label(x) = \langle r_x, h_x, r^R_x, h^R_x\rangle$\;
}

\textbf{Computation of $2$-edge-connected components:}\;
Sort the tuples $\langle label(x),x\rangle$ lexicographically by their labels\;
Partition the vertices into components, where $x, y \in V$ are in the same component if and only if $label(x) = label(y)$
\caption{\textsf{HD2ECC}}
\label{algorithm:HD2ECC}
\end{algorithm}

Next we prove the correctness of our algorithm.
Let $x$ and $y$ be two distinct vertices in $G$. We say that a strong bridge $e$ \emph{separates $x$ and $y$} (or equivalently that $e$ is a \emph{separating edge for $x$ and $y$}) if either all paths from $x$ to $y$ or all paths from $y$ to $x$ contain edge $e$ (i.e., $x$ and $y$ belong to different strongly connected components of $G\setminus e$).
Clearly, $x$ and $y$ are $2$-edge-connected if and only if there exists no separating edge for them.
We give two key lemmata that supplement the results from Section~\ref{sec:scc} and form the basis of our algorithm.

\begin{lemma}
\label{separators-ancestors}
Let $e$ be strong bridge of $G$ that is a separating edge for vertices $x$ and $y$. Then $e$ must appear in at least one of the paths $D[s,x]$, $D[s,y]$, $D^R[s,x]$, and $D^R[s,y]$.
\end{lemma}
\begin{proof}
Assume by contradiction that a strong bridge $e=(u,v)$ separates $x$ and $y$ but it does not appear in any of the paths $D[s,x]$, $D[s,y]$, $D^R[s,x]$, and $D^R[s,y]$.
The fact that $(v,u) \not\in D^R[s,x]$ implies that there is a path $\pi$ in $G$ from $x$ to $s$ that avoids $(u,v)$.
Similarly, the fact that $(u,v) \not \in D[s,y]$ implies that there is a path $\pi'$ in $G$ from $s$ to $y$ that avoids $(u,v)$.
Then $\pi \cdot \pi'$ is a path in $G$ from $x$ to $y$ that does not contain the edge $(u,v)$.
Analogously, the fact that $(v,u) \not\in D^R[s,y]$ and $(u,v) \not \in D[s,x]$ implies that there is a path in $G$ from $y$ to $x$ that does not contain the edge $(u,v)$.
This contradicts the assumption that $e$ separates $x$ and $y$, i.e., that $x$ and $y$ are not strongly connected in $G\setminus e$.
\end{proof}

\begin{lemma}
\label{bridge-relation}
Let $x$ and $y$ be vertices such that $r_x = r_y \neq s$, i.e., $x$ and $y$ are in the same tree $D_r$ of the bridge decomposition of $D$ and $D_r$ is not rooted at $s$, i.e., $r \neq s$.
A bridge $e$ that is not a descendant of $r$ in $D$ is a separating edge for $x$ and $y$ only if the bridge $(d(r), r)$ is a separating edge for $x$ and $y$.
\end{lemma}
\begin{proof}
Let $e=(u,v)$ be a bridge that is not a descendant of $r$ in $D$ and that separates $x$ and $y$ in $G$.
Since $e=(u,v)$ separates $x$ and $y$, then either any path from $x$ to $y$ or any path from $y$ to $x$ must contain $e$. Without loss of generality, assume that any path from $x$ to $y$ contains $e$ (otherwise swap $x$ and $y$ in the following). Observe that any path from $x$ to $y$ containing edge $e=(u,v)$ must contain a path from $v$ to $y$ as a subpath. The fact that $e=(u,v)$ is not a descendant of $r$ in $D$ implies that $v$ is not a descendant of $r$ in $D$ as well. Since $r=r_y=r_x$, $y$ is a descendant of $r$ in $D$, and thus by Lemma \ref{lemma:partition-paths}, any path from $v$ to $y$ must contain the bridge $(d(r),r)$. Therefore if any path from $x$ to $y$ contains the bridge $e$, it must also contain the bridge $(d(r),r)$.
\end{proof}

\begin{theorem}
\label{theo:correct}
Algorithm \textsf{HD2ECC} is correct.
\end{theorem}
\begin{proof}
First we show that if two vertices $x$ and $y$ have the same label after the execution of Algorithm  \textsf{HD2ECC}, then they must be $2$-edge-connected in the digraph $G$.
Assume by contradiction that
$label(x) = label(y)$ but $x$ and $y$ are not 2-edge-connected.
The fact that $label(x) = label(y)$ implies that $x$ and $y$ are in the same trees in the bridge decomposition of $D$ and $D^R$.
Since $x$ and $y$ are not 2-edge-connected,
there must exist a strong bridge $e$ that separates $x$ and $y$ in $G$.
By Property~\ref{property:strong-bridge}, $e$ must be either a bridge in the flow graph $G_s$ or in the flow graph $G_s^R$ (or in both). Since $x$ and $y$ are in the same trees in the bridge decomposition of $D$ and $D^R$,
by Lemma~\ref{separators-ancestors},
the only possibility for bridge $e$ is to lie either in the path $D[s,r_x]$ or in the path $D^R[s, r^R_x]$ (or in both).
Suppose that $e$ appears in $D[s,r_x]$. Then, by Lemma~\ref{bridge-relation}, the bridge $e' = (d(r_x),r_x)$ also separates $x$ and $y$. However, since  $label(x) = label(y)$, we have that $h_x = h_y$, and thus
both $x$ and $y$ are in the subtree $H(h_x)$ of the loop nesting tree $H$. By definition of nearest boundary vertex $h_x$ and by Lemma \ref{lemma:subtree}, $H(h_x)$ must be strongly connected in $G \setminus e'$. Since both $x$ and $y$ are in $H(h_x)$, this contradicts the fact the $e'$ separates $x$ and $y$.
The case where $e$ appears in $D^R[s,r_x^R]$ is symmetric; by Lemma~\ref{bridge-relation}, the bridge $(r_x^R, d^R(r_x^R),)$ of $G_s^R$ also separates $x$ and $y$, which now contradicts the fact that $h^R_x = h^R_y$.

Next, we prove that two vertices $x$ and $y$ that are given different labels by the algorithm cannot be $2$-edge-connected. To show this, we go through a case analysis and in each case we exhibit a separating edge $e$ for $x$ and $y$.
Suppose first that $r_x \neq r_y$. Without loss of generality, assume that $r_x$ is not a descendant of $r_y$ in $D$ (otherwise swap the roles of $r_x$ and $r_y$ in what follows). Then,
by Lemma \ref{lemma:partition-paths}, every path from $x$ to $y$ passes through the strong bridge $e=(d(r_y),r_y)$, which is a separating edge for $x$ and $y$.
The case where $r^R_x \neq r^R_y$ is symmetric, and in this case $e=(r^R_y, d^R(r^R_y))$ is a separating edge for $x$ and $y$.
Assume now that $r_x=r_y$, $r^R_x=r^R_y$ and $h_x \neq h_y$. Again, without loss of generality, suppose that $h_x$ is not a descendant of $h_y$ in $H$, and let $e = (d(r_y),r_y)$.
By Lemma \ref{lemma:subtree}, $H(h_y)$ induces a strongly connected component of $G \setminus e$. But since $h_x \not\in H(h_y)$ we also have $x \not\in H(h_y)$. Thus, $x$ and $y$ are not strongly connected in $G \setminus e$, which implies that $e= (d(r_y),r_y)$ is a separating edge for $x$ and $y$. A symmetric argument shows that if $h^R_x \neq h^R_y$ then  $e=(r^R_y, d^R(r^R_y))$ is a separating edge for $x$ and $y$.
\end{proof}

We observe that, given the labels of two vertices $x$ and $y$, it is straightforward to test in constant time if $x$ and $y$ are $2$-edge-connected. Furthermore, if $x$ and $y$ are not $2$-edge-connected, we can provide in constant time a separating edge for $x$ and $y$, as specified in the proof of Theorem \ref{theo:correct}.

\begin{theorem}
Algorithm \textsf{HD2ECC} runs in $O(m+n)$ time, or in $O(n)$ time if the dominator trees, the loop nesting trees, and the bridges of flow graphs $G_s$ and $G_s^R$ are provided as input.
\label{lemma:2ECB}
\end{theorem}
\begin{proof}
As already mentioned, the dominator trees, the loop nesting trees, and the bridges of $G_s$ and $G_s^R$ can be computed in $O(m+n)$ time. Since the total number of bridges is at most $2n-2$, it is straightforward to obtain the bridge decomposition $\mathcal{D}$ and $\mathcal{D}^R$ of $D$ and $D^R$ in $O(n)$ time.
The roots $r_x$ can be computed easily by traversing each tree in $\mathcal{D}$ separately, and for such vertex $x$ of a tree $D_r$ with root $r$ set $r_x=r$.
We can obtain all the roots $r^R_x$ in the dominator tree $D^R$ analogously.
We next describe how to compute all the values $h_x$ in time $O(n)$.
First we construct, for each vertex $x$, a list of its children in the loop nesting tree $H$. We set $h_s=s$ and we traverse $H$ in a top-down order. For each vertex $x$ that we visit we iterate over all of its children, and for each child $y$ of $x$ we do the following. If $r_x=r_y$ then we set $h_y=h_x$, otherwise we set $h_y=y$. Clearly this can be done in $O(n)$ time.
The computation of $h^R_x$ is completely analogous.
Finally, we sort the list of tuples $\langle label(y),y\rangle$ lexicographically in $O(n)$ time by bucket sort.
\end{proof}

Similarly to the algorithm in \cite{2ECC:GILP:TALG}, we can adapt Algorithm \textsf{HD2ECC} to provide
a \emph{sparse certificate} for the $2$-edge-connected components, i.e., a subgraph of the input digraph $G$ that has $O(n)$ edges and has the same $2$-edge-connected components as $G$. However, using our framework, we can obtain a sparse certificate that has some
additional interesting properties. We discuss them in Section \ref{sec:sparse-certificate}.

\section{Pairwise $2$-edge connectivity queries}
\label{sec:other}

Let $G=(V,E)$ be a strongly connected digraph, let $x$ and $y$ be any two vertices of $G$, and let $e=(u,v)$ be a strong bridge of $G$. Recall that we say that $e$ is a \emph{separating edge for $x$ and $y$} (or equivalently that $e$ \emph{separates $x$ and $y$}) if $x$ and $y$ are not strongly connected in $G\setminus e$.
In this section we show how to extend our data structure to answer in asymptotically optimal time the following types of queries:
\begin{itemize}
\item[(a)] Test if two query vertices $x$ and $y$ are $2$-edge-connected; if not, report a separating edge for $x$ and $y$.
\item[(b)] Test whether a given edge $e$ separates two query vertices $x$ and $y$.
\item[(c)] Report all the separating edges for a given pair of vertices $x$ and $y$.
\end{itemize}

We note that the data structure of \cite{2ECC:GILP:TALG} only supports, in constant time, the queries of type (a). We can answer such queries, also in constant time, with our framework as follows.
We compute in $O(m+n)$ time $label(x) = \langle r_x, h_x, r^R_x, h^R_x\rangle$, for all vertices $x$, as in Algorithm \textsf{HD2ECC} of Section \ref{sec:2ECBs}.
Recall that $r_x$ (resp., $r_x^R$) is the root of the tree that contains $x$ in the bridge decomposition of $D$ (resp., $D^R$),
and $h_x$ (resp., $h_x^R$) is the nearest boundary vertex of $x$ in $H$ (resp., $H^R$), i.e., the nearest ancestor $z$ of $x$ in $H$ (resp., $H^R$)
such that $h(z) \not\in D_z$ (resp., $h^R(z) \not\in D_z^R$).
From Theorem \ref{theo:correct} we have that $x \leftrightarrow_{\mathrm{2e}} y$ if and only if $label(x) = label(y)$.
Thus, we can test in constant time if  $x$ and $y$ are $2$-edge-connected.
Now suppose that $x$ and $y$ are not $2$-edge-connected.
Then $label(x) \not= label(y)$. Consider that $r_x \not= r_y$.
Assume, without loss of generality, that $r_x$ is not an ancestor of $r_y$
in $D$. Then, by Lemma \ref{lemma:partition-paths}, bridge $(d(r_x), r_x)$ is a separation edge for $x$ and $y$.
Now consider that $r_x = r_y$ and $h_x \not= h_y$. By the definition of a boundary vertex we have $h_x, h_y \in D_x$ and
$h(h_x), h(h_y) \not \in D_x$. Hence, $H(h_x)$ and $H(h_y)$ are disjoint, so by Theorem \ref{cor:scc},
$x$ and $y$ are not strongly connected in $G \setminus (d(r_x), r_x)$. Thus, bridge $(d(r_x), r_x)$ is a separation edge for $x$ and $y$.
We can find a separating edge when $r_x^R \not= r_y^R$ or $h_x^R \not= h_y^R$ similarly.

Next we deal with queries of type (b) and (c).

\begin{lemma}
\label{lemma:separating_two_vertices}
Let $x$ and $y$ be two vertices, and let $w$ and $w^R$ be their nearest common ancestors in the loop nesting trees $H$ and $H^R$, respectively.
Then, a strong bridge $e=(u,v)$ is a separating edge for $x$ and $y$ if and only if one of the following conditions holds:
\begin{itemize}
\item[(1)] The edge $e$ is an ancestor of $x$ or $y$ in $D$ and $w$ is not a descendant of $e$ in $D$.
\item[(2)] The edge $e$ is an ancestor of $x$ or $y$ in $D^R$ and $w^R$ is not a descendant of $e$ in $D^R$.
\end{itemize}
\end{lemma}
\begin{proof}
We will only prove (1), as the proof of (2) is completely analogous.
By Lemma~\ref{separators-ancestors}, a strong bridge $e=(u,v)$ that separates $x$ and $y$ lies in at least one of the paths $D[s,x]$, $D[s,y]$, $D^R[s,x]$, or $D^R[s,y]$.
To prove (1)
we assume that the strong bridge $e$ lies either in $D[s,x]$ or in $D[s,y]$, and argue that $e$ separates $x$ and $y$ if and only if their nearest common ancestor $w$ in $H$ is not a descendant
of $e$ in $D$. Let $z$ be a descendant of $v$ in $D$.
By Lemma~\ref{lemma:subtree}, all the vertices of $H(z)$ are strongly connected in $G\setminus e$. Hence, if $w$ is a descendant of $e$ then $x$ and $y$ must be strongly connected in $G\setminus e$, since $x, y \in H(w)$.

Next we prove the opposite direction.
Assume by contradiction that
the edge $e=(u,v)$ is an ancestor of $x$ or $y$ in $D$,
$w$ is not a descendant of edge $e$ but $e$ is not a separating edge for $x$ and $y$. Let $C$ be the strongly connected component of $G \setminus e$ that contains $x$ and $y$. From
Theorem \ref{cor:scc} and the assumption that $e=(u,v)$ is an ancestor of $x$ or $y$ in $D$ but not a separating edge for $x$ and $y$, we have that $C \subseteq D(v)$.
From Lemma~\ref{lemma:subtree2}, we have that there is a vertex $z \in C$ such that $C \subseteq H(z)$.
But then $x, y \in H(z)$, and thus $z$ is ancestor of $w$ in $H$ and a descendant of $e$ in $D$, which contradicts the assumption that $w \not\in D(v)$.
\end{proof}

Note that, by Lemma \ref{lemma:partition-paths}, a bridge $e$ of $G_s$ (resp., $G_s^R$) is not a separating edge for $x$ and $y$ only if it is an ancestor of both $x$ and $y$ in $D$ (resp., $D^R$).
This fact, combined with
Lemma~\ref{lemma:separating_two_vertices} suggests the following algorithm for computing, in an online fashion, all the separating edges for a given pair of query vertices $x$ and $y$.
First, we preprocess the loop nesting trees $H$ and $H^R$ in $O(n)$ time so that we can compute nearest common ancestors in constant time~\cite{nca:ht}.
To answer a reporting query for vertices $x$ and $y$, we compute their nearest common ancestor $w$ in $H$, and visit the bridges of $G_s$ that are ancestors of $x$ and $y$ in $D$ in a bottom-up order, as follows.
Starting from $x$ and $y$, we visit the bridges that are ancestors of $x$ and $y$ in $D$ until we reach a bridge $e=(u,v)$ such that $v$ is an ancestor of $w$ in $D$, or until we reach $s$ if no such bridge $e$ exists.
(As we showed above, if there is a bridge $e$ that is an ancestor of $x$ or $y$ in $D$ and is not a separating vertex for $x$ and $y$ then $e$ is an ancestor of both $x$ and $y$ in $D$.)
If $e=(u,v)$ exists then the bridges in $D[v,x] \cup D[v,y]$ are separating edges for $x$ and $y$. Otherwise
the bridges in $D[s,x] \cup D[s,y]$ are separating edges for $x$ and $y$.
We do the same for the reverse direction, i.e., compute the nearest common ancestor $w^R$ of $x$ and $y$ in $H^R$, and
visit the bridges of $G_s^R$ in a bottom-up order in $D^R$, starting from $x$ and $y$,
until we reach a bridge $e=(u,v)$ such that $v$ is an ancestor of $w^R$ in $D^R$, or until we reach $s$ if no such $e$ exists.
We finally return the union of the separating edges that we found in both directions.
To speed up this process, we can compute in a preprocessing step compressed versions of $D$ and $D^R$ that are formed by contracting, respectively, each vertex $u \in D$ into $r_u$, and each vertex $u \in D^R$ into $r^R_u$.

With the same data structure we can test in constant time, for any pair of query vertices $x$ and $y$ and a query edge $e$, if $e$ is a separating edge for $x$ and $y$. By Lemma~\ref{lemma:separating_two_vertices}, it suffices to find their nearest common ancestors, $w$ in $H$ and $w^R$ in $H^R$, and then test
whether $e$ is an ancestor of $x$ or $y$ and $w$ is not a descendant of $e$ in $D$, or whether $e$ is an ancestor of $x$ or $y$ and $w^R$ is not a descendant of $e$ in $D^R$. This gives the following theorem.

\begin{theorem}
Let $G$ be a strongly connected digraph with $m$ edges and $n$ vertices. After $O(m)$-time preprocessing, we can build
 an $O(n)$-space data structure, that can:
\begin{itemize}
\item Test if two query vertices $x$ and $y$ are $2$-edge-connected and if not report a corresponding separating edge.
\item Report all edges that separate two query vertices $x$ and $y$ in $O(k)$ time, where $k$ is the total number of separating edges reported.
For $k=0$, the time is $O(1)$.
\item Test in constant time if a query edge is a separating edge for a pair of query vertices.
\end{itemize}
\end{theorem} 
\section{Strongly connected components of $G\setminus u$}
\label{sec:vertices}

In this section we show how to extend the techniques of Section \ref{sec:scc} to $2$-vertex connectivity.
As before, we assume without loss of generality that $G$ is strongly connected, and we let $s \in V$ be an arbitrarily chosen start vertex in $G$.
Analogously to Section \ref{sec:scc}, we describe our
compact representation of the structure of all the $1$-vertex cuts (given by strong articulation points) of a strongly connected digraph $G$. In particular, we show that the four trees $D$, $D^R$,
 $H$ and $H^R$ are also sufficient to encode efficiently the decompositions that the strong articulation points induce in $G$, i.e., all the strongly connected components of $G\setminus u$, \emph{for all strong articulation points $u$} in $G$.
In particular, let $u$ be a strong articulation point in $G$. We will show how the four trees $D$, $D^R$,
 $H$ and $H^R$ can be effectively used to solve the following problems:

\begin{itemize}
\item Compute all the strongly connected components of $G\setminus u$;
\item Count the number of strongly connected components of $G\setminus y$;
 \item{Find the smallest or the largest strongly connected components of $G\setminus u$.}
\end{itemize}

Recall that a vertex $u \not= s$ is a nontrivial dominator of the flow graph $G_s$ (resp., $G^R_s$) if $u$ is not a leaf in the dominator tree $D$ (resp., $D^R$), and that we let $N$ (resp., $N^R$) denote the set of nontrivial dominators of  $G_s$ (resp., $G^R_s$).
We call a vertex $u \in N \cap N^R$ a \emph{common nontrivial dominator}.
For any vertex $u$ in $G$, we let $\widetilde{D}(u)$ (resp., $\widetilde{D}^R(u)$) denote the set of proper descendants of $u$ in $D$ (resp., $D^R$), i.e., $\widetilde{D}(u) = D(u) \setminus u$ (resp., $\widetilde{D}^R(u) = D^R(u) \setminus u$).
Also, we denote by $c(u)$ (resp., $c^R(u)$) the set of children of $u$ in $D$ (resp., $D^R$).
Clearly, $\widetilde{D}(u), c(u) \not= \emptyset$ (resp., $\widetilde{D}^R(u), c^R(u) \not= \emptyset$) if and only if either $u$ is a nontrivial dominator of $G_s$ (resp., $G_s^R$) or $u=s$.

Let $u$ be a strong articulation point of $G$. Consider the dominator relations in $G_s$ and $G^R_s$.
By Property \ref{property:strong-articulation-point} we have the following cases:
\begin{itemize}
\item[(a)] $u$ is a nontrivial dominator in $G_s$ but not in $G^R_s$, i.e., $\widetilde{D}(u) \neq \emptyset$ and $\widetilde{D}^R(u) = \emptyset$.
\item[(b)] $u$ is a nontrivial dominator in $G^R_s$ but not in $G_s$, i.e., $\widetilde{D}(u) = \emptyset$ and $\widetilde{D}^R(u) \neq \emptyset$.
\item[(c)] $u$ is a common nontrivial dominator of $G_s$ and $G^R_s$, or $u=s$, i.e., $\widetilde{D}(u) \neq \emptyset$ and $\widetilde{D}^R(u) \neq \emptyset$.
\end{itemize}

We identify the strongly connected component of $G \setminus u$ in each of these cases.
Consider first the case where $u$ is a nontrivial dominator in $G_s$ or $u=s$, i.e.,  $\widetilde{D}(u)\neq \emptyset$ and either (a) or (c) holds.
Case (b) is symmetric to (a).
By Lemma~\ref{lemma:paths-through-SAP}, the deletion of $u$ separates the proper descendants of $u$ in $D$, denoted by $\widetilde{D}(u)$, from $V \setminus D(u)$. Therefore, we can compute separately the strongly connected components of the subgraphs of $G \setminus u$ induced by $\widetilde{D}(u)$ and by $V \setminus D(u)$.
We begin with some lemmata that show how to compute the strongly connected components in the subgraph of $G \setminus u$ that is induced by $\widetilde{D}(u)$. Recall that $H(w)$ is the set of descendants of the vertex $w$ in the loop nesting tree of $G_s$.

\begin{lemma}
\label{lemma:vertices-subtree1}
Let $u \in N \cup \{s\}$.
For any vertex $w \in \widetilde{D}(u)$ the vertices in $H(w)$ are contained in a strongly connected component $C$ of $G \setminus u$ such that $C \subseteq \widetilde{D}(u)$.
\end{lemma}
\begin{proof}
Let $x \in \widetilde{D}(u)$ be a vertex such that $w=h(x) \in \widetilde{D}(u)$.
We claim that $w$ and $x$ are strongly connected in $G \setminus u$. Let $T$ be the dfs tree that generated the loop nesting tree $H$.
Note that, since $w=h(x)$ and $x,w\in \widetilde{D}(u)$,
the path $\pi_1$ from $w$ to $x$ in the dfs tree $T$ avoids the vertex $u$. To show that $w$ and $x$ are strongly connected in $G \setminus u$, we exhibit a path $\pi_2$ from $x$ to $w$ that avoids the vertex $u$.
Indeed, by the definition of the loop nesting forest, there is a path $\pi_2$ from $x$ to $w$ that contains only descendants of $w$ in $T$.
Note that $\pi_2$ cannot contain $u$ since $u$ is a proper ancestor of $w$ in $T$ and all descendants of $w$ in $T$ are descendants of $u$ in $T$.
Suppose, for contradiction, that either $\pi_1$ or $\pi_2$ contains a vertex $z \notin \widetilde{D}(u)$.
Then $u \not =s$, and by Lemma~\ref{lemma:paths-through-SAP}, it follows that either the subpath of $\pi_1$ from $z$ to $x$ or the subpath of $\pi_2$ from $z$ to $w$ contains $u$, a contradiction.
This implies that every pair of vertices in $H(x)$ is strongly connected in $G \setminus u$.
Let $C$ be the strongly connected component of $G \setminus u$ that contains $H(x)$.
The same argument implies that all vertices in $C$ are proper descendants of $u$ in $D$.
\end{proof}

\begin{lemma}
\label{lemma:vertices-subtree2}
Let $u \in N \cup \{s\}$.
For every strongly connected component $C$ in $G\setminus u$ such that $C \subseteq \widetilde{D}(u)$ there is a vertex $w \in C$ that is a common ancestor in $H$ of all vertices in $C$. Moreover, $C=H(w)$. 
\end{lemma}
\begin{proof}
Let $C$ be a strongly connected component that contains only proper descendants of $u$ in $D$.
Let $T$ be the dfs tree that generated $H$ and let $pre$ be the corresponding preorder numbering of the vertices.
Let $w$ be the vertex in $C$ with minimum preorder number with respect to $T$.
Consider any vertex $z \in C \setminus w$.
Since $C$ is a strongly connected component, there is a path from $w$ to $z$ that contains only vertices in $C$.
By the choice of $w$,  $\mathit{pre}(w)<\mathit{pre}(z)$, so Lemma \ref{lemma:dfs} implies that $w$ is an ancestor of $z$ in $T$.
Hence $w$ is also an ancestor of $z$ in $H$. 
Moreover, this implies $C \subseteq H(w)$.
Since $w \in \widetilde{D}(u)$, we have $H(w) \subseteq C$ by Lemma \ref{lemma:vertices-subtree1}.
Hence, $C=H(w)$.
\end{proof}

\begin{lemma}
\label{lemma:vertices-subtree}
Let $u \in N \cup \{s\}$.
Let $w$ be a vertex such that $w\in \widetilde{D}(u)$ and $h(w) \notin \widetilde{D}(u)$. Then, the subgraph induced by $H(w)$ is a strongly connected component in $G \setminus u$.
\end{lemma}
\begin{proof}
From Lemma \ref{lemma:vertices-subtree1} we have that $H(w)$ is contained in a strongly connected component $C \subseteq \widetilde{D}(u)$ of $G \setminus u$. Let $x \in C$ be the vertex that is a common ancestor in $H$ of all vertices in $C$, as stated by Lemma \ref{lemma:vertices-subtree2}. The fact that $h(w) \not\in \widetilde{D}(u)$ implies $x=w$. 
So, by Lemma \ref{lemma:subtree2}, $H(w)$ is a maximal subset of vertices that are strongly connected in $G \setminus u$. Thus $H(w)$ induces a strongly connected component of $G \setminus u$.
 \end{proof}

Next we consider the strongly connected components of the subgraph of $G \setminus u$ induced by $V \setminus D(u)$.

\begin{lemma}
\label{lemma:vertices-above-tree-1}
Let $u$ be a strong articulation point of $G$ that is a nontrivial dominator in $G_s$ but not in $G_s^R$ (i.e., $u \in N \setminus N^R$).
Let $C= V \setminus D(u)$. Then, the subgraph induced by $C$ is a strongly connected component of $G \setminus u$.
\end{lemma}
\begin{proof}
By Lemma~\ref{lemma:paths-through-SAP}(a), a vertex in $C$ cannot be strongly connected to a vertex in $\widetilde{D}(u)$ in $G \setminus u$.
Thus, it remains to show that the vertices in $C$ are strongly connected in $G\setminus u$.
Note that by the definition of subset $C$ we have that $s \in C$. Now it suffices to show that for any vertex $w \in C$, digraph $G$ has a path $\pi$ from $s$ to $w$ and a path $\pi'$ from $w$ to $s$ containing only vertices in $C$.
Suppose, for contradiction, that all paths in $G$ from $s$ to $w$ contain a vertex in $D(u)$. Then, Lemma~\ref{lemma:paths-through-SAP} implies that all paths from $s$ to $w$ contain $u$, which contradicts the fact that $w \not \in D(u)$. We use a similar argument for the paths from $w$ to $s$. If all such paths contain a vertex in $D^R(u)$, then by Lemma~\ref{lemma:paths-through-SAP} we have that all paths from $w$ to $s$ contain $u$. This implies that $u$ is a nontrivial dominator in $G_s^R$, clearly a contradiction.
\end{proof}

Finally we deal with the more complicated case (c).

\begin{lemma}
\label{lemma:vertices-above-tree-2}
Let $u$ be a strong articulation point of $G$ that is a common nontrivial dominator of $G_s$ and $G_s^R$ (i.e., $u \in N \cap N^R$).
Let $C = V  \setminus \big( D(u) \cup D^R(u) \big)$.
Then, the subgraph induced by $C$ is a strongly connected component of $G \setminus u$.
\end{lemma}
\begin{proof}
By the fact that $u$ is a strong articulation point and by Lemma \ref{lemma:paths-through-SAP}, we have that the following properties hold in $G$:
\begin{itemize}
\item[(1)]  There is a path from $s$ to any vertex in $V \setminus D(u)$ that does not contain $u$.
\item[(2)]  There is a path from any vertex in $V \setminus D^R(u)$ to $s$ that does not contain $u$.
\item[(3)] There is no edge $(x,y)$ such that $x \notin D(u)$ and $y \in \widetilde{D}(u)$. In particular, since $C \subseteq V \setminus D(u)$, there is no edge $(x,y)$ such that $x \in C$ and $y \in \widetilde{D}(u)$.
\item[(4)] Symmetrically, there is no edge $(x,y)$ such that $x \in \widetilde{D}^R(u)$ and $y \not \in D^R(u)$. In particular, since $C \subseteq V \setminus D^R(u)$, there is no edge $(x,y)$ such that $x \in \widetilde{D}^R(u)$ and $y\in C$.
\end{itemize}
Let $K$ be a strongly connected component in $G \setminus u$ such that $K \cap C \neq \emptyset$. By properties (3) and (4) we have that $K$ contains no vertex in $V \setminus C = D(u) \cup D^R(u)$. Thus, $K\subseteq C$. Let $G_C$ be the subgraph of $G$ induced by the vertices in $C$. We will show that for any vertex $x \in C$ the digraph $G_C$ contains a path from $s$ to $x$ and a path from $x$ to $s$. This implies that all vertices in $C$
are strongly connected in $G_C$, and hence also in $G \setminus u$. Since $K \subseteq C$ is a strongly connected component of $G\setminus u$, it must be $K = C$ thus yielding the lemma.

First we argue about the existence of a path from $s$ to $x \in C$ in $G_C$. Let $\pi$ be a path from $s$ to $x$ that does not contain $u$.
Property (1) guarantees that such a path exists.
Also, Lemma~\ref{lemma:paths-through-SAP} implies that $\pi$ does not contain a vertex in $D(u)$.
It remains to show that $\pi$ also avoids $\widetilde{D}^R(u)$.
Assume, for contradiction, that $\pi$ contains a vertex $z \in \widetilde{D}^R(u)$.
Choose $z$ to be the last such vertex in $\pi$.
Since $x \notin {D}^R(u)$ we have that $z \neq x$.
Let $w$ be the successor of $z$ in $\pi$.
From the fact that $\pi$ does not contain vertices in $D(u)$ and by the choice of $z$ it follows that $w \in C$. But then, edge $(z,w)$ violates property (4), a contradiction.
We conclude that path $\pi$ also exists in $G_C$ as claimed.

The argument for the existence of a path from $x \in C$ to $s$ in $G_C$ is symmetric.
Let $\pi$ be a path from $x$ to $s$ that does not contain $u$.
From property (2) and Lemma~\ref{lemma:paths-through-SAP} we have that such a path exists and does not contain a vertex in $D^R(u)$.
Now we show that $\pi$ also avoids $\widetilde{D}(u)$.
Assume, for contradiction, that $\pi$ contains a vertex $z \in \widetilde{D}(u)$.
Choose $z$ to be the first such vertex in $\pi$.
Since $x \notin \widetilde{D}(u)$ we have that $z \neq x$.
Let $w$ be the predecessor of $z$ in $\pi$.
From the fact that $\pi$ does not contain vertices in $D^R(u)$ and by the choice of $z$ it follows that $w \in C$.
So, edge $(w,z)$ violates property (3), a contradiction. Hence path $\pi$ also exists in $G_C$.
\end{proof}

\begin{lemma}
\label{lemma:common-nontrivial}
Let $u$ be a strong articulation point of $G$ that is a common nontrivial dominator of $G_s$ and $G_s^R$.
Let $C$ be a strongly connected component of $G \setminus u$ that contains a vertex in $\widetilde{D}(u) \cap \widetilde{D}^R(u)$. Then, $C \subseteq \widetilde{D}(u) \cap \widetilde{D}^R(u)$.
\end{lemma}
\begin{proof}
Consider any two vertices $x$ and $y$ such that $x \in \widetilde{D}(u) \cap \widetilde{D}^R(u)$ and
 $y \not\in \widetilde{D}(u) \cap \widetilde{D}^R(u)$.
We claim that $x$ and $y$ are not strongly connected in $G \setminus u$, which implies the lemma.
To prove the claim, note that by  Lemma \ref{lemma:vertices-above-tree-2}, $x$ is not strongly connected in $G \setminus u$ with any vertex in $V \setminus \big ( \widetilde{D}(u) \cup \widetilde{D}^R(u)  \big)$.
Hence, we can assume that $y \in \widetilde{D}(u) \setminus \widetilde{D}^R(u)$ or $y \in \widetilde{D}^R(u) \setminus \widetilde{D}(u)$. In either case, $x$ and $y$ are not strongly connected in $G \setminus u$ by Lemma \ref{lemma:paths-through-SAP}.
\end{proof}

The following theorem summarizes the results of Lemmata \ref{lemma:vertices-subtree}, \ref{lemma:vertices-above-tree-1}, \ref{lemma:vertices-above-tree-2}, and \ref{lemma:common-nontrivial}.

\begin{theorem}
\label{theorem:vertices-scc}
Let $u$ be a strong articulation point of $G$, and let $s$ be an arbitrary vertex in $G$.
Let $C$ be a strongly connected component of $G \setminus u$. Then one of the following cases holds:
\begin{itemize}
\item[(a)] If $u$ is a nontrivial dominator in $G_s$ but not in $G_s^R$ then either $C \subseteq \widetilde{D}(u)$ or $C = V \setminus D(u)$.
\item[(b)] If $u$ is a nontrivial dominator in $G_s^R$ but not in $G_s$ then either $C \subseteq \widetilde{D}^R(u)$ or $C = V \setminus D^R(u)$.
\item[(c)] If $u$ is a common nontrivial dominator of $G_s$ and $G_s^R$ then either $C \subseteq \widetilde{D}(u) \setminus \widetilde{D}^R(u)$, or $C \subseteq \widetilde{D}^R(u) \setminus \widetilde{D}(u)$, or $C \subseteq \widetilde{D}(u) \cap \widetilde{D}^R(u)$, or $C = V \setminus \big( D(u) \cup D^R(u) \big)$.
\item[(d)] If $u = s$ then $C \subseteq \widetilde{D}(u)$.
\end{itemize}
Moreover, if  $C \subseteq \widetilde{D}(u)$ (resp., $C \subseteq \widetilde{D}^R(u)$) then $C=H(w)$ (resp., $C=H^R(w)$) where $w$ is a vertex in $\widetilde{D}(u)$ (resp., $\widetilde{D}^R(u)$) such that $h(w) \not \in \widetilde{D}(u)$ (resp.,  $h^R(w) \not \in \widetilde{D}^R(u)$).
\end{theorem}

In the remainder of this section we consider the effect of removing a strong articulation point $u \not =s$. If $u = s$ then we can
compute in $O(n)$ time the properties related to $G\setminus s$ for all the problems considered in this section (i.e., finding all strongly connected components of $G\setminus s$, counting the number of strongly connected components of $G\setminus s$, finding the smallest/largest strongly connected component in $G\setminus s$), as suggested by Theorem \ref{theorem:vertices-scc}.
Indeed, for each child $x$ of $s$ in $H$, the vertices in $H(x)$ induce a strongly connected component in $G\setminus s$. In particular, we note the following:

\begin{corollary}
\label{cor:start-vertex}
Let $G$ be a strongly connected digraph, and let $s$ be an arbitrary start vertex of $G$.
Then, $s$ is a strong articulation point of $G$ if and only if $s$ has at least two children in $H$.
\end{corollary}

\subsection{Finding all strongly connected components of $G\setminus u$}
\label{sec:vertices-all-scc}

In this section we show how to exploit  Theorem \ref{theorem:vertices-scc}  to answer the following reporting query in a strongly connected graph $G$:
``Return all the strongly connected components of $G \setminus u$, for each vertex $u$.''. As it was previously mentioned, we only need to consider the case where $u$ is a strong articulation point and $u \not= s$.
Note that the simple-minded solution is to compute from scratch the strongly connected components of $G \setminus u$, for each strong articulation point $u$, which takes $O(mn)$ time. After a linear-time preprocessing, our algorithm will report the vertices of each strongly connected component of $G \setminus u$ in asymptotically optimal $O(n)$ time. Thus, we can output in a total of $O(m+np)$ worst-case time the strongly connected components of $G\setminus u$, for each strong articulation point $u$, where $p$ is the total number of strong articulation points in $G$.

We do this as follows. First, we process the dominator trees $D$ and $D^R$ in $O(n)$ time, so that we can test the ancestor/descendant relation in each tree in constant time~\cite{domin:tarjan}. Next we answer the query by executing a preorder traversal of the loop nesting trees $H$ and $H^R$. During these traversals, we will assign a label $scc(v)$ to each vertex $v$ that specifies the strongly connected component of $v$ in $G\setminus u$: namely, each strongly connected component will consist of vertices with the same label. We initialize $scc(v)=v$ for all vertices $v \neq u$.
Then we execute a preorder traversal of $H$ which will identify the strongly connected components of all vertices in $G\setminus u$, except for vertices in $\widetilde{D}^R(u) \setminus \widetilde{D}(u)$ and for vertex $u$; the strongly connected components of vertices in $\widetilde{D}^R(u) \setminus \widetilde{D}(u)$ will be discovered during a preorder traversal of $H^R$.
During our preorder traversals of $H$ and $H^R$ we will update $\mathit{scc}(v)$ for all vertices $v\notin \{s, u\}$. 
Throughout, we will always have $\mathit{scc}(s)=s$ while $\mathit{scc}(u)$ will be undefined.
When we visit a vertex $w \not\in \widetilde{D}^R(u) \setminus \widetilde{D}(u)$, $w\neq s$ (note that both $u$ and $s$ are excluded from this set), we test if the condition $(w \in \widetilde{D}(u) \wedge h(w) \not\in \widetilde{D}(u))$ holds.
If it does, then the label of $w$ remains $scc(w)=w$, otherwise we set $scc(w) = scc(h(w))$.
This process assigns $scc(x)=s$ to all vertices $x \in C = V \setminus \big ( D(u) \cup D^R(u) \big )$.
Also, all vertices $x \in H(w)$, where $w \in \widetilde{D}(u)$ and $h(w) \not\in \widetilde{D}(u)$ are assigned $scc(x) = w$.
Finally, we need to assign appropriate labels for the vertices in $\widetilde{D}^R(u) \setminus \widetilde{D}(u)$ (note again that both $u$ and $s$ are excluded).
We do that by executing a similar procedure on $H^R$.
This time, when we visit a vertex $w \in \widetilde{D}^R(u) \setminus \widetilde{D}(u)$ we test if the condition $(w \in \widetilde{D}^R(u) \wedge h^R(w) \not\in \widetilde{D}^R(u))$ holds.
If it does, then the label of $w$ remains $scc(w)=w$, otherwise we set $scc(w) = scc(h^R(w))$.
At the end of this process we have that all vertices $x \in H^R(w)$, such that $w \in \widetilde{D}^R(u)$ and $h(w) \not\in \widetilde{D}^R(u)$, are assigned $scc(x) = w$.
Thus, by Theorem~\ref{theorem:vertices-scc}, we have assigned correct labels to all vertices. Figure \ref{figure:example1V} shows the result of a reporting query for the digraph of Figure \ref{figure:example1ab} after the removal of vertex $f$.

\begin{figure}[t!]
\begin{center}
\centerline{\includegraphics[trim={0 0 0 0cm}, clip=true, width=1.4\textwidth]{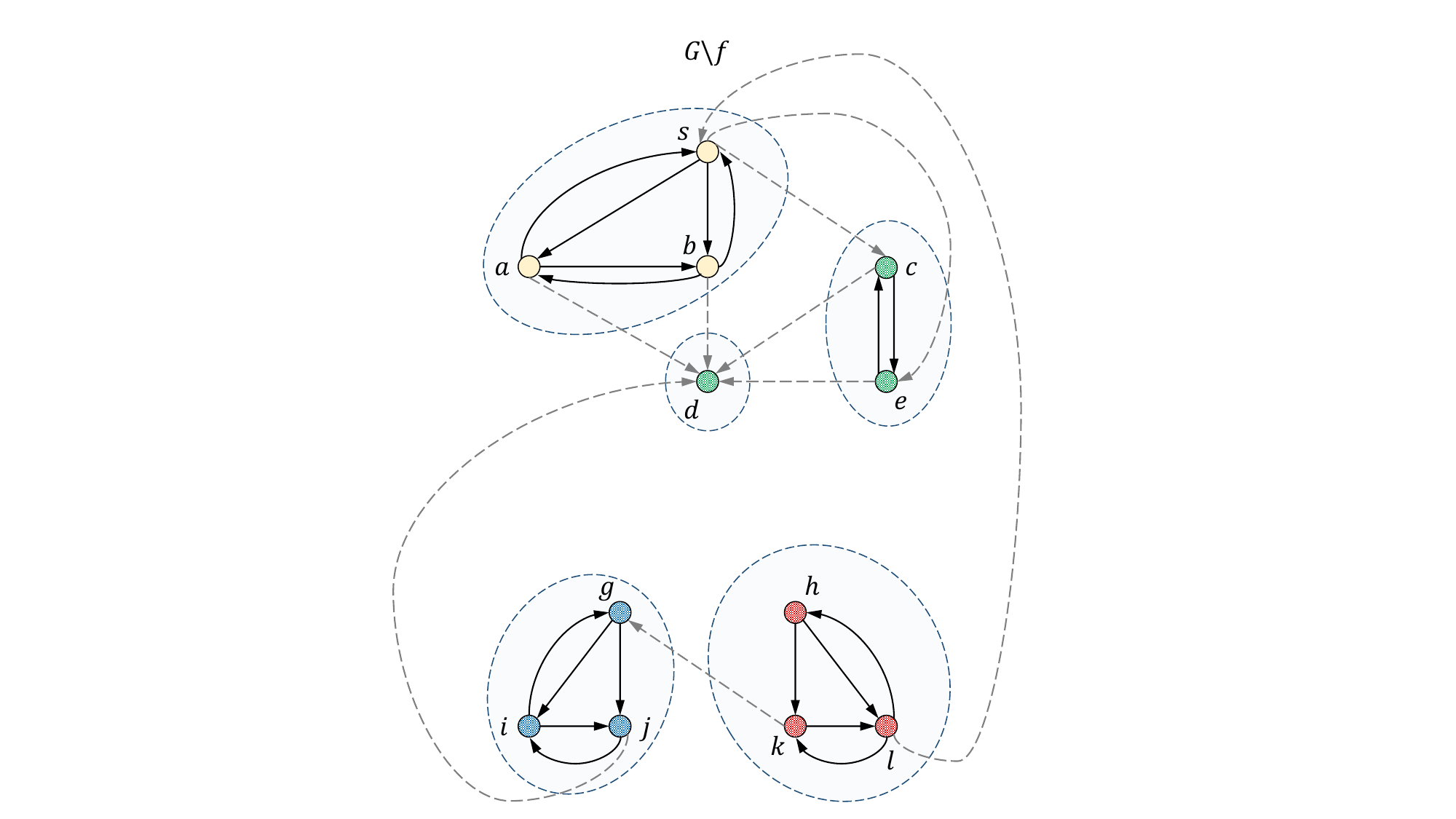}	}
\caption{The strongly connected components of $G \setminus f$ where $G$ is the graph of Figure \ref{figure:example1ab}.}
\label{figure:example1V}
\end{center}
\end{figure}

\begin{theorem}
	\label{theorem:all-scc-v}
	Let $G$ be a strongly connected digraph with $n$ vertices and $m$ edges. We can preprocess $G$ in $O(m)$ time and construct an $O(n)$-space data structure so that, given
	a vertex $u$ of $G$, we can report in $O(n)$ time all the strongly connected components of $G\setminus u$.
\end{theorem}

\begin{corollary}
	\label{corollary:all-scc-v}
Let $G$ be a strongly connected digraph with $n$ vertices, $m$ edges and $p$ strong articulation points. We can output the strongly connected components of $G\setminus u$, for all strong articulation points $u$ in $G$, in a total of $O(m+np)$ worst-case time.
\end{corollary}

\subsection{Counting the number of strongly connected components of $G\setminus u$}
\label{sec:vertices-SCCs}

Let $G=(V,E)$ be a strongly connected graph digraph.
In this section we consider the problem of computing the total number of strongly connected components obtained after the deletion of any vertex from $G$.
Note that we need to consider only vertices that are strong articulation points in $G$; indeed if $u$ is not a strong articulation point, then $G\setminus u$ has exactly the same strongly connected components as $G$.
Specifically, we show that by applying a \emph{vertex-splitting} transformation, we can reduce in $O(n)$ time this problem to the computation of the number of the strongly connected components obtained after the deletion of any strong bridge.
Using our results from Section~\ref{sec:SCCs-num}, this provides a linear-time algorithm for computing the total number of the strongly connected components obtained after the deletion of a single vertex, for all vertices in $G$.
Similarly to Section~\ref{sec:SCCs-num}, our bound is tight and it improves sharply over the simple-minded $O(mn)$ solution, which computes from scratch the strongly connected components of $G \setminus u$ for each strong articulation point $u$.

Let $G=(V,E)$ be the input strongly connected digraph, and let $s \in V$ be an arbitrary start vertex.
For the reduction, we construct a new graph $\overline{G}=(\overline{V},\overline{E)}$ that results from $G$ by applying the following transformation.
For each strong articulation point $x \not=s$ of $G$ (i.e., nontrivial dominator in $G_s$ or $G^R_s$), we add an auxiliary vertex $\overline{x} \in \overline{V}$ and add the auxiliary edge $(\overline{x},x)$. Then, we replace each edge $(u, x)$ entering $x$ in $G$ with an edge $(u,\overline{x})$ entering $\overline{x}$ in $\overline{G}$. 
The edges outgoing from $x$ in $G$ remain outgoing from $x$ in $\overline{G}$.
Note that this transformation maintains the strong connectivity of $G$. Call the vertices of $V$ ordinary.
The resulting graph $\overline{G}$ has at most $n-1$ auxiliary vertices and at most $n-1$ auxiliary edges (since the start vertex $s$ is not split).
Hence, $|\overline{V}| \le 2n-1$ and $|\overline{E}| \le m+n-1$. See Figure \ref{figure:vertex-splitting-a} for an example.

\begin{figure}[t!]
	\begin{center}
\centerline{		\includegraphics[trim={0 0 0 0cm}, clip=true, width=1.1\textwidth]{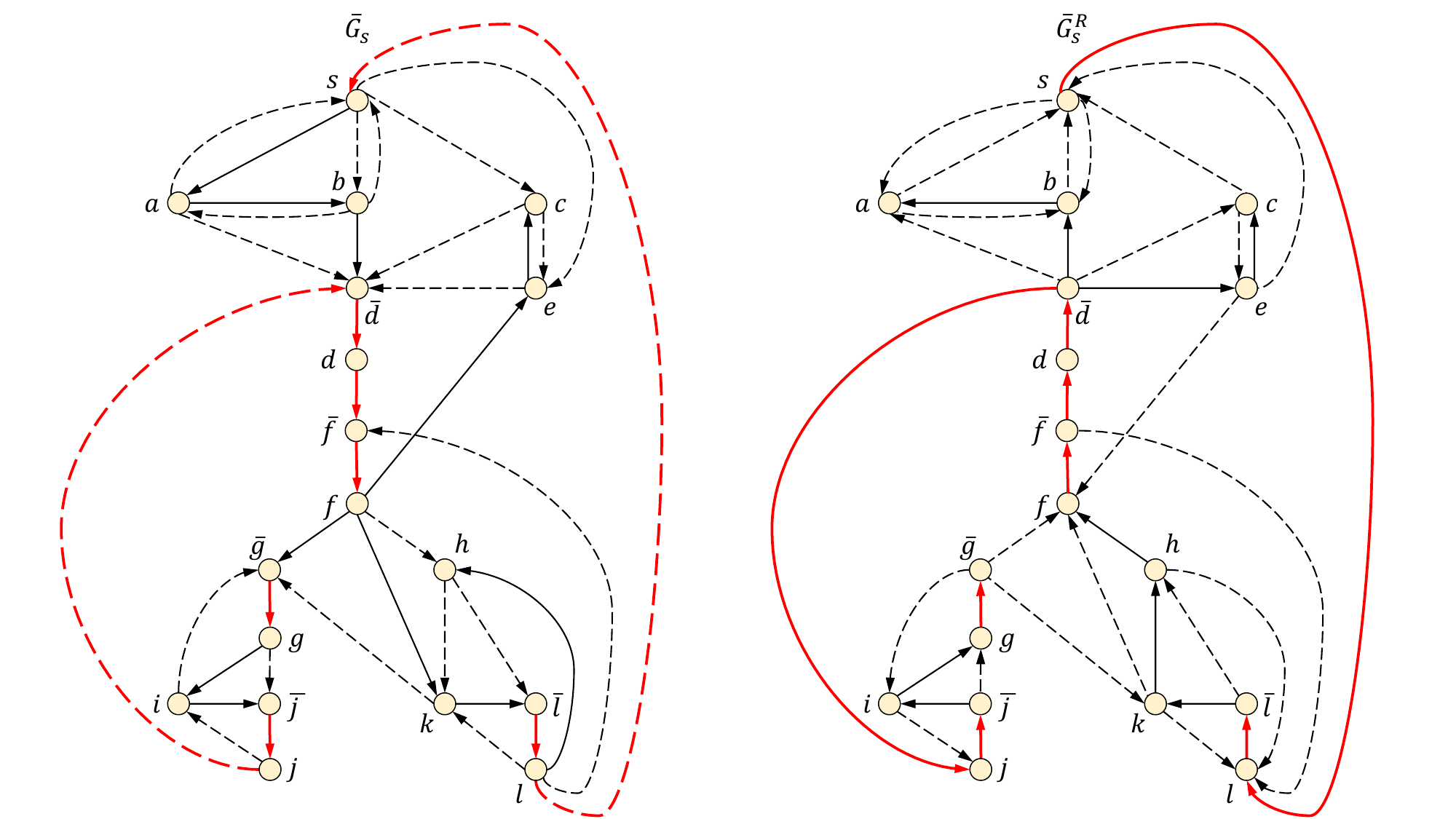}	}
		\caption{Flow graph $\overline{G}_s$ that results from the flow graph $G_s$ of Figure \ref{figure:example1ab} after splitting the nontrivial dominators $d$, $f$, $g$, $j$ and $l$; flow graph $\overline{G}^R_s$ is obtained from $\overline{G}_s$ after reversing the direction of all edges.
The solid edges in $\overline{G}_s$ and  $\overline{G}^R_s$ are the edges of depth first search trees with root $s$.
Strong bridges of $\overline{G}$ are shown red. (Better viewed in color.)}
		\label{figure:vertex-splitting-a}
	\end{center}
\end{figure}

Let $\overline{D}$ and $\overline{D}^R$ be the dominator trees of $\overline{G}_s$ and $\overline{G}^R_s$, respectively.
The following lemma states the correspondence between the strong articulation points in $G$ (except for $s$) and the strong bridges in $\overline{G}$. See Figure \ref{figure:vertex-splitting-bc}.

\begin{figure}[t!]
	\begin{center}
\centerline{		\includegraphics[trim={0 0 0 0cm}, clip=true, width=1.1\textwidth]{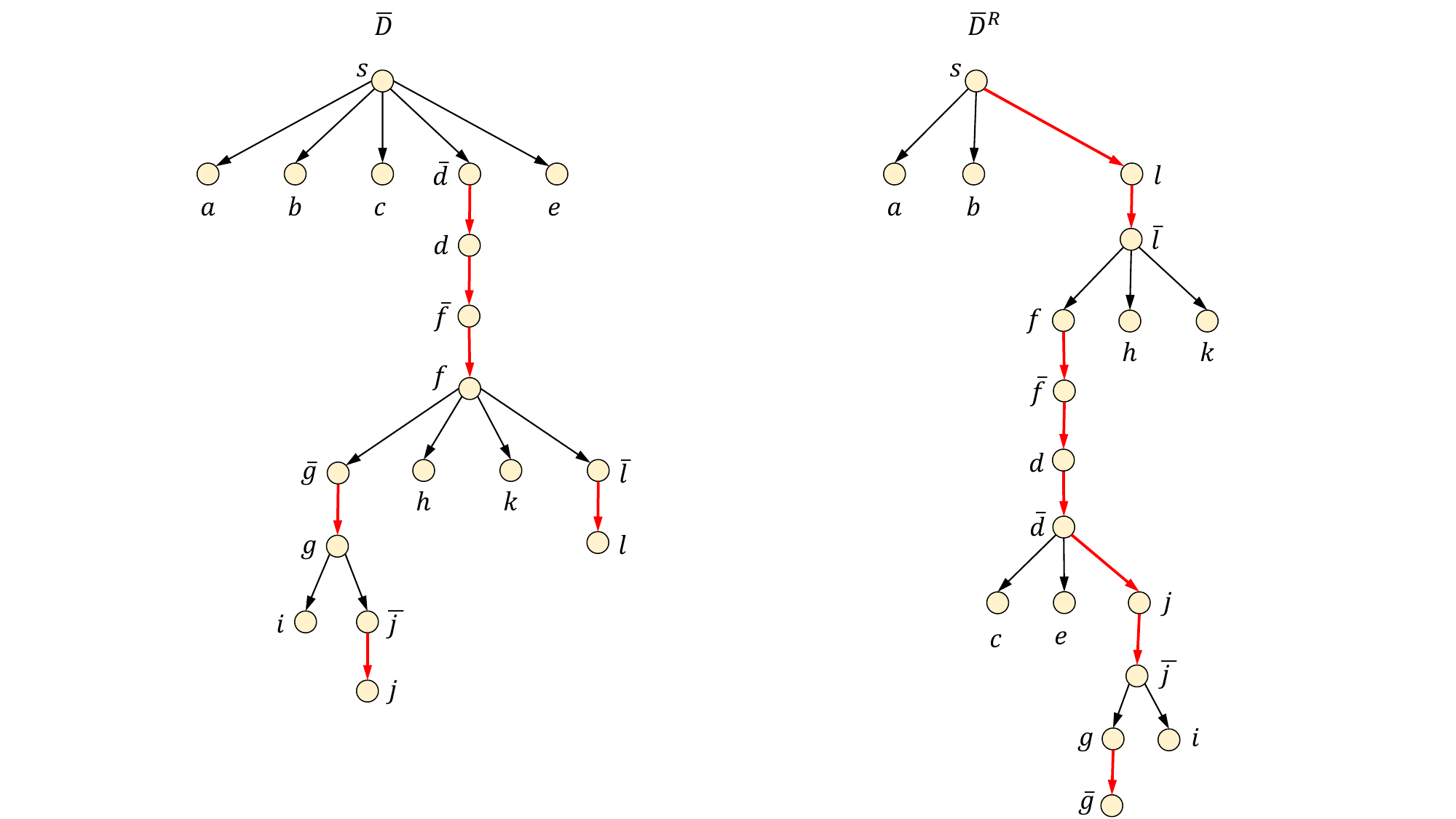}}
\centerline{	        \includegraphics[trim={0 0 0 5cm}, clip=true, width=1.1\textwidth]{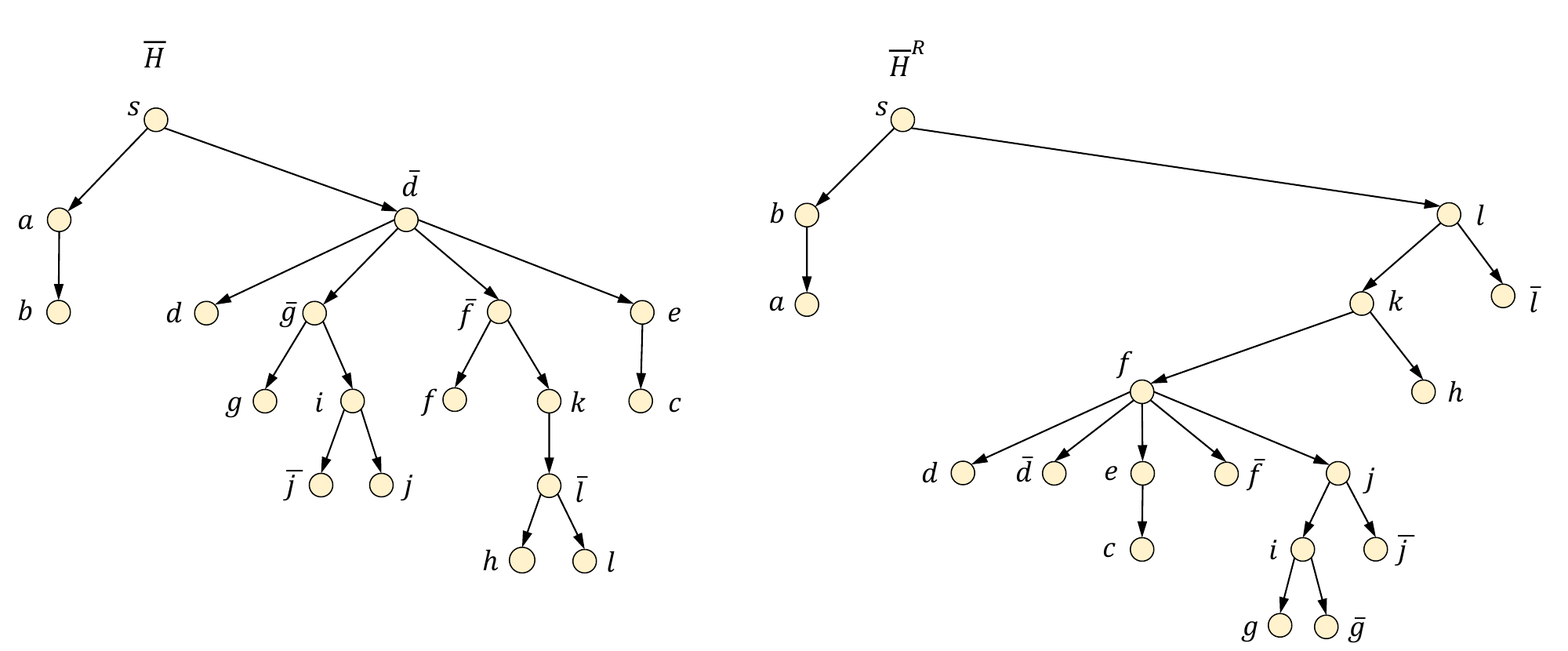}}
		\caption{Dominator trees $\overline{D}$ and $\overline{D}^R_s$, and loop nesting trees $\overline{H}$ and $\overline{H}^R$, of the flow graphs $\overline{G}_s$ and $\overline{G}^R_s$ of Figure \ref{figure:vertex-splitting-a}, respectively. The loop nesting forests are computed with respect to the dfs trees shown in Figure \ref{figure:vertex-splitting-a}.}
		\label{figure:vertex-splitting-bc}
	\end{center}
	\vspace{1.5cm}
\end{figure}

\begin{lemma}
Let $x \not =s$ be a strong articulation point of digraph $G$. Then the following hold:
	\begin{itemize}
		\item[(i)] Let $x$ be a nontrivial dominator of $G_s$ (resp., $G^R_s$).
		Then the auxiliary edge $(\overline{x},x)$ is a strong bridge of $\overline{G}$ and a bridge
		of $\overline{G}_s$ (resp., $\overline{G}^R_s$).
		\item[(ii)] For any ordinary vertex $u \in V \setminus x$, we have that $u \in D(x)$ (resp., $u\in D^R(x)$) if and only if $u \in \overline{D}(x)$ (resp., $u \in \overline{D}^R(\overline{x})$).
		\item [(iii)] All vertices in a strongly connected component of $G\setminus x$ are ordinary vertices in a strongly connected component of $\overline{G} \setminus (\overline{x},x)$.
	\end{itemize}
	\label{lemma:vertex-to-edge-reduction}
\end{lemma}
\begin{proof}
We prove (i) and (ii) only for the case where $x$ is a nontrivial dominator in $G_s$ since the case where $x$ is a nontrivial dominator in $G^R_s$ is completely analogous.
Note that $x$ has indegree one in $\overline{G}$, and hence $(\overline{x},x)$ is a strong bridge of $\overline{G}$.
By the fact that $\overline{G}$ is strongly connected, there is at least one path from $s$ to $x$. All such paths must contain $(\overline{x},x)$, so $(\overline{x},x)$ is
a bridge in $\overline{G}_s$.
From the above we have that $\overline{x}$ is the parent of $x$ in $\overline{D}$. This implies that any ordinary vertex $u \in D(x) \setminus x$
is a descendant of $x$ also in $\overline{D}(x)$. Thus, (i) and (ii) follow.
	
Now we prove (iii). Assume that two vertices $u$ and $v$, $u \not= v \not= x$, are in the same strongly connected component in $G \setminus x$ but not in the same strongly connected component in $\overline{G} \setminus (\overline{x},x)$.
The fact that $u$ and $v$ are not strongly connected in $\overline{G}\setminus (\overline{x},x)$ implies that either all paths from $u$ to $v$ contain $(\overline{x},x)$ or all paths from $v$ to $u$ contain $(\overline{x},x)$.
Without loss of generality, assume that all paths from $u$ to $v$ in $\overline{G}$ contain the strong bridge $(\overline{x},x)$.
Since $u$ and $v$ are in the same strongly connected component in $G\setminus x$, there is a path $\pi$ from $u$ to $v$ in $G$ that avoids all edges incoming to $x$ and all edges outgoing to $x$.
That means in $\overline{G} \setminus (\overline{x},x)$ the corresponding path $\overline{\pi}$ of $\pi$ avoids all incoming edges to $\overline{x}$ and all outgoing edges from $x$.
This is a contradiction to the fact that all paths from $u$ to $v$ in $\overline{G}$ contain $(\overline{x},x)$.
Now suppose that in a strongly connected component in $\overline{G} \setminus (\overline{x},x)$, there are two ordinary vertices $u$ and $v$, $u \not= v \not= x$, that are in different strongly connected components in $G \setminus x$.
Therefore, there is a path $\overline{\pi}$ in $\overline{G}$ that avoids $(\overline{x},x)$, and thus vertices $\overline{x}$ and $x$ by construction.
Then the corresponding path $\pi$ of $\overline{\pi}$ in $G$ avoids $x$, a contradiction.
\end{proof}

The next corollary is an immediate consequence of Lemma \ref{lemma:vertex-to-edge-reduction}.

\begin{corollary}
\label{corollary:vertex-to-edge-reduction-1}
Let $x$ be a nontrivial dominator of the flow graph $G_s$ (resp., $G_s^R$).
The number of strongly connected components in $G\setminus x$ is equal to the number of strongly connected components
that contain at least one ordinary vertex $u$ in $\overline{G}\setminus (\overline{x},x)$, where $u \not = x$.
Moreover, for every strongly connected component $C$ in $G \setminus x$ with $C \subseteq D(x)$ (resp., $C \subseteq D^R(x)$), it holds that
$C \subseteq \overline{D}(x)$ (resp., $C \subseteq \overline{D}^R(\overline{x})$) in $\overline{G}\setminus (\overline{x},x)$.
\end{corollary}

Given a flow graph $G_s=(V,E,s)$ of a strongly connected digraph $G$, Algorithm \textsf{SCCsDescendants}
of Section \ref{sec:SCCs-num}
computes the numbers $\#SCC(D(v))$ and $\#SCC(D^R(u))$ for each strong bridge $e=(u,v)$, i.e., the number of the strongly connected components that contain only vertices in $D(v)$ and $D^R(u)$, respectively.
On the other hand, Algorithm \textsf{SCCsCommonDescendants} of Section \ref{sec:SCCs-num} computes the numbers $\#SCC(D(v)\cap D^R(u))$ for each strong bridge $e=(u,v)$, i.e., the number of the strongly connected components that contain only vertices in $D(v) \cap D^R(u)$.
After computing the quantities  $\#SCC(D(v))$, $\#SCC(D^R(u))$, and $\#SCC(D(v)\cap D^R(u))$, by Corollary \ref{cor:nscc} we are able to answer queries on the number of the strongly connected components obtained after the deletion of any strong bridge $e=(u,v)$.

As suggested by Corollary \ref{corollary:vertex-to-edge-reduction-1}, we can apply Algorithms \textsf{SCCsDescendants} and \textsf{SCCsCommonDescendants} of Section \ref{sec:SCCs-num} on $\overline{G}$, in order to compute the total number of strongly connected components after any vertex deletion in $G$.
To do that, we have to guarantee that for a strong bridge $(\overline{x},x)$ such that $x$ is a nontrivial dominator of $G_s$ or $G_s^R$, only the strongly connected components that contain at least one ordinary vertex other than $x$ are accounted for.
We can accomplish this task by slightly modifying Algorithms \textsf{SCCsDescendants} and \textsf{SCCsCommonDescendants}.
To avoid confusion, we describe the modifications with respect to the notation that we used in Section \ref{sec:SCCs-num}, despite the fact that we execute the modified algorithm on the graph $\overline{G}$ that results from $G$ after vertex-splitting.

As we already mentioned, we need to count the number of strongly connected components, after any edge deletion, that contain at least one ordinary vertex.
However, for a bridge $(\overline{x},x)$ in $\overline{G}_s$ such that $x$ is a nontrivial dominator in $G_s$, we should ignore $x$ in $\overline{G}_s \setminus (\overline{x},x)$ since we are interested in the strong connectivity among vertices in $V \setminus x$.
Note that since $(\overline{x},x)$ is the only incoming edge to $x$ in $\overline{G}$, $\{x\}$ induces a singleton strongly connected component in $\overline{G} \setminus (\overline{x},x)$.

In Lines 8--11 (resp., Lines 12--15) of Algorithm \textsf{SCCsDescendants}, we update the values $bundle(r_x)$ and $bundle(r_{h(x)})$ (resp., $bundle^R(r^R_x)$ and $bundle^R(r^R_{h^R(x)})$) for a vertex $x$ such that $h(x)\notin D(r_x)$ (resp., $h^R(x)\notin D^R(r^R_x)$).
This indicates that for each bridge $e$ in the path $D[r_{h(x)},x]$ (resp., in the path $D^R[r^R_{h^R(x)},x]$), the graph $G \setminus e$ contains the strongly connected component $H(x)$ (resp., $H^R(x)$).
This fact follows from Theorem \ref{cor:scc}.
On a final step, namely in Lines 17--20 (resp., Lines 21--24), Algorithm \textsf{SCCsDescendants} propagates the $bundle$ values of the vertices over $D$ (resp., $D^R$) to compute the correct number $\#SCC(D(v))$ (resp., $\#SCC(D^R(u))$) for each strong bridge $e = (u,v)$.
We modify Line 8 (resp., Line 12) of the algorithm to take into account whether or not $H(x)$ (resp., $H^R(x)$) contains at least one ordinary vertex other than $v$.\footnote{This is indeed necessary. Consider for example the case where $G$ is directed cycle. After vertex-splitting, $\overline{G}$ is also a directed cycle, and the deletion of any edge makes each vertex a singleton strongly connected component.}
Thus, if for the vertex $x$ the set $H(x)$ (resp., $H^R(x)$) contains no ordinary vertex, then we do not update the values $bundle(r_x)$ and $bundle(r_{h(x)})$ (resp., $bundle^R(r^R_x)$ and $bundle^R(r^R_{h^R(x)})$).
As a special case, if $x$ is a nontrivial dominator in $D$ and $d(x)$ is the auxiliary vertex that resulted after applying vertex-splitting to $x$, we update the values $bundle(r_{d(x)})$ and $bundle(r_{h(x)})$, instead of $bundle(r_x)$ and $bundle(r_{h(x)})$, to avoid counting $\{x\}$ as a strongly connected component in $\overline{G} \setminus (d(x),x)$ for the reason explained earlier.
Note that this last case is not necessary in the reverse direction: for a strong bridge $(x,d^R(x))$, $x$ is the auxiliary vertex that resulted after applying vertex-splitting to $d^R(x)$,
and therefore the singleton strongly connected component $\{x\}$ is not counted.
This modified version of Algorithm \textsf{SCCsDescendants} still runs in $O(n)$ time. Indeed, we can precompute the number of ordinary vertices in $H(x)$, for all $x \in V$, in $O(n)$ time by executing a single bottom-up traversal on $H$, where we propagate the information about the number of ordinary vertices that are descendants of a vertex to its parent in $H$.

We also modify Algorithm \textsf{SCCsCommonDescendants}.
After vertex-splitting a nontrivial dominator $x$ into the strong bridge $(\overline{x},x)$, neither $\overline{x}$ nor $x$ is a common descendant of $(\overline{x},x)$.
Thus, we do not need to treat either $\overline{x}$ or $x$ as a special case as we did in the modification of the Algorithm \textsf{SCCsDescendants} shown above.
Recall that  Algorithm \textsf{SCCsCommonDescendants} works with the common bridge decomposition $\widetilde{\mathcal{D}}$ of $D$.
Now, we need to count only the number of strongly connected components that contain at least one ordinary vertex.
In Lines 11--13, Algorithm \textsf{SCCsCommonDescendants} adds the pair $\langle{}h(x),x\rangle{}$ to the list $L$ for each vertex $x$ such that $h(x)\notin D(\breve{r}_x)$.
This pair stores implicitly the information
that for each common bridge $e=(u,v)$ in the path $D[r_{h(x)}, x]$, where $e$ is a common bridge ancestor of $x$,
the vertices in $H(x) \subseteq D(v)\cap D^R(u)$ induce a strongly connected component in $G \setminus e$.
Indeed, this follows from Lemma \ref{lemma:common-strong-bridge-relation}.
Algorithm \textsf{SCCsCommonDescendants} then processes the pairs in Lines 17--33 in order to assign to each common bridge $e=(u,v)$ the number of
strongly connected components that contain common descendants of $e$ in $D$ and $D^R$, i.e.,  $\#SCC(D(v)\cap D^R(u))$.
We modify Algorithm \textsf{SCCsCommonDescendants} as follows: for each vertex $x$ such that $h(x)\notin D(\breve{r}_x)$, the modified algorithm
adds the pair $\langle h(x),x\rangle$ to the list $L$  only if $H(x)$ contains at least one ordinary vertex.
As in the case of Algorithm \textsf{SCCsDescendants},  we need only to precompute the number of ordinary vertices that are descendants of each vertex $u$ in $H$, which can be done in $O(n)$. Thus, the modified version of Algorithm \textsf{SCCsCommonDescendants}, executed on $\overline{G}$, still runs in $O(n)$ time.

Finally, we observe that we do not actually need to compute explicitly the digraphs $\overline{G}$ and $\overline{G}^R$.
Algorithms \textsf{SCCsDescendants} and \textsf{SCCsCommonDescendants} require as input the dominator trees $\overline{D}$ and $\overline{D}^R$, the loop nesting trees $\overline{H}$ and $\overline{H}^R$, and a list of the
bridges of $\overline{G}_s$ and $\overline{G}^R_s$. All this information can be computed directly from
the dominator trees ${D}$ and ${D}^R$, the loop nesting trees ${H}$ and ${H}^R$, and the
bridges of ${G}_s$ and ${G}^R_s$.
The bridges of $\overline{G}_s$ and $\overline{G}^R_s$ can be obtained by simply adding  all the auxiliary edges $(\overline{x},x)$ of $\overline{G}$ to the
bridges of $G_s$ and $G^R_s$.
Also, Lemma \ref{lemma:vertex-to-edge-reduction} implies that we can immediately construct $\overline{D}$ and $\overline{D}^R$ from $D$ and $D^R$, respectively, via vertex-splitting.
For the construction of $\overline{D}^R$ we have to exchange the role of $x$ and $\overline{x}$ for
each nontrivial dominator $x$. 
That is, we add the auxiliary edge $(x,\overline{x})$ and replace each edge $(x,u)$ leaving $x$ in $D^R$ with an edge $(\overline{x},u)$ leaving $\overline{x}$ in $\overline{D}^R$.
See Figure \ref{figure:vertex-splitting-bc}.
We now describe how to construct the loop nesting tree $\overline{H}$ from $H$ in $O(n)$ time.
Let $T$ be the dfs tree that generated $H$, and let $\overline{T}$ be the tree after setting $\overline{t}(\overline{x})=t(x)$ and $\overline{t}(x)=\overline{x}$ for every nontrivial dominator $x$.
The tree $\overline{T}$ will be the dfs tree of $\overline{G}_s$ that corresponds to $\overline{H}$.
It can be easily seen that $\overline{T}$ is a valid dfs traversal of $\overline{G}_s$.
The construction of $\overline{H}$ takes place in two phases. Initially we set $\overline{h}(x) = h(x)$ for every ordinary vertex $x$, so $\overline{H}=H$.
In the first phase, for every auxiliary vertex $\overline{x}$ we set $\overline{h}(\overline{x}) = h(x)$, and if $x$ has proper descendants in $H$ we set $\overline{h}(x) = \overline{x}$.
We now justify this first phase.
By construction, $\overline{x}$ is the parent of $x$ in $\overline{T}$.
Thus, the nearest ancestor $w$ of $x$ in $T$, to whom $x$ has path using only vertices in $T(w)$, is also the nearest ancestor of $\overline{x}$ with this property, thus $\overline{h}(\overline{x})=h(x)$. (We deal with the case where $\overline{h}(\overline{x})$ is an ordinary vertex of a nontrivial dominator in the second phase.)
Additionally, if $x$ has proper descendants in $H$ it means that there is a path from any $y \in H(x)$ to $x$ using only descendants of $x$ in $T$.
The corresponding path after vertex-splitting is a path from $y$ to $\overline{x}$ that contains only vertices in $\overline{T}(\overline{x})$, since $x \in \overline{T}(\overline{x})$ and $y \in \overline{T}(x)$.
Hence, there is a path from $x$ to $\overline{x}$ containing only vertices in $\overline{T}(\overline{x})$, so $\overline{h}(x) = \overline{x}$.
In the second phase of the algorithm, for every vertex $u$ such that $\overline{h}(u)$ is an ordinary vertex $x$ in $\overline{G}_s$, where $x$ is a nontrivial dominator, we set $\overline{h}(u)=\overline{x}$. This is correct since all incoming edges to $x$ in $G$ are incoming edges to $\overline{x}$ in $\overline{G}$ and vertex $\overline{x}$ is the parent of $x$ in $\overline{T}$.
The analogous procedure can be applied to construct
$\overline{H}^R$ from $H^R$, by swapping the roles of $\overline{x}$ and $x$, similar to the construction of $\overline{D}^R$.

Given the dominator trees $D$ and $D^R$, the loop nesting trees $H$ and $H^R$, and the strong bridges of $G$, the above reduction runs in $O(n)$ time.
Each of the trees $\overline{D}$, $\overline{D}^R$, $\overline{H}$, and $\overline{H}^R$ has at most $2n$ vertices, which implies that
we can compute the number of the strongly connected components after any vertex deletion in total $O(n)$ time, by
the modified Algorithms \textsf{SCCsDescendants} and \textsf{SCCsCommonDescendants}. The next theorem summarizes the result.

\begin{theorem}
Given the dominator trees $D$ and $D^R$, the loop nesting trees $H$ and $H^R$, and the common bridges of $G$, we can compute the number of the strongly connected components after the deletion of a strong articulation point in $O(n)$ time for all strong articulation points.
\end{theorem}

This yields immediately the following corollary.

\begin{corollary}
Given a directed graph $G=(V,E)$,  we can find in worst-case time $O(m+n)$ a strong articulation point $v$  in $G$ that maximizes/minimizes the total number of strongly connected components of $G\setminus v$.
\end{corollary}

As for the case of edge deletions, our results can be extended to the computation of more general functions:

\begin{theorem}
\label{theorem:generic-v}
Let $\odot$ be an associative and commutative binary operation such that its inverse operation $\odot^{-1}$ is defined and both $\odot$  and $\odot^{-1}$  are computable in constant time. Let $f(x)$ be a function defined on positive integers which can be computed in constant time. Given a strongly connected digraph $G$, we can compute in $O(m+n)$ time, for all vertices $v$, the function $f(|C_1|)\odot f(|C_2|)\odot ...\odot f(|C_k|)$, where $C_1,C_2,...,C_k$ are the strongly connected components in $G \setminus v$.
\end{theorem}

Note, in particular, that by setting $\langle \odot, \odot^{-1} \rangle = \langle +, - \rangle$ and $f(x)=x(x-1)/2$ we can compute the number of strongly connected pairs in $G \setminus u$, for all vertices $u$.
Hence, we obtain a linear-time algorithm to compute the \emph{most critical node} of a directed graph with respect to strong connectivity, i.e., the vertex $u$ of $G$ such that the number of strongly connected pairs of vertices in $G \setminus u$ is minimized. 
Based on our framework, \cite{Paudel:2018} presented and evaluated an alternative linear-time algorithm for computing the most critical node of a digraph that avoids vertex-spitting.

\subsection{Finding all the smallest and all the largest strongly connected components of $G\setminus u$}

\label{sec:vertices-minmax-scc}

Let $G$ be a strongly connected digraph. In this section we consider the problem of answering the following aggregate query: ``Find the size of the largest/smallest strongly connected component of $G \setminus u$, for all vertices $u$.'' Here ``size'' refers to the number of vertices in a strongly connected component, so the largest component (resp., smallest component) is the one with maximum (resp., minimum) number of vertices. Note again that the naive solution is to compute the strongly connected components of $G \setminus u$ for all strong articulation points $u$ of $G$, which takes $O(mn)$ time.
We  provide a linear-time algorithm for this problem.
Once we find such a strong articulation point $u$ that minimizes (resp., maximizes) the size of the largest (resp., smallest) strongly connected component of $G\setminus u$,
we can report the actual strongly connected components of $G\setminus u$ in $O(n)$ additional time by using the algorithm of Section \ref{sec:vertices-all-scc}.

Recall that for a subset of vertices $X \subseteq V$ we denote by $\mathit{LSCC}(X)$ the size of the largest strongly connected component in the subgraph of $G$ induced by $X$.
Also, for a strong articulation point $u$, we denote by $\mathit{LSCC}_u(V)$ the size of the largest strongly connected component of $G\setminus u$.
Theorem \ref{theorem:vertices-scc} immediately implies the following:
\begin{corollary}
\label{cor:vertices-lscc}
Let $u$ be a strong articulation point of $G$, and let $s$ be an arbitrary vertex in $G$.
The cardinality of the largest strongly connected component of $G \setminus u$ is equal to
\begin{itemize}
\item[(a)] $\max \{ \mathit{LSCC}(\widetilde{D}(u)), |V \setminus D(u)| \}$ when $u$ is a nontrivial dominator in $G_s$ but not in $G_s^R$.
\item[(b)] $\max \{ \mathit{LSCC}(\widetilde{D}^R(u)), |V \setminus D^R(u)| \}$ when $u$ is a nontrivial dominator in $G_s^R$ but not in $G_s$.
\item[(c)] $\max \{ \mathit{LSCC}(\widetilde{D}(u)), \mathit{LSCC}(\widetilde{D}^R(u)), |V \setminus\big( D(u) \cup D^R(u) \big)| \}$ when $u$ is a common nontrivial dominator of $G_s$ and $G_s^R$.
\item[(d)] $\mathit{LSCC}(\widetilde{D}(u))$ when $u=s$.
\end{itemize}
Moreover, $\mathit{LSCC}(\widetilde{D}(u)) = \max_{w} \{ |H(w)| \ : \ w \in \widetilde{D}(u) \mbox{ and } h(w) \not \in \widetilde{D}(u) \}$ and
$\mathit{LSCC}(\widetilde{D}^R(u)) = \max_{w} \{ |H^R(w)| \ : \ w \in \widetilde{D}^R(u) \mbox{ and } h^R(w) \not \in \widetilde{D}^R(u) \}$.
\end{corollary}

Now we develop an algorithm that applies Corollary \ref{cor:vertices-lscc}. Our algorithm, detailed below (see Algorithm \ref{algorithm:LSCC-V} for pseudocode), uses the dominator and the loop nesting trees of $G$ and its reverse $G^R$, with respect to a start vertex $s$, and computes
for each strong articulation point $u$ of $G$ the size of the largest strongly connected component of $G \setminus u$, denoted by $\mathit{LSCC}_u(V)$.
A small modification of the algorithm is able to compute the size of the smallest strongly connected component of $G \setminus u$.

\begin{algorithm}
	\LinesNumbered
	\DontPrintSemicolon
	\KwIn{Strongly connected digraph $G=(V,E)$}
	\KwOut{Size of the largest strongly connected component of $G\setminus u$ for each strong articulation point $u$ }
	
	\textbf{Initialization:}\;
	
		Compute the reverse digraph $G^R$.
Select an arbitrary start vertex $s \in V$. 	

	Compute the dominator trees $D$ and $D^R$ of the flow graphs $G_s$ and $G_s^R$, respectively.\;
	
	Compute the loop nesting trees $H$ and $H^R$ of the flow graphs $G_s$ and $G_s^R$, respectively.\;
	
	Compute the sets of nontrivial dominators $N$ and $N^R$ of the flow graphs $G_s$ and $G_s^R$, respectively.\;

\textbf{Process start vertex:}\;
\If{$s$ is a strong articulation point}
{
Compute $\mathit{LSCC}(\widetilde{D}(s))$\;
Set $\mathit{LSCC}_s(V) =\mathit{LSCC}(\widetilde{D}(s))$	
}
	\textbf{Process nontrivial dominators:}\;
	\ForEach{$u \in N$ in a bottom-up order of $D$}
	{
		\eIf{$u \not\in N^R$}
		{
			Compute $\mathit{LSCC}(\widetilde{D}(u))$\;
			Set $\mathit{LSCC}_u(V) = \max \{ \mathit{LSCC}(\widetilde{D}(u)), |V|-|D(u)| \}$
		}
		{					
						Compute $\mathit{LSCC}(\widetilde{D}(u))$, $\mathit{LSCC}(\widetilde{D}^R(u))$, and $|D(u) \cup D^R(u)|$\;
			Set $\mathit{LSCC}_u(V) =\max\{ \mathit{LSCC}(\widetilde{D}(u)), \mathit{LSCC}(\widetilde{D}^R(u)), |V|-|D(u)\cup D^R(u)| \}$
		}
	}
	
	\ForEach{$u \in N^R$ in a bottom-up order of $D^R$}
	{
		\If{$u \not\in N$}
		{
			Compute $\mathit{LSCC}(\widetilde{D}^R(u))$\;
			Set $\mathit{LSCC}_u(V) =\max\{\mathit{LSCC}(\widetilde{D}^R(u)), |V|-|D^R(u)|\}$
		}
	}
	\caption{\textsf{LSCC-V}}
		\label{algorithm:LSCC-V}
\end{algorithm}

As already mentioned, by Corollary \ref{cor:vertices-lscc} we can compute $\mathit{LSCC}_s(V)$ in $O(n)$ time by simply taking the maximum $|H(x)|$
among all children $x$ of $s$ in $H$.
So, in order to get an efficient implementation of our algorithm, we need to specify how to compute efficiently the following quantities:
\begin{itemize}
	\item[(a)] $\mathit{LSCC}(\widetilde{D}(u))$  for every nontrivial dominator $u$ of the flow graph $G_s$.
	\item[(b)] $\mathit{LSCC}(\widetilde{D}^R(u))$ for every nontrivial dominator $u$ of the flow graph $G_s^R$.
	\item[(c)] $|D(u)\cup D^R(u)|$ for every vertex $u$ that is a common nontrivial dominator of $G_s$ and $G_s^R$.
\end{itemize}

We deal with computations of type (a) first. The computations of type (b) are analogous. We precompute for all vertices $u$ the number of their descendants in the loop nesting tree $H$, and we initialize $\mathit{LSCC}(\widetilde{D}(u)) = 0$.
For every vertex $u\not =s$ in a bottom-up order of $D$ we update the current value of $\mathit{LSCC}(\widetilde{D}(d(u)))$ by setting
$$
\mathit{LSCC}(\widetilde{D}(d(u)))=\max\{\mathit{LSCC}(\widetilde{D}(d(u))), \mathit{LSCC}(\widetilde{D}(u)), |H(u)| \}.
$$
This computation takes $O(n)$ time in total.

Finally, we need to specify how to compute the values of type (c), i.e., how to compute the cardinality of the union $D(u)\cup D^R(u)$ for a strong articulation point $u \neq s$ that is a common nontrivial dominator of both flow graphs $G_s$ and $G_s^R$.
By the equality $|D(u)\cup D^R(u)| = |D(u)| + |D^R(u)| - |D(u)\cap D^R(u)| = |D(u)| + |D^R(u)| - ( |\widetilde{D}(u)\cap \widetilde{D}^R(u)| + 1)$, it suffices to show how to compute the cardinality of the intersection $\widetilde{D}(u)\cap \widetilde{D}^R(u)$.
We do this via a reduction to computing $D(v)\cap D^R(u)$ for all common bridges $(u,v)$.
This reduction uses the vertex-splitting reduction of Section \ref{sec:vertices-SCCs} and implies an $O(n)$-time solution for computing the values of type (c) with the algorithm
of Section~\ref{sec:minmax-scc}.
By doing that, we will obtain a linear-time implementation of Algorithm \textsf{LSCC-V}.

\subsubsection*{Counting the common descendants of strong articulation points}
\label{sec:vertices-common-descendants}

In this section we consider the problem of computing the number of common descendants of each strong articulation point $u \not= s$.
Specifically, given a flow graph $G_s$ of a strongly connected digraph $G$, and the dominator trees $D$ and $D^R$ of $G_s$ and $G_s^R$, respectively, we wish to compute the number of the vertices that are proper descendants of $u$ in both $D$ and $D^R$, i.e., $|\widetilde{D}(u) \cap \widetilde{D}^R(u)|$.
We show that by \emph{vertex-splitting} we can reduce in $O(n)$ time this problem to the problem of computing the number of the common descendants after the deletion of any strong bridge (Section
\ref{sec:minmax-scc}).

We recall here the vertex-splitting reduction of Section \ref{sec:vertices-SCCs}.
We add an auxiliary vertex $\overline{x}$ for each nontrivial dominator $x$ of $G_s$ or $G^R_s$,  together with the auxiliary edge $(\overline{x},x)$. Then, we replace each edge $(u, x)$ entering $x$ in $G$ with an edge $(u,\overline{x})$ entering $\overline{x}$.
Let $\overline{G}$ be the resulting digraph. Also let $\overline{D}$ and $\overline{D}^R$ be the dominator trees of $\overline{G}$ and $\overline{G}^R$, respectively.
The next corollary is an immediate consequence of Lemma \ref{lemma:vertex-to-edge-reduction}.

\begin{corollary}
\label{corollary:vertex-to-edge-reduction-2}
Let $x \not= s$ be a strong articulation point of $G$.
The number of common descendants $w \not= x$ of $x$ in $D$ and $D^R$ is equal to the number of ordinary common descendants in $\overline{D}$ and $\overline{D}^R$
of the corresponding common bridge $(\overline{x},x)$ of $\overline{G}_s$ and $\overline{G}^R_s$.
\end{corollary}

Note that a strong articulation point $u\not=s$ of $G$ is not a common descendant of the corresponding common bridge $(\overline{u},u)$, so our reduction indeed computes
$|\widetilde{D}(u) \cap \widetilde{D}^R(u)|$, as required.
Corollary \ref{corollary:vertex-to-edge-reduction-2} suggests that we can apply the algorithm of Section
\ref{sec:minmax-scc} on $\overline{G}$ in order to compute the number of common descendants in $D$ and $D^R$ of all common nontrivial dominators in $D$ and $D^R$.
We need to modify the algorithm of Section
\ref{sec:minmax-scc} in order to make sure that we count the number of ordinary common descendants of each common strong bridge $(\overline{x},x)$ of $\overline{G}$.
To avoid confusion, we describe the modification with respect to the notation used in Section
\ref{sec:minmax-scc}, despite the fact that we execute the modified algorithm on digraph $\overline{G}$.
Let $u$ be a vertex such that $h(u) \notin
\breve{D}_{\breve{r}_u}$.
By Lemma \ref{lemma:number-of-SCCs-ancestors}, the set of vertices $H(u)$ induces a strongly connected component in $G \setminus e$, for each bridge $e$ in the path $D[r_{h(u)},u]$.
In the algorithm of Section
\ref{sec:minmax-scc}, we add $|H(u)|$ to the number of the common descendants of each common bridge that is in the path $D[\breve{r}_{h(u)},u]$ and is a common ancestor of $u$.
This is done by inserting the triple $\langle h(u),u,|H(u)| \rangle$ into a list $L$.
The algorithm then processed the triples in $L$, and computes the number of common descendants for each strong bridge.
Thus, it is sufficient to replace $H(u)$ with the number of ordinary vertices in $H(u)$ in such a triple  $\langle h(u),u,|H(u)| \rangle$.
Note that we can precompute the number of the ordinary vertices in $H(u)$, for all $u$, in $O(n)$ time as we described in Section \ref{sec:vertices-SCCs}.
Moreover we can compute $\overline{D}$ and $\overline{H}$ in $O(n)$ time, as shown in Section \ref{sec:vertices-SCCs}.

\begin{lemma}
Given the dominator trees $D$ and $D^R$, the loop nesting trees $H$ and $H^R$, and the common bridges of $G$, we can compute the number of the common descendants of each vertex in $O(n)$ total time.
\end{lemma}

This gives the following results.

\begin{theorem}
Given the dominator trees $D$ and $D^R$, the loop nesting trees $H$ and $H^R$, and
the strong bridges of $G$, we can compute the size of the largest or the smallest strongly connected
component after the deletion of a strong articulation point in $O(n)$ time for all strong articulation points.
\end{theorem}

\begin{corollary}
Given a strongly connected directed graph $G$,  we can find in worst-case time $O(m+n)$ a  strong articulation point $v$ that minimizes/maximizes the size of the largest / smallest strongly connected component in $G\setminus v$.
\end{corollary}

\section{Computing $2$-vertex-connected components}
\label{sec:2VCBs}

\newcommand{\vr}{\leftrightarrow_{\mathrm{vr}}}

We can use our framework to obtain a new, simpler, linear-time algorithm for computing the $2$-vertex-connected components of a digraph. As in \cite{2VCB}, we do this by computing as an intermediary the following relation\footnote{We note that \cite{2VCB} refers to the $2$-vertex-connected components of a digraph as \emph{$2$-vertex-connected blocks}.}.
Vertices $x$ and $y$ are said to be \emph{vertex-resilient},
denoted by $x \leftrightarrow_{\mathrm{vr}} y$,
if the removal of any vertex different from $x$ and $y$
leaves $x$ and $y$ in the same strongly connected component. We define a \emph{vertex-resilient component}
of a digraph $G=(V,E)$ as a maximal subset $B \subseteq V$ such that $x \leftrightarrow_{\mathrm{vr}} y$
for all $x, y \in B$. See Figure \ref{figure:VRB-example}.
As a (degenerate) special case, a vertex-resilient component might consist of a singleton vertex only, in which case we have a \emph{trivial vertex-resilient component}.
Here, we consider only nontrivial vertex-resilient components, and since there is no danger of ambiguity, we will call them simply vertex-resilient components.

\begin{figure}[t!]
\begin{center}
\centerline{\includegraphics[trim={0 0 0 0cm}, clip=true, width=1.2\textwidth]{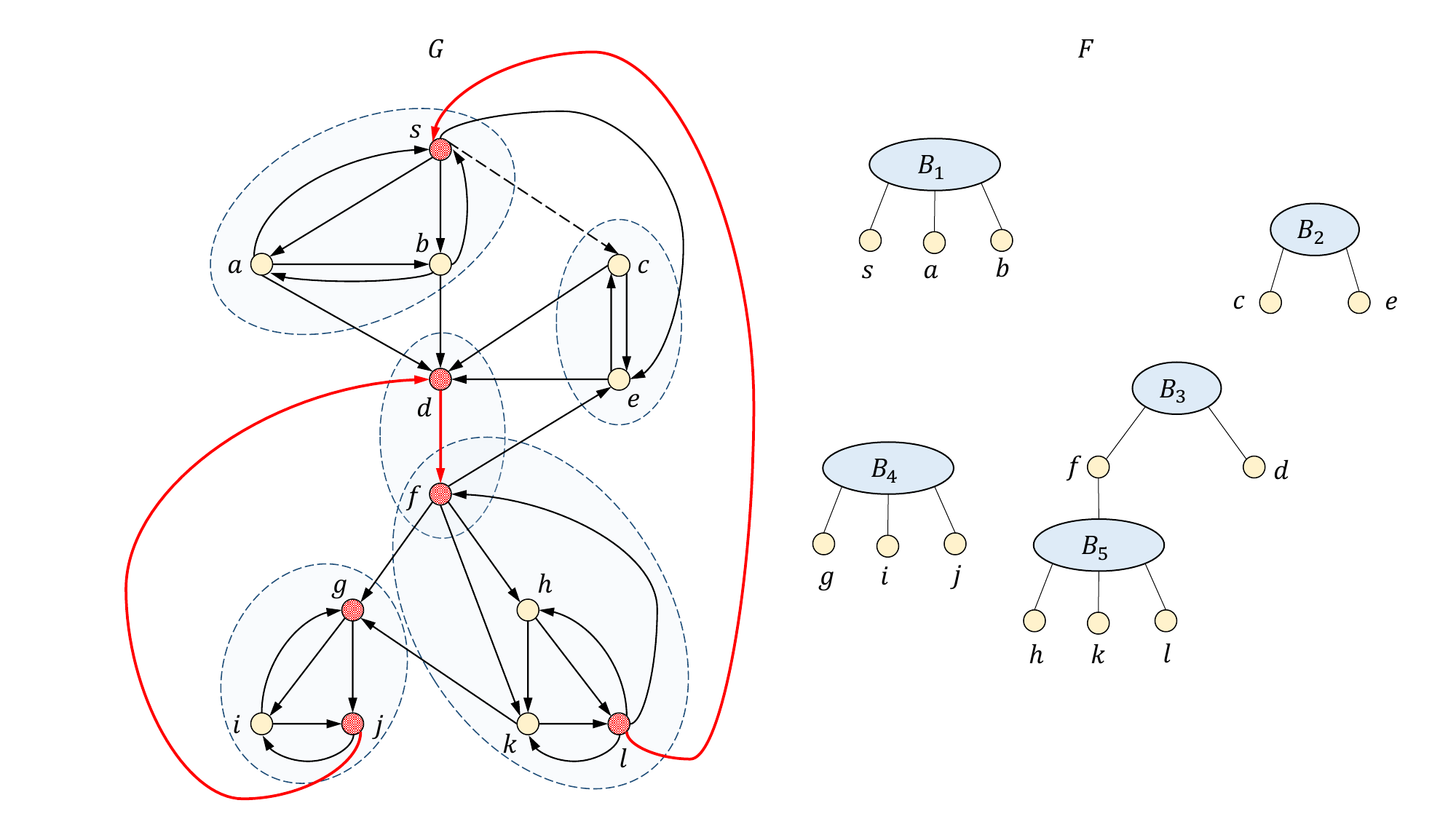}	}
\caption{The vertex-resilient components of the digraph of Figure \ref{figure:example1ab}, and the corresponding component forest. Strong bridges and strong articulation points of $G$ are shown red. (Better viewed in color.)}
\label{figure:VRB-example}
\end{center}
\end{figure}

The following lemma states that we can compute the $2$-vertex-connected components from the $2$-edge-connected components and the vertex-resilient components.
Moreover, this computation can be done in $O(n)$ time~\cite{2VCB}.

\begin{lemma} \emph{(\cite{2VCB})}
\label{cor:2vc-resilient}
For any two distinct vertices $x$ and $y$, $x \leftrightarrow_{\mathrm{2v}} y$ if and only if $x \leftrightarrow_{\mathrm{vr}} y$ and $x \leftrightarrow_{\mathrm{2e}} y$.
\end{lemma}

\begin{lemma} \emph{(\cite{2VCB})}
\label{lemma:vertex-resilient-necessary}
Let $G=(V,E)$ be a strongly connected digraph, and let $s \in V$ be an arbitrary start vertex.
Any two vertices $x$ and $y$ are vertex-resilient only if the following condition holds in both $D$ and $D^R$:
$x$ and $y$ are either siblings or one is the parent of the other.
\end{lemma}

Recall that, for any vertex $u$, $c(u)$ (resp., $c^R(u)$) denotes the set of children of $u$ in $D$ (resp., $D^R$). For any pair of vertices $u$ and $v$ we
let $c(u,v) = ( c(u) \cup \{u\} ) \cap ( c^R(v) \cup \{v\} )$. That is,
set $c(u,v)$ contains all vertices in $c(u) \cap c^R(v)$. Also, if $u=v$ or $u \in c^R(v)$ then $u \in c(u,v)$, and if $v \in c(u)$ then $v \in c(u,v)$.
We can compute all nonempty $c(u,v)$ sets in $O(n)$ time~\cite{LuigiGILP15}. The following corollary is an immediate consequence of Lemma \ref{lemma:vertex-resilient-necessary}.

\begin{corollary}
\label{corollary:vertex-resilient-necessary}
Let $G=(V,E)$ be a strongly connected digraph, and let $s \in V$ be an arbitrary start vertex.
Any two vertices $x$ and $y$ are vertex-resilient only if they are located in a common set $c(u,v)$.
\end{corollary}

Let $x$ and $y$ be two distinct vertices in $G$. We say that a strong articulation point $u$ \emph{separates $x$ and $y$} (or equivalently that $u$ is a \emph{separating vertex for $x$ and $y$})
if all paths from $x$ to $y$ or all paths from $y$ to $x$ contain vertex $u$ (i.e., $x$ and $y$ belong to different strongly connected components of $G\setminus u$).
By definition,  $x$ and $y$ are vertex-resilient if and only if there exists no separating vertex for them.
Our new algorithm for the vertex-resilient components is based on Theorem \ref{theorem:vertices-scc}, together with some lemmata, stated below,
that identify the strong articulation points that are candidate
separating vertices of a given pair of vertices.

\begin{lemma}
\label{vertices-separators-ancestors}
Let $u$ be a strong articulation point that is a separating vertex for vertices $x$ and $y$. Then $u$ must appear in at least one of the paths $D[s,x]$, $D[s,y]$, $D^R[s,x]$, and $D^R[s,y]$.
\end{lemma}
\begin{proof}
The lemma trivially holds  for $u=s$. So assume, by contradiction,
that a strong articulation point $u \not= s$ separates $x$ and $y$ but it does not appear in any of the paths $D[s,x]$, $D[s,y]$, $D^R[s,x]$, and $D^R[s,y]$.
The fact that $u \not\in D^R[s,x]$ implies that there is a path $\pi$ in $G$ from $x$ to $s$ that avoids $u$.
Similarly, the fact that $u \not \in D[s,y]$ implies that there is a path $\pi'$ in $G$ from $s$ to $y$ that avoids $u$.
Then $\pi \cdot \pi'$ is a path in $G$ from $x$ to $y$ that does not contain vertex $u$.
Analogously, the fact that $u \not\in D^R[s,y]$ and $u \not \in D[s,x]$ implies that there is a path in $G$ from $y$ to $x$ that does not contain $u$.
This contradicts the assumption that $u$ separates $x$ and $y$, i.e., that $x$ and $y$ are not strongly connected in $G\setminus u$.
\end{proof}

\begin{lemma}
\label{SAP-relation}
Let $x$ and $y$ be vertices such that $x \not= s$ and $y$ is either a sibling of $x$ or a child of $x$ in $D$. Also, let $w=d(x)$.
A strong articulation point $u$ that is not a descendant of $w$ in $D$ is a separating vertex for $x$ and $y$ only if $w$ is a separating vertex for $x$ and $y$.
\end{lemma}
\begin{proof}
Let $u$ be a strong articulation point that is not a descendant of $w$ in $D$ and that separates $x$ and $y$ in $G$.
Since $u \not \in D(w)$ we have that $w \not =s$.
The fact that $u$ separates $x$ and $y$ implies that all paths from $x$ to $y$ or all paths from $y$ to $x$ contain $u$.
Without loss of generality, assume that all paths from $x$ to $y$ contain $u$ (otherwise swap $x$ and $y$ in the following).
Observe that any path from $x$ to $y$ containing $u$ must contain a path from $u$ to $y$ as a subpath.
The fact that $u$ is not a descendant of $w$ in $D$ and Lemma \ref{lemma:paths-through-SAP} imply that any path from $u$ to $y$ must contain $w$.
Therefore if any path from $x$ to $y$ contains $u$, it must also contain $w$.
\end{proof}

Our algorithm uses an additional data structure from \cite{2VCB}, the \emph{component forest} $F=(V_F, E_F)$ of $G$, defined as follows. The vertex set $V_F$ consists of the vertices in $V$ and also contains one \emph{component node} for each vertex-resilient component of $G$. The edge set $E_F$ consists of the edges $\{x, B\}$ for every vertex $x \in V$ and every component $B$ such that $x \in B$.

\begin{lemma} \emph{(\cite{2VCB})}
\label{lemma:blocks-number}
The number of vertex-resilient components in a digraph $G$ with $n$ vertices is at most $n-1$. Moreover, the total number of vertices in all vertex-resilient components is at most $2n-2$.
\end{lemma}

Note that, by Corollary \ref{corollary:vertex-resilient-necessary}, each vertex-resilient component is contained in a $c(u,v)$ set. Thus, the $c(u,v)$ sets
define a ``coarse'' component forest, that our algorithm will refine later by using the loop nesting trees $H$ and $H^R$, as Theorem \ref{theorem:vertices-scc} indicates.
The fact that the $c(u,v)$ sets can be represented by a component forest of size $O(n)$
follows from the fact that these sets can be constructed by applying the \emph{split}
operation of  \cite{2VCB} to the sets $c(u) \cup \{u\}$, for each vertex $u$ that is not a leaf in $D$.
(See \cite{2VCB} for the details.)

In order to refine the initial component forest, we will use the following operation from \cite{2VCB}.
Let $\mathcal{B}$ be a set of components, let $\mathcal{S}$ be a partition of a set $U \subseteq V$, and let $x$ be a vertex not in $U$.

\begin{description}
	\item[\emph{refine}$(\mathcal{B}, \mathcal{S}, x)$:] For each component $B \in \mathcal{B}$, substitute $B$ by the sets  $B \cap (S \cup \{ x \} )$ of size at least two, for all $S \in \mathcal{S}$.
\end{description}

This operation can be executed in time that is linear in the total number of elements in all sets of $\mathcal{B}$ and $\mathcal{S}$:
\begin{lemma}\emph{(\cite{2VCB})}
\label{lemma:refine}
Let $M$ be the total number of elements in all sets of $\mathcal{B}$ ($M = \sum_{B \in \mathcal{B}}|B|$), and let $K$ be the number of elements in $U$. Then, the operation $\mathit{refine}(\mathcal{B}, \mathcal{S}, x)$ can be executed in $O(M+K)$ time.
\end{lemma}

The algorithm needs to locate the components that contain a specific vertex, and, conversely, the vertices that are contained in a specific component. To enable this, we store the adjacency lists of the current component forest $F$.
Note that $F$ is bipartite, so the adjacency list of a vertex $v$ stores the components that contain $v$, and the adjacency list of a component node $B$ stores the vertices in $B$.
Initially $F$ contains one component for each $c(u,v)$ set.
These components are later refined by executing a sequence of $\mathit{refine}$ operations.
The $\mathit{refine}$ operation maintains the invariant that $F$ is a forest~\cite{2VCB}, so it follows that the total number of vertices and edges in $F$ is $O(n)$
throughout the execution of the algorithm.
Moreover, when we execute a $\mathit{refine}$ operation we can update the adjacency lists of $F$, while maintaining the bounds given in Lemmas \ref{lemma:refine}.

Our new algorithm, dubbed \textsf{HDVRC}
since it uses the loop nesting trees $H$ and $H^R$, and the dominator trees $D$ and $D^R$, works as follows (see Algorithm \ref{fig:HDVRB} below).
After defining the initial components, in Lines 6--7, we perform a ``forward pass'' which processes $D$ and $H$.
During this pass, we visit the non-leaf vertices of $D$ in bottom-up order.
For each such vertex $u$,  Line 11 computes a partition $\mathcal{S}$ of $c(u)$, such that each set $S \in \mathcal{S}$ contains a subset of children of $u$ in $D$ that are strongly connected in $G \setminus u$.
Finally, we need to find the components that may contain vertices that are vertex-resilient with $u$. This is done in the loop of Lines 14--21.
After the completion of this forward pass, we execute a ``reverse pass'' that performs the analogous computations in $D^R$ and $H^R$.

\begin{algorithm}
	\LinesNumbered
	\DontPrintSemicolon
	\KwIn{Strongly connected digraph $G=(V,E)$}
	\KwOut{The vertex-resilient components of $G$}
	
    \textbf{Initialization:}\;

 	Compute the reverse digraph $G^R$.
Select an arbitrary start vertex $s \in V$.

	Compute the dominator trees $D$ and $D^R$ of the flow graphs $G_s$ and $G_s^R$, respectively.\;

	Compute the loop nesting trees $H$ and $H^R$ of the flow graphs $G_s$ and $G_s^R$, respectively.\;
		
     \textbf{Initialize component forest:}\;
     Compute the sets $c(u,v)$ for all pairs of vertices $u$ and $v$.\;
     Initialize the component forest $F$ to contain one component for each set $c(u,v)$ with at least two vertices.\;

	\textbf{Forward direction:}\;
    \ForEach{$u \in N\cup \{s\}$ in a bottom-up order of $D$}
	{
        Find the set of components $\mathcal{B}$ that contain at least two vertices in $c(u) \cup \{ u\}$\;
        Compute the collection of vertex subsets $\mathcal{S} = \{ H(v) \cap c(u) : h(v) \notin \widetilde{D}(u) \wedge v \in c(u)\}$\;
        Execute $\mathit{refine}(\mathcal{B}, \mathcal{S}, u)$\;
\If{$u \not= s$}
   {
        \ForEach{$B \in \mathcal{B}$  such that $u \in B$}
		{
           Choose an arbitrary vertex $v \not=u$ in $B$\;
           Compute the nearest common ancestor $w$ of $u$ and $v$ in $H$\;
           \If{$w \not\in c(d(u))$} {
                        Set $B = B \setminus u$\;
                        \lIf{$|B|=1$}{delete $B$ from $F$}
                }
		}
   }
 }
    \textbf{Reverse direction:}\;
	\ForEach{$u \in N^R \cup \{s\}$ in a bottom-up order of $D^R$}
	{
        Find the set of components $\mathcal{B}$ that contain at least two vertices in $c^R(u) \cup \{ u\}$\;
        Compute the collection of vertex subsets $\mathcal{S} = \{ H^R(v) \cap c^R(u) : h^R(v) \notin \widetilde{D}^R(u) \wedge v \in c^R(u)\}$\;
		Execute $\mathit{refine}(\mathcal{B}, \mathcal{S}, u)$\;
\If{$u \not= s$}
 {
        \ForEach{$B \in \mathcal{B}$ such that $u \in B$}
		{
            Choose an arbitrary vertex $v \not= u$  in $B$\;
            Compute the nearest common ancestor $w^R$ of $u$ and $v$ in $H^R$\;
             \If{$w^R \not\in c^R(d^R(u))$} {	
                       Set $B = B \setminus u$\;
                       \lIf{$|B|=1$}{delete $B$ from $F$}
               }
	}
 }
}
	\caption{\textsf{HDVRC}}
	\label{fig:HDVRB}
\end{algorithm}

Next, we prove the correctness of our new algorithm.

\begin{theorem}
Algorithm \textsf{HDVRC} is correct.
\end{theorem}
\begin{proof}
First we show that if two vertices $x$ and $y$ are vertex-resilient, then they are included in the same component in the resulting component forest $F$.
This is true after the initialization in Lines 6--7, as indicated by Corollary \ref{corollary:vertex-resilient-necessary}.
Now consider the two possible cases for the location of $x$ and $y$ in $D$: (a) $x$ and $y$ are siblings in $D$, or (b) $x=d(y)$.
We argue that $x$ and $y$ are not separated by the refine operation in Line 12.
Let $w=d(x)$. Since $x \vr y$, $x$ and $y$ are not separated by any vertex $u \in V \setminus \{x,y\}$. In particular,
$x$ and $y$ are strongly connected in $G \setminus u$ for any vertex $u$ in $D[s,w]$.
Then, Lemma~\ref{lemma:vertices-subtree} implies that there is a vertex $z$ such that $z \in \widetilde{D}(w)$ and $h(z) \notin \widetilde{D}(w)$, for which $x,y \in H(z)$.
Thus, in case (a),  the intersection operation in Line 11 will include $x$ and $y$ in the same set $S \in \mathcal{S}$, so $x$ and $y$ will remain in the same component after the refine operation in Line 12.
Now consider case (b). Let $S \in \mathcal{S}$ be the subset of children of $x$ that contains $y$, computed in Line 11.
After the refine operation in Line 12, a component $B=S \cup \{x\}$ will be created. The fact that $x, y \in H(z)$ implies $z \in c(w)$.
Therefore, the nearest common ancestor of $x$ and $y$ in $H$ is also a child of $w$, so $u$ is not removed from $B$ in Line 18.
By the symmetric arguments we have that $x$ and $y$ are also not separated into different components during the pass in the reverse direction.

We prove now that if $x$ and $y$ are not vertex-resilient then they will be located in different components at the end of the algorithm.
From Lemma \ref{vertices-separators-ancestors}, we have that $x$ and $y$ have a separating vertex $u$
that appears in at least one of the paths $D[s,x]$, $D[s,y]$, $D^R[s,x]$, and $D^R[s,y]$.
If $x$ and $y$ are neither siblings nor one is the parent of the other in $D$ or $D^R$, then by Corollary~\ref{corollary:vertex-resilient-necessary} they are not vertex-resilient. So $x$ and $y$ are
correctly placed in different components during the initialization in Lines 6--7.
Thus, we assume that $x$ and $y$
are either siblings of one is the parent of the other in both $D$ and $D^R$.
Consider first the case where $u$ is in $D[s,x] \cup D[s,y]$. Suppose, without loss of generality, that $y \not= d(x)$. Let $w=d(x)$.
Then $w$ is an ancestor of both $x$ and $y$ in $D$, so $u$ is in the path $D[s,w]$.
Then, Lemma~\ref{SAP-relation} implies that $w$ is a separating vertex for $x$ and $y$, so $x$ and $y$ lie in different strongly connected components in $G \setminus w$.
So, by Lemma~\ref{lemma:vertices-subtree}, $x$ and $y$ cannot be in the same subtree $H(z)$, where $z$ is a proper descendant of $w$ in $D$ such that $h(z) \notin \widetilde{D}(w)$.
Suppose $x$ and $y$ are siblings.
In this case, $x$ and $y$ will be in separate sets in the collection $\mathcal{S}$ computed in Line 11, so they will be placed in different components after the execution of the refine operation in Line 12.
Now suppose $x=d(y)$. Then, the nearest common ancestor of $y$ and $x$ in $H$ is not a proper descendant of $w$, hence $x$ will be removed from $B$ in Line 18.

Thus, we conclude that if $u$ is in $D[s,x] \cup D[s,y]$ then $x$ and $y$ are not contained in the same component. The case where $u$ is in $D^R[s,x] \cup D^R[s,y]$ is symmetric, so the theorem follows.
\end{proof}

Next we show that Algorithm \textsf{HDVRC} can be implemented in linear time. To do so, we will use the following lemma that will help us compute efficiently the collections $\mathcal{S}$ in Lines 11 and 27.

\begin{lemma}
Let $D$ and $H$ be the dominator tree and the loop nesting forest, respectively, of a flow graph $G_s$.
Suppose $x\neq s$ is a vertex such that $h(x)$ is not a sibling of $x$ in $D$.
Then no vertex in the path $H[s,h(x)]$ is a sibling of $x$ in $D$.
\label{lemma:siblings-loop-nesting-forest}
\end{lemma}
\begin{proof}
We show first that $d(x)$ is not a proper dominator of $h(x)$ in $G_s$.
Let $T$ be the dfs tree that generated $H$.
By the definition of the loop nesting forest, we have that $h(x)$ is a proper ancestor of $x$
in $T$. Also, we have that $d(x)$ is a proper ancestor of $x$ in $T$~\cite{domin:lt}.
If $h(x)$ is an ancestor of $d(x)$ in $T$, then clearly $d(x)$ is not a proper dominator of $h(x)$.
Now assume that $h(x)$ is in $T(d(x),x]$. By the definition of $h(x)$, there is a path $\pi$
in $G_s$ from $x$ to $h(x)$ such that, for all vertices $x'$ in $\pi$, $x' \in H(h(x))$.
Suppose, for contradiction, that
$h(x)$ is dominated by $d(x)$. Let $y$ be the child of $d(x)$ that is an ancestor of $h(x)$ in $D$.
Since $h(x)$ is not a sibling of $x$ in $D$, $y \not= h(x)$. Also, $y \not= x$ by the fact that
$h(x)$ is a proper ancestor of $x$ in $T$.
Then, Lemma \ref{lemma:paths-through-SAP} implies that $\pi$ contains $y$, so $y \in H(h(x))$.
But since $y$ is a proper dominator of $h(x)$, it is a proper ancestor of $h(x)$ in $T$, a contradiction.

Let $w$ be any vertex in the path $H[s,h(x))$. We show that $d(x)$ is not a proper dominator of $w$ in $G_s$, which implies the lemma.
This trivially holds if $d(x)$ is a descendant of $w$ in $T$. So suppose that $d(x)$ is not a descendant of $w$ in $T$.
Then $h(x) \not= d(x)$ because $h(x)$ is a descendant of $w$ and therefore a proper descendant of $d(x)$ in $T$.
Since $h(x) \in H(w)$, by the definition of the loop nesting forest we have that there is a path $\pi'$
from $h(x)$ to $w$ that contains only descendants of $w$ in $T$.  Then $d(x) \not\in \pi'$.
Since $d(x) \not= h(x)$, from the paragraph above we have that $d(x)$ is not a dominator of $h(x)$.
Then, there is a path $\pi''$ in $G_s$ from $s$ to $h(x)$ that avoids $d(x)$.
Hence $\pi'' \cdot \pi'$ is a path from $s$ to $w$ in $G_s$ that avoids $d(x)$, so $d(x)$ is not a dominator of $w$.
This means that $w$ and $x$ are not siblings in $D$.
\end{proof}

\begin{corollary}
Let $D$ and $H$ be the dominator tree and the loop nesting forest, respectively, of a flow graph $G_s$.
Suppose $x$ is a vertex in $c(u)$ such that $h(x) \not\in \widetilde{D}(u)$. Then the subgraph of $H$ induced by the vertices in
$H(x) \cap c(u)$ is a tree rooted at $x$.
\label{corollary:siblings-loop-nesting-forest}
\end{corollary}
\begin{proof}
Immediate from Lemma \ref{lemma:siblings-loop-nesting-forest}.
\end{proof}

\begin{theorem}
Algorithm \textsf{HDVRC} runs in $O(m+n)$ time, or in $O(n)$ time if the dominator trees $D$ and $D^R$, and the loop nesting trees $H$ and $H^R$ are given as input.
\label{theorem:HDVRB}
\end{theorem}
\begin{proof}
The initialization, including the computation of the $G^R$, takes $O(m+n)$ time.
The computation of the
dominator and loop nesting trees also takes $O(m+n)$ time by~\cite{dominators:bgkrtw}.
Given the dominator trees $D$ and $D^R$, we can compute the $c(u,v)$ sets in $O(n)$ time~\cite{LuigiGILP15}, and
at the same time initialize the component forest $F$.

We consider now the time required to locate the appropriate set of components in Lines 10 and 26.
In Line 10 (resp., Line 26), this is done by iterating over the adjacency list of each vertex $v\in c(u)$ (resp., $v\in c^R(u)$) in the component forest $F$, and marking the component nodes that are visited at least twice.
Now we bound the total number of vertices and nodes that we traverse in the component forest during this process.
During the \textbf{foreach} loop in the ``forward'' (resp., ``reverse'') direction, we traverse every adjacency list of a vertex $v$ at most twice; the first time when $u=d(v)$ (resp., $u=d^R(v)$) and the second time when $u=v$.
Also, for a given $u$, each component in the set $\mathcal{B}$ computed, in Line 10 (resp., Line 26), contains vertices in $c(u) \cup \{u\}$ (resp., $c^R(u) \cup \{u\}$).
So the total number of vertices in the components of $\mathcal{B}$ in Line 10 (resp., Line 26) is at most $|c(u)|+1$ (resp., $|c^R(u)|+1$).
Since, at any given time, $F$ contains at most $n-1$ components, it follows that each \textbf{foreach} loop (both ``forward'' and ``reverse'') runs in $O(n)$ time.

Now we need to show how to perform the operations in Lines 11 and 27 in $O(n)$ total time.
Corollary~\ref{corollary:siblings-loop-nesting-forest} suggests an efficient way to compute
$H(x) \cap c(u)$ for all children $x$ of $u$ such that $h(x) \not\in \widetilde{D}(u)$.
While we are processing $u$, we find its children $x$ that satisfy the condition $h(x) \not \in \widetilde{D}(u)$.
For each such child $x$, we do a preorder traversal $H(x)$ as follows. When we are at a vertex $x$ we visit only the children of $x$ in $H$ that
are siblings of $x$ in $D$. By Corollary~\ref{corollary:siblings-loop-nesting-forest}, a vertex is in $H(x) \cap c(u)$ if and only if it is visited during this traversal of $H(x)$.
Notice also that, in all these traversals, we visit every vertex and every edge of $H$ only once.
Thus, this procedure takes $O(n)$ time in total.
We apply the symmetric process for the computation of the sets $H^R(x) \cap c^R(u)$ in Line 27.

Now we account for the contribution of the $\mathrm{refine}$ operations.
As we already mentioned, for a given $u$, the total number of vertices in the components computed in Line 10
(resp., Line 26) is at most $|c(u)|+1$ (resp., $|c^R(u)|+1$).
Also, since in Line 11 (resp., 27) set $\mathcal{S}$ is a partition of $c(u)$ (resp., $c^R(u)$), the total number of the vertices that appear in this partition (which is set $U$ in Lemma~\ref{lemma:refine}) is $|c(u)|$ (resp., $|c^R(u)|$). Thus, by Lemma~\ref{lemma:refine}, each refine operation in Lines 11 and 27 takes time $O(|c(u)|)$ and $O(|c^R(u)|)$, respectively, which is $O(n)$ in total.

Finally, since the component forest contains at most $n-1$ components at any given time, the loops in Lines 14--21 and 30--37 execute at most $2(n-1)$ nearest common ancestor (nca) calculations. So we can compute all the required nca's in $O(n)$ time by \cite{nca:ht}. In fact, in both the ``forward'' and the ``reverse'' direction loop, we can precompute the nca's for $n-1$ pairs of vertices, where each pair consists of a vertex $v \not= s$ and its parent $u=d(v)$ in $D$. This way we can use an \emph{offline} nca algorithm~\cite{dominators:bgkrtw}.
The bound follows.
\end{proof}

We remark that we can extend Algorithm \textsf{HDVRC} in the same way as for \textsf{HD2ECC} (see Section \ref{sec:2ECBs}),
so that we obtain a sparse certificate for the $2$-vertex-connected components, i.e., a subgraph of the input digraph $G$ that has $O(n)$ edges and has the same $2$-vertex-connected components as $G$.
See Section \ref{sec:sparse-certificate}.

\section{Pairwise $2$-vertex connectivity queries}
\label{sec:vertices-other}

Let $G=(V,E)$ be a strongly connected graph, let $x$ and $y$ be any two vertices of $G$, and let $u$ be a strong articulation point of $G$. Recall that we say that $u$ is a \emph{separating vertex for $x$ and $y$} (or equivalently that $u$ \emph{separates $x$ and $y$}) if $x$ and $y$ are not strongly connected in $G\setminus u$.
In this section, we extend our approach in order to answer in asymptotically optimal time the following types of queries:
\begin{itemize}
\item[(a)] Test if two query vertices $x$ and $y$ are $2$-vertex-connected; if not, report a separating vertex or a separating edge for $x$ and $y$.
\item[(b)] Test whether a given vertex $u$ separates two query vertices $x$ and $y$.
\item[(c)] Report all the separating vertices for a given pair of vertices $x$ and $y$.
\end{itemize}

We note that the data structure of \cite{2VCB} only supports queries of type (a), which it answers in constant time by applying Lemma \ref{cor:2vc-resilient}.
We can use the same approach here and test if $x \leftrightarrow_{\mathrm{2e}} y$ and $x \leftrightarrow_{\mathrm{vr}} y$ separately.
The former test can be done by the algorithm of Section \ref{sec:other}, which also reports a separating edge in case $x$ and $y$ are not $2$-edge-connected.
It remains to test if  $x \leftrightarrow_{\mathrm{vr}} y$, and provide a separating vertex for $x$ and $y$ if the outcome of this test is negative.
As in \cite{2VCB} we can use the component forest $F$, computed in $O(m)$ time by Algorithm \textsf{HDVRC} of Section \ref{sec:2VCBs}, in order to test in constant time if two query vertices $x$ and $y$ are
vertex-resilient. To do this, we view $F$ as a forest of rooted trees by choosing an arbitrary vertex as the root of each tree.
Then $x \leftrightarrow_{\mathrm{vr}} y$ if and only if $x$ and $y$ are siblings or one is the grandparent of the other.
So we can perform this test in constant time simply by storing the parent of each vertex in $F$.
If $x$ and $y$ are not vertex-resilient then we can provide in constant time a separating vertex for $x$ and $y$ by adapting the query of type (c) so that it terminates as soon as it finds the first separating vertex for $x$ and $y$.

Next we deal with queries of type (b) and (c).

\begin{lemma}
	\label{lemma:vertices-separating_two_vertices}
Let $x$ and $y$ be two vertices, and let $w$ and $w^R$ be their nearest common ancestors in $H$ and $H^R$, respectively.
Then, a strong articulation point $u$ is a separating vertex for $x$ and $y$ if and only if one of the following holds:
	\begin{itemize}
		\item[(1)] $u$ is an ancestor of $x$ or $y$ in $D$ and $w$ is not a proper descendant of $u$ in $D$.
		\item[(2)] $u$ is an ancestor of $x$ or $y$ in $D^R$ and $w^R$ is not a  proper descendant of $u$ in $D^R$.
	\end{itemize}
\end{lemma}
\begin{proof}
We will only prove the first case, since the second is completely symmetric.
By Lemma~\ref{vertices-separators-ancestors}, a strong articulation point $u$ that separates $x$ and $y$ lies in at least one of the paths $D[s,x]$, $D[s,y]$, $D^R[s,x]$, or $D^R[s,y]$.
So, for the proof of (1), we assume that $u$ is in $D[s,x]$ or $D[s,y]$. Now we argue that $u$ separates $x$ and $y$ if and only if $w$ is not a proper descendant of $u$ in $D$. Let $z$ be a descendant of $u$ in $D$.
From Lemma~\ref{lemma:vertices-subtree}, we have that the vertices of $H(z)$ are strongly connected in $G\setminus u$. Hence, if $w$ is a descendant of $u$ in $D$, then $x$ and $y$ are strongly connected in $G\setminus u$, since $x, y \in H(w)$.

Now we prove the opposite direction.
Assume, for contradiction, that $u$ is an ancestor of $x$ and $y$ in $D$,
$w$ is not a proper descendant of $u$ but $u$ is not a separating vertex for $x$ and $y$. Let $C$ be the strongly connected component of $G \setminus u$ that contains $x$ and $y$.
From Theorem \ref{theorem:vertices-scc} and the assumption that
$u$ is an ancestor of $x$ and $y$ in $D$ but not a separating vertex for $x$ and $y$, we have that $C \subseteq \widetilde{D}(u)$.
By Lemma \ref{lemma:vertices-subtree2},  we have that there is a vertex $z \in C$ such that $C = H(z)$.
But then $x, y \in H(z)$, so $z$ is ancestor of $w$ in $H$ and a proper descendant of $u$ in $D$.
This contradicts the assumption that $w \not\in \widetilde{D}(u)$.
\end{proof}

Note that, by Lemma \ref{lemma:paths-through-SAP}, an ancestor $u$ of $x$ or $y$ in $D$ (resp., $D^R$) is not a separating vertex for $x$ and $y$ only if it is a common ancestor of $x$ and $y$ in $D$ (resp., $D^R$). So, the application of Lemma~\ref{lemma:vertices-separating_two_vertices} gives the following algorithm for computing, in an online fashion, all the separating vertices for a given pair of query vertices $x$ and $y$. To that end, we can preprocess $H$ and $H^R$ in $O(n)$ time so we can compute nearest common ancestors in constant time~\cite{nca:ht}.
To answer a reporting query for vertices $x$ and $y$, we compute their nearest common ancestor $w$ in $H$, and visit the ancestors of $x$ and $y$ in $D$ in a bottom-up order, as follows.
Starting from $x$ and $y$, we visit the ancestors of $x$ and $y$ in $D$ until we reach a vertex $u$ that is a proper ancestor of $w$ in $D$, or until we reach $s$ if no such $u$ exists.
(As we showed above, if there is an ancestor $u$ of $x$ or $y$ in $D$ that is not a separating vertex for $x$ and $y$ then $u$ is a common ancestor of $x$ and $y$ in $D$.)
If $u$ exists then the vertices in $D(u,d(x)] \cup D(u,d(y)]$ are separating vertices for $x$ and $y$.
(We let $D(w,z]$ be empty if $w=z$.)
Otherwise
the vertices in $D[s,d(x)] \cup D[s,d(y)]$ are separating vertices for $x$ and $y$.
We do the same for the reverse direction, i.e.,
we compute the nearest common ancestor $w^R$ of $x$ and $y$ in $H^R$
and visit the ancestors of $x$ and $y$ in a bottom-up order in $D^R$,
until we reach a vertex $u$ that is a proper ancestor of $w^R$ in $D^R$, or until we reach $s$ if no such $u$ exists.
We finally return the union of the separating vertices that we found in both directions.

With the same data structure we can test in constant time, for any pair of query vertices $x$ and $y$ and a query vertex $u$, if $u$ is a separating vertex for $x$ and $y$. By Lemma~\ref{lemma:vertices-separating_two_vertices}, all we have to do is to find their nearest common ancestors, $w$ in $H$ and $w^R$ in $H^R$, and then test
if $u$ is an ancestor of $x$ or $y$ and $w$ is not a proper descendant of $u$ in $D$, or if $u$ is an ancestor of $x$ or $y$ and $w^R$ not a proper descendant of $u$ in $D^R$.

\begin{theorem}
Let $G$ be a strongly connected digraph with $m$ edges and $n$ vertices. After $O(m)$-time preprocessing, we can build
 an $O(n)$-space data structure, that can:
\begin{itemize}
\item Test if two query vertices $x$ and $y$ are $2$-vertex-connected and if not report a corresponding separating vertex or a corresponding separating edge.
\item Report all vertices that separate two query vertices $x$ and $y$ in $O(k)$ time, where $k$ is the total number of separating vertices reported.
 For $k=0$, the time is $O(1)$.
\item Test in constant time if a query vertex is a separating vertex for a pair of query vertices.
\end{itemize}
\end{theorem} 
\section{Sparse certificate}
\label{sec:sparse-certificate}

Our framework also provides a \emph{sparse certificate} ${\cal C}(G)$ for the $2$-connectivity relations of the input digraph $G$.
That is, ${\cal C}(G)$ is a strongly connected spanning subgraph of $G$ with $O(n)$ edges, such that it maintains:
(i) the $1$-connectivity cuts of $G$ and the decompositions induced by those cuts,
and (ii) the $2$-edge-connected and $2$-vertex-connected components of $G$.
Hence, for any edge $e$ (resp., vertex $v$) of $G$, the strongly connected components of
${\cal C}(G) \setminus e$ (resp., ${\cal C}(G) \setminus v$) are identical to the strongly connected components of
$G \setminus e$ (resp., $G \setminus v$).

The construction of this sparse certificate is provided by Algorithm \textsf{SparseCertificate} below.
It uses the notion of \emph{divergent spanning trees}~\cite{DomCert:TALG} of a flow graph $G_s$, together with edges that define
a loop nesting tree of $G_s$.
Here, a spanning tree $T$ of a flow graph $G_s$ is a tree with root $s$ that contains a path from $s$ to $v$ for all vertices $v$. Two spanning trees $T_1$ and $T_2$ rooted at $s$ are \emph{divergent} if for all vertices $v$, the paths from $s$ to $v$ in $T_1$ and $T_2$ share only the dominators of $v$. Every flow graph $G_s$ has two such spanning trees, and they can be computed in linear time \cite{DomCert:TALG}. Moreover, the computed spanning trees are \emph{maximally edge-disjoint}, i.e., the only edges they have in common are the bridges of $G_s$.

\begin{algorithm}
\LinesNumbered
\DontPrintSemicolon
 \KwIn{Strongly connected digraph $G=(V,E)$}
 \KwOut{A sparse strongly connected spanning subgraph ${\cal C}(G)$ of $G$ that maintains the same $1$-connectivity cuts and the same decompositions induced by
the $1$-connectivity cuts of $G$, and the $2$-edge and $2$-vertex-connected components of $G$.
}

 \textbf{Initialization:}\;
 	Compute the reverse digraph $G^R$. Select an arbitrary start vertex $s \in V$.  \;

 \textbf{Process flow graph $G_s$:}\;
	Compute two divergent spanning trees of $G_s$. Let $E_1$ be the set of edges of those spanning trees.\;
	Compute a loop nesting tree $H$ of $G_s$. Let $E_2$ be edges of $G_s$ that define $H$.\;

 \textbf{Process reverse flow graph $G^R_s$:}\;
	Compute two divergent spanning trees of $G^R_s$. Let $E_3$ be the set of edges of those spanning trees.\;
    Compute a loop nesting tree $H$ of $G_s$. Let $E_4$ be edges of $G^R_s$ that define $H^R$.\;

 \textbf{Postprocessing:}\;
	Reverse the direction of the edges in $E_3$ and $E_4$.\;
 	Insert into ${\cal C}(G)$ the edges of the set $E_1 \cup E_2 \cup E_3 \cup E_4$.\;

 \caption{\textsf{SparseCertificate}}
 \label{algorithm:SparseCertificate}
\end{algorithm}

Now we prove that Algorithm \textsf{SparseCertificate} is correct. Since ${\cal C}(G)$ contains two divergent spanning trees of $G_s$, it follows from \cite{DomCert:TALG} that flow graph ${\cal C}_s(G)$ has the same dominator tree as $G_s$.
Also, in order to maintain the same loop nesting tree $H$ of $G_s$
we include in ${\cal C}(G)$ the dfs trees of $G_s$ that defines $H$, together with at most $n-1$
additional edges that generate the loops of $G_s$.
These can be found while doing, for each vertex $x$, a backward search from $x$ in $G_s$
that
visits only the children of $x$ in $H$, as explained in~\cite{dominators:Fraczak2013,st:t}.
We perform the corresponding steps in $G^R_s$, so ${\cal C}_s^R(G)$ has the same dominator tree and
the same loop nesting tree as $G^R_s$.
Since the four trees $D$ and $D^R$, and $H$ and $H^R$ are the same for both
$G_s$ and ${\cal C}_s(G)$, each of our algorithms computes on input ${\cal C}(G)$
the same output as for input $G$. Moreover, it is easy to see
that ${\cal C}(G)$ contains $O(n)$ edges of $G$.
This proves the correctness of Algorithm \textsf{SparseCertificate}.
Regarding its running time, by \cite{dominators:bgkrtw,DomCert:TALG} we have that
all steps of the algorithm take linear time.
Hence, we obtain the following result.

\begin{theorem}
\label{theorem:sparse-certificate}
Given a strongly connected directed graph $G$ with $m$ edges and $n$ vertices, we can compute in $O(m+n)$ time a sparse certificate that maintains the same $1$-connectivity cuts and the same decompositions induced by the $1$-connectivity cuts of $G$, together with the $2$-edge and $2$-vertex-connected components of $G$.
\end{theorem}

We also remark that the sparse certificates of \cite{2ECC:GILP:TALG} and \cite{2VCB} maintain, respectively, only the $2$-edge and the $2$-vertex-connected components.
Such sparse certificates, including our new construction, can be used to obtain fast approximation algorithms for computing sparse $2$-connectivity preserving subgraphs~\cite{GIKPP:TCS,2vcb:jaberi15}.
Our new sparse certificate also provides an alternative way to achieve the bounds stated in Theorems \ref{theorem:all-scc} and \ref{theorem:all-scc-v}.
That is, we can use ${\cal C}(G)$ as our data structure, so that given a query edge $e$ (resp., vertex $v$), we can report the strongly connected components of
${\cal C}(G) \setminus e$ (resp., ${\cal C}(G) \setminus v$) instead of the strongly connected components of
$G \setminus e$ (resp., $G \setminus v$); the properties of our sparse certificate guarantee that the reported output
is correct.
Also, since ${\cal C}(G)$ has $O(n)$ edges, these computations take $O(n)$ time.

\section{Conclusions}
\label{sec:conclusions}

In this paper, we have investigated some basic problems related to the strong connectivity and $2$-connectivity of directed graphs, by considering the effect of edge and vertex deletions on their strongly connected components. Let $G$ be a directed graph with $m$ edges and $n$ vertices.
We have presented a collection of
$O(n)$-space data structures that, after $O(m+n)$-time preprocessing, can solve efficiently several problems, including
reporting in $O(n)$ worst-case time all the strongly connected components obtained after deleting a single edge (resp., a single vertex) in $G$;
computing the total number of strongly connected components obtained after deleting a single edge (resp., a single vertex) in total worst-case $O(n)$ time for all edges (resp., for all vertices);
computing the size of the largest and of the smallest strongly connected components obtained after deleting a single edge (resp., a single vertex) in total worst-case $O(n)$ time for all edges (resp., for all vertices).
After $O(m+n)$-time preprocessing, we can also build an $O(n)$-space data structure that can
 answer efficiently basic $2$-edge and $2$-vertex connectivity queries on directed graphs. All our bounds are asymptotically tight.

Our work raises some new
and perhaps intriguing questions.
First, using the algorithmic framework developed in this paper, we can find in linear time
an edge or a vertex whose removal minimizes/maximizes several properties of the resulting strongly connected components, such as their number, or their largest or their smallest size.
Can our approach be used to find in linear time an edge (resp., a vertex) whose removal optimizes some more complex properties of the resulting strongly connected components?
In particular,
can we find an edge $e$ (resp., a vertex $v$) that minimizes/maximizes a given function $f(|C_1|,|C_2|,...,|C_k|)$, where $C_1$, $C_2$, $\ldots$, $C_k$ are the strongly connected components of $G\setminus e$ (resp.,  $G\setminus v$), and $f(x_1,x_2,\ldots,x_k)$ can be computed in time $O(k)$?
We showed how to achieve an $O(m+n)$ time bound for some special cases of functions in Theorems \ref{theorem:generic} and \ref{theorem:generic-v}.
For general functions $f(x_1,x_2,\ldots,x_k)$, the algorithms presented in this paper achieve $O(m+nb)$ (resp., $O(m+np)$) worst-case time to solve this problem, where $b$ (resp., $p$) is the number of strong bridges (resp., strong articulation points) in the input digraph.
Can we solve this problem in linear time in case of general functions?
Second, can the algorithms given in this paper be used for designing efficient approximation algorithms for other hard optimization problems, such as critical edge (resp.,  vertex) detection problems? In those problems, one is interested in removing a set of $k$ edges (resp., vertices) so as to minimize the
pairwise connectivity of the resulting graph.
Third, the dynamic maintenance of $2$-connectivity properties in directed graphs deserves further attention. After this work, we have been able to apply the algorithmic framework developed in this paper to the incremental maintenance of the $2$-edge- and the $2$-vertex-connected
components of directed graphs \cite{GIN16:ICALP,GIN18:LATIN}. The decremental version of these problems, where we wish to maintain the $2$-edge- and the $2$-vertex-connected components of a directed graph under edge deletions, were considered in \cite{decdom17} using different techniques. The decremental algorithms in \cite{decdom17} achieve $O(mn \log{n})$ total running time, using $O(n^2 \log{n})$ space. Can these bounds be improved?
Finally, we note that our approach is also able to provide alternative linear-time algorithms for computing
 the $2$-edge-connected and $2$-vertex-connected components of a digraph, which appear to be simpler than previous algorithms, and therefore likely to perform better in practice. We refer to \cite{2CC:ALENEX18} for an experimental evaluation of such algorithms.
%

\section*{Acknowledgments}
We are indebted to Peter Widmayer for suggesting the problem of finding the edge whose deletion yields the smallest largest strongly connected components in a strongly connected digraph, and to
Mat\'u\v{s} Mihal\'ak,
Przemys\l aw Uzna\'nski and Pencho Yordanov for useful discussions on the practical applications of this problem and for pointing out references~\cite{Gunawardena12,MihalakUY15}.

\bibliographystyle{plain}
\bibliography{ltg}

\end{document}